%% file: final.tex
\documentclass[11pt]{article}

\usepackage{amssymb}
\usepackage{amsmath,amsthm}
\usepackage{graphicx}
\usepackage[english]{babel}
\usepackage{subfigure}
\usepackage{color,wrapfig,sidecap}
\usepackage{enumitem}
\usepackage{hyperref}
\usepackage{pgf,tikz}
\usetikzlibrary{arrows}
\usetikzlibrary{snakes}
\usepackage{geometry}
\usepackage{fullpage}

\usepackage{url}

\newtheorem{theorem}{Theorem}
\newtheorem{lemma}{Lemma}
\newtheorem{define}{Definition}
\newtheorem{observation}{Observation}
\newcommand{\pw}{\mathrm{pw}}
\newcommand{\bw}{\mathrm{bw}}
\newcommand{\smallO}[1]{\mathop{}\!o{\left(#1\right)}}
\newcommand{\tw}{\mathrm{tw}}
\newcommand{\univ}{\mathbb{U}}
\newcommand{\ds}{DS}
\newcommand{\rds}{$r$DS}
\newcommand{\cds}{CDS}
\newcommand{\crds}{$r$CDS}

\begin{document}


\title{Optimal dynamic program for r-domination problems over tree decompositions}


%
%
\author{Glencora Borradaile \and Hung Le}
%


%
%

\maketitle
\thispagestyle{empty}

\begin{abstract} 
  There has been recent progress in showing that the exponential dependence on treewidth in dynamic programming algorithms for solving NP-hard problems are optimal under the Strong Exponential Time Hypothesis (SETH).  We extend this work to $r$-domination problems.  In $r$-dominating set, one wished to find a minimum subset $S$ of vertices such that every vertex of $G$ is within $r$ hops of some vertex in $S$.  In connected $r$-dominating set, one additionally requires that the set induces a connected subgraph of $G$.  We give a $O((2r+1)^\tw n)$ time algorithm for $r$-dominating set and a $O((2r+2)^\tw n^{O(1)})$ time algorithm for connected $r$-dominating set in $n$-vertex graphs of treewidth $\tw$.   We show that the running time dependence on $r$ and $\tw$ is the best possible under SETH.  This adds to earlier observations that a ``+1'' in the denominator is required for connectivity constraints.
\end{abstract}

\newpage
\setcounter{page}{1}

\section{Introduction}

There has been recent progress in showing that the exponential dependence on treewidth\footnote{Treewidth is defined formally in the appendix.} in dynamic programming algorithms for solving NP-hard problems are optimal under the Strong Exponential Time Hypothesis (SETH)~\cite{IP01}.   Lokshtanov, Marx and Saurabh showed that for a wide variety of problems with local constraints, such as maximum independent set, minimum dominating setand $q$-coloring, require $\Omega^*((2-\epsilon)^\tw)$, $\Omega^*((3-\epsilon)^\tw)$ and $\Omega^*((q-\epsilon)^\tw)$ time in graphs of treewidth $\tw$, where $\Omega^*$ hides polynomial dependence on the size of the graph~\cite{LMS11}; these lower bounds met the best-known upper bounds for the same problems.  For problems with connectivity constraints, such as connected dominating set, some thought that a dependence of $\tw^\tw$ would be required.  Cygan et al.\ showed that this is not the case, giving tight upper and lower bounds for the dependence on treewidth for many problems, including connected dominating set~\cite{CNPPRW11}.  They also observed that the base of the constant increased by one when adding a connectivity constraint.  For example, vertex cover has tight upper and lower bounds of $O^*(2^\tw)$ and $\Omega^*((2-\epsilon)^\tw)$ while connected vertex cover has tight upper and lower bounds of $O^*(3^\tw)$ and $\Omega^*((3-\epsilon)^\tw)$.  Similarly, dominating set has tight upper and lower bounds of $O^*(3^\tw)$ and $\Omega^*((3-\epsilon)^\tw)$ while connected dominating set has tight upper and lower bounds of $O^*(4^\tw)$ and $\Omega^*((4-\epsilon)^\tw)$.

\subsection{Generalization to $r$-domination}

In this paper, we show that this pattern of dependence extends to domination problems over greater distances.  
The $r$-dominating set (\rds) problem is a natural extension of the dominating set (\ds) problem, in which, given a graph $G$, the goal is to find a minimum subset $S$ of vertices such that every vertex of $G$ is within $r$ hops of some vertex in $S$.  Likewise, the connected $r$-dominating set (\crds) is the connected generalization of connected dominating set (\cds).  We show that \rds\ can be solved in $O^*((2r+1)^\tw)$ time and that \crds\ can be solved in $O^*((2r+2)^\tw)$ time.  Further, we show that these upper bounds are tight, assuming SETH; in particular, we assume that $n_0$-variable SAT cannot be solved in $O^*(2^{\delta n_0})$ time for any $\delta < 1$~\cite{IP01}. 

\begin{center}
  \begin{tabular}{ r | c| c | c| c }
     & Lower Bound & Reference & Upper Bound & Reference\\ \hline
    \ds & $\Omega^*((3-\epsilon)^{\tw})$ &  \cite{LMS11} & $O^*(3^{\tw})$ & \cite{AN02}  \\ 
    \cds & $\Omega^*((4-\epsilon)^{\tw})$ &  \cite{CNPPRW11} & $O^*(4^{\tw})$ & \cite{CNPPRW11}  \\ 
    \rds & $\Omega^*((2-\epsilon)r+1)^{\tw})$ & Theorem~\ref{thm:dp-thm-ds} & $O^*(\left( 2r + 1\right)^{\tw})$ & Theorem~\ref{thm:lb-ds} \\ 
    \crds & $\Omega^*((2-\epsilon)r+2)^{\tw})$ & Theorem~\ref{thm:dp-thm-cds} & $O^*(\left( 2r + 2\right)^{\tw})$ & Theorem~\ref{thm:lb-cds} 
  \end{tabular}
\end{center}

\paragraph{Upper and lower bounds for $r$-dominating set}

The algorithm we give is a relatively straightforward generalization of the $O(3^\tw \tw^2 n)$-time algorithm for \ds\ given by Rooij, Bodlaender and Rossmanith, but for completeness, we provide the proof of the following in Appendix~\ref{App:AppendixA}.  
\begin{theorem} \label{thm:dp-thm-ds}There is an $O((2r+1)^{\tw+1}n)$-time algorithm for \rds\ in graphs $G$ of treewidth  $\tw$.
\end{theorem}
\noindent Demaine et al.\ gave an algorithm with running time $O((2r+1)^{\frac{3}{2}\bw}n)$ for \rds\ in graphs of branchwidth $\bw$; since branchwidth and treewidth are closely related by the inequality $\bw \leq \tw+1$ (for which there are tight examples)~\cite{RS91}, our algorithm improves the exponential dependence.
Our proof of the corresponding lower bound uses a high level construction similar to that of Lokshtanov, Marx and Saurabh for \ds~\cite{LMS11}, but the gadgets we require are non-trivial generalizations.  We prove the following in Section~\ref{sec:lb-ds}.
\begin{theorem}\label{thm:lb-ds} 
For every $\epsilon < 1$ and for $r = n^{\smallO 1}$ such that $r$-dominating set can be solved in $((2-\epsilon)r+1)^{\pw}n^{O(1)}$ time in a graph with pathwidth\footnote{A path decomposition is a tree decomposition whose underlying structure is a path.} $\pw$ and $n$ vertices, there is a $\delta < 1$ such that SAT can be solved in $O^*(2^{\delta n_0})$.
\end{theorem}
\noindent We point out that for sufficiently large $r$, there is a sufficiently small $r$-dominating set such that one can find it by enumeration more quickly than suggested by Theorem~\ref{thm:dp-thm-ds}.  At the extreme, if the diameter of the graph is at most $2r$, then there is an $r$-dominating set of size at most $\tw+1$ that is, additionally, contained by one bag of the decomposition~\cite{GPRS01,BC14}.  The optimal solution in this case could be smaller than $\tw +1$. We pose, as an open problem, improving the dependence on running time for the larger values of $r$, $r = n^c$ for some $c$, for which Theorem~\ref{thm:lb-ds} does not hold.

\paragraph{Upper and lower bounds for connected $r$-dominating set} As with the algorithms for connectivity problems with singly-exponential time dependence on treewidth as introduced by Cygan et al.~\cite{CNPPRW11}, our algorithm for \crds\ is a randomized Monte-Carlo algorithm.  As for \rds, this upper bound is relatively straightforward, but we include the details in Appendix~\ref{App:AppendixC}:
\begin{theorem} \label{thm:dp-thm-cds}
There is a $O^*((2r+2)^{\tw+1})$-time true-biased Monte-Carlo algorithm that decides \crds\ for graphs of treewidth $\tw$.
\end{theorem}
\noindent Our corresponding lower bound uses a new gadget construction from that of the lower bound of Cygan et al.  The following theorem is proved in Section \ref{sec:lb-cds}.
\begin{theorem}\label{thm:lb-cds} 
For every $\epsilon < 1$ and for $r = n^{\smallO 1}$ such that connected $r$-dominating set can be solved in $((2-\epsilon)r+2)^{\pw}n^{O(1)}$ time in a graph with pathwidth $\pw$ and $n$ vertices, there is a $\delta < 1$ such that SAT can be solved in $O^*(2^{\delta n_0})$.
\end{theorem}

\subsection{Motivation}

Algorithms for such problems in graphs of bounded treewidth are useful as subroutines in many approximation algorithms for graphs having bounded local treewidth~\cite{Eppstein00}; specifically, polynomial-time approximation schemes (PTASes) for many problems, including dominating set, TSP and Steiner tree, in planar graphs and graphs of bounded genus all reduce to the same problem in a graph of bounded treewidth whose width depends on the desired precision~\cite{BDT12,BKM09,Klein08,Baker94}.  For sufficiently small $r$, Baker's technique and Demaine and Hajiaghayi's bidimensionality framework imply PTASes for \rds\ and \crds\ (respectively) \cite{Baker94,DHM11}.  For larger values of $r$, {\em approximate} $r$-domination can be achieved by the recent bi-criteria PTAS due to Eisenstat, Klein and Mathieu~\cite{EKM14}; they guarantee a $(1+\epsilon)r$-dominating set of size at most $1+\epsilon$ times the optimal $r$-dominating set.  It is an interesting open question of whether a true PTAS (without approximating the domination distance) can be achieved for \rds\ in planar graphs for arbitrary values of $r$.  We also note that the bi-criteria PTAS of Eisenstat et al.\ is not an {\em efficient} PTAS one which the degree of the polynomial in $n$ (the size of the graph) does not depend on the desired precision, $\epsilon$.  Our new lower bounds suggest that, for large $r$, it may not be possible to design an efficient PTAS for \rds without also approximating the domination distance, since the $O^*(r^\tw)$ run-time of the dynamic program becomes an $O^*(r^{1/\epsilon})$ run time for the PTAS.

\section{Lower Bound for $r$-dominating Set} \label{sec:lb-ds}

In this section we prove Theorem~\ref{thm:lb-ds}: for every $\epsilon < 1$  such that $r$-dominating set can be solved in $O^*(((2-\epsilon)r+1)^{\pw})$ in a graph with pathwidth $\pw$, there is a $\delta < 1$ such that SAT can be solved in $O^*(2^{\delta n_0})$.

We give a reduction from an instance of SAT to an instance of $r$-dominating set in a graph of pathwidth 
\begin{equation}
\pw \leq \frac{n_0p}{\lfloor p\log(2r+1) \rfloor} + O\left(rp\right)\mbox{ for any integer $p$.}\label{eq:pw}
\end{equation}
Therefore, an $O^*(((2-\epsilon)r+1)^{\pw})$-time algorithm for $r-$dominating set would imply an algorithm for SAT of time $O^*\left(\left(2(r-\epsilon)+1\right)^\pw\right) = O^*\left(2^{\pw\log{(2(r-\epsilon)+1)}}\right)$.

We argue that for sufficiently large $p$ depending only $\epsilon$, there is a $\delta$ such that $\pw\log{(2(r-\epsilon)+1)} = \delta n_0$, completing the reduction.  By Equation~(\ref{eq:pw}), 
\begin{equation*}
  \pw\log{(2(r-\epsilon)+1)} = n_0\left(\frac{p\log{(2(r-\epsilon)+1)}}{\lfloor p\log(2r+1) \rfloor} + {\frac{O(rp)}{n_0\lfloor p\log(2r+1) \rfloor}}\right)
\end{equation*}
The second term in the bracketed expression is $\smallO 1$ for large $n_0$; we show that the first term in the bracketed expression $=\delta < 1$ for sufficiently large $p$:
\[
\frac{p\log{(2(r-\epsilon)+1)}}{\lfloor p\log(2r+1) \rfloor}\le \frac{p\log(2r+1-t_\epsilon)}{\lfloor p\log(2r+1) \rfloor} = \frac{\lfloor p\log(2r+1) \rfloor + s_p-pt_\epsilon}{\lfloor p\log(2r+1) \rfloor}= 1-\frac{pt_\epsilon - s_p}{\lfloor p\log(2r+1) \rfloor}
\]
In the above, $t_\epsilon$ is a constant depending only on $\epsilon$ and $s_p < 1$ depends only on $p$.  Therefore, choosing $p$ sufficiently large makes this expression sufficiently smaller than 1.
Note that this hardness result holds for $r = n^{\smallO 1}$ since our construction results in a graph with $O(r^{p+3})$ vertices.

Given an integer $p$, we assume, without loss of generality, that $n$ is a multiple of $\lfloor
p\log(2r+1) \rfloor$. Divide the $n_0$ variables of the SAT formula into
$t$ groups, $F_1,....,F_t$, each of size $\lfloor p\log(2r+1)
\rfloor$; $t = \frac{n_0}{\lfloor p\log(2r+1) \rfloor}$. We assume that $r \geq 2$.

\paragraph{An r-frame}An \emph{$r-$frame} is a graph obtained from a grid of size $r \times r$, adding edges along the diagonal, removing vertices on one side of the diagonal, subdividing edges of the diagonal path and connecting the subdividing vertices to the vetices of the adjacent triangles (refer to Fig.~\ref{fig:frame}). The vertex $A$ is the \emph{top} of the $r-$frame and the path $BC$ is the \emph{bottom path} of the $r-$frame. An \emph{$r-$frame avoiding a vertex $p$} of the bottom path is the graph obtained by deleting edges not in the bottom path and incident $p$. We define \emph{identification} to be the operation of identifying the bottom paths of one or more $r-$frames with a path of length $2r+1$ (refer to Fig.~\ref{fig:frame}).
 
 \begin{figure}[h] \label{fig:frame}
  \centering
  \vspace{-20pt}
  \input{figs/s-frame-path.tex}
  \vspace{-15pt}
  \caption{(a) An $r$-frame (r = 3). (b) An $r-$frame avoiding $p$ (dashed edges to be deleted). (c) Two $r$-frames and a path of length $2r+1$ (r=3).  (d) Graph obtained from identifying the two $r$-frames and path in (c).}
\vspace{-10pt}
\end{figure}
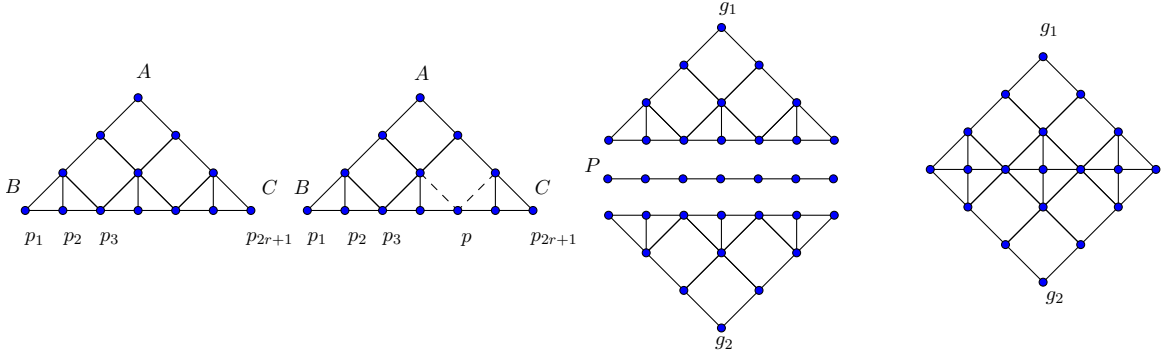


\begin{figure}[h] 
  \centering
  \vspace{-20pt}
  \input{figs/s-group-gadgets.tex}
  \vspace{-15pt}
  \caption{The group gadget. The red edges are edges of $r-$frames with top $x_S$} \label{fig:group-gadgets}
\vspace{-10pt}
\end{figure}
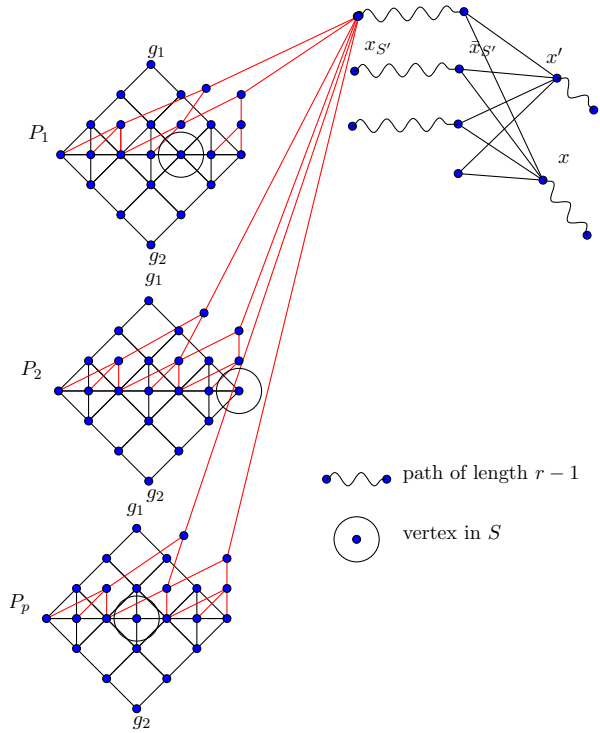

\paragraph{The group gadget} We construct a gadget to represent each
group of variables as follows (refer to Fig.~\ref{fig:group-gadgets}). Let ${\cal P} = \{P_1,P_2,...,P_p\}$ be a set of $p$ paths, each of length $2r + 1$.  For each path $P_i$,
we construct a graph $C_i$ from $P_i$ by identifying two $r-$frames with top vertices $g_1$ and $g_2$, called {\em guards}, to $P_i$. In the remaining steps of the construction, we will only connect to the vertices of $\cal P$.  In order to $r$-dominate the guards, we see:

\begin{observation} \label{obs:ob1}
  At least one vertex of $C_i$ will be required in the dominating set
  in order to $r$-dominate $C_i$.
\end{observation}

Let $S$ be a set of $p$ vertices, one selected from each path in $\cal 
P$.  Let $\cal S$ be the collection of all such sets. 
We injectively map each set in ${\cal S}$ to a particular truth assignment for the corresponding group of variables.  Since the number of sets in ${\cal S}$ maybe larger than the number of truth assignments, we remove the sets that are not mapped to any truth assignment. For every $S \in
{\cal S}$, add a vertex $x_S$, and for each $P \in {\cal P}$,
identify an $r-$frame with top $x_S$ avoiding the vertex in $S \cap P$ with $P$ (refer to Fig.\ref{fig:group-gadgets} (a)). Attach each $x_S$ to a distinct vertex $\bar x_S$ via a path of length $r-1$. We then connect $\bar{x}_S $ to two new vertices $x$ and $x'$, for all $S \in {\cal S}$, and attach paths of length $r-1$ to each of $x$ and $x'$. Since no vertex in ${\cal P}$ can $r-$dominate, for example, $x$, we get: 

\begin{observation} \label{obs:ob2}
  The group gadget requires at least $p+1$ vertices to be $r$-dominated.
\end{observation}

\begin{SCfigure}
  \vspace{-20pt}
  \centering
  \input{figs/s-assemble.tex}
  \caption{The super group gadgets. Each shaded square represents a group gadget. Each row of group gadgets represents one group of variables. Each column of group gadgets represents all groups.}  \label{fig:assemble}
\end{SCfigure}
\paragraph{The super group gadgets} Refer to
Fig.~\ref{fig:assemble}.  For each group $F_i$, create $m(2rpt+1)$
copies $\{\hat{B}^1_i,...,\hat{B}^{(2rpt+1)m}_i\}$ of a group gadget. For every $j = 1, \ldots, (2rpt+1)m-1$ and $\ell = 1,\ldots,p$, connect the last vertex of $P_\ell$ in $\hat B_i^j$ to the first vertex of $P_\ell$ in $\hat B_i^{j+1}$. Add two vertices $h_1$ and $h_2$, connect $h_1$ to the first vertices of paths in $\hat{B}^1_i$ and $h_2$ to the last vertices of paths in $\hat{B}_i^{m(2rpt+1)}$ and attach two paths of length $r$ to each of $h_1$ and $h_2$.

\paragraph{Connecting the super group gadget to represent a SAT
  formula} Recall that each set $S \in {\cal S}$ for a particular
group of variables corresponds to a particular truth assignment for that group of
variables.  For each clause $j$, we create $2rpt+1$ {\em clause}
vertices $c^\ell_j$ for $\ell = 0, \ldots, 2ptr$ and connect each
clause vertex to all vertices $\bar x_S$ in $\{\hat B^{m\ell+j}_i | i = 1,\ldots,t\}$, for all $S \in {\cal S}$ that correspond to truth assignments that satisfy the clause $j$. Connect a path of length $r-1$ to each clause vertex.

\begin{lemma} \label{lm:sat-graph}
  If $\phi$ has a satisfying assignment, $G$ has a dominating of size
  $(p+1)tm(2rpt+1) + 2$.
\end{lemma}

\begin{proof}
Given a satisfying assignment of $\phi$, we construct the dominating set $D$ of $G$ as follows. For each group gadget $B_i^j, 1 \leq i \leq t, 1 \leq j \leq m(2rpt+1)$, we will select $\{\bar{x}_S\} \cup S$, for $S$ corresponding to the satisfying assignment of the group variables, for the $r-$dominating set. $S$ $r-$dominates:
	\begin{itemize} [noitemsep,nolistsep]
	\item All the guards and some vertices of their $r-$frames within distance $r$ from $S$.
	\item All the vertices $x_{S'}$ and some vertices of their $r-$frames within distance $r$ from $S$ for all $S' \in {\cal S} \backslash \{S\}$.
	\item All the vertices of the path $P_i$ in $B_i^j$ and maybe some vertices of its copies in $B_i^{j+1}$ and $B_i^{j-1}$ within distance $r$ from $S$(refer to
Fig.~\ref{fig:consecutive}) 
	\end{itemize}
The remaining vertices of the $r-$frames of guards and $x_S$ for $S \in {\cal S}$ that are not $r-$dominated by $S$ in $B_i^j$ would be $r-$dominated by the set $S$ of the nearby group gadgets.\\
The set of vertices that are $r$-dominated by the vertex $\bar{x}_S$ include:
	\begin{itemize} [noitemsep,nolistsep]
	\item The vertices of the path from $x_S$ to $\bar{x}_S$.
	\item The clause vertex connected to $\bar{x}_S$ and its attached path.
	\item The vertex $x$ and $x'$ and their attached paths.
	\item The vertices $\bar{x}_{S'}$ and the vertices of the path from $x_{S'}$ to $\bar{x}_{S'}$ for $S' \in {\cal S} \backslash \{S\}$.    
	\end{itemize}	
Taking the union over all $t$ groups, and all $m(2rpt+1)$ copies of the group gadgets in the super group gadgets gives $(p+1)tm(2rpt+1)$ vertices. Adding vertices $h_1$ and $h_2$ gives the lemma.
\end{proof}

\begin{lemma} \label{lm:graph-sat}
  If $G$ has a dominating set of size $(p+1)tm(2ptr+1) + 2$, then
  $\phi$ has a satisfying assignment.
\end{lemma}

\begin{proof} Let $D$ be the $r-$dominating set of size $(p+1)tm(2ptr+1) + 2$. Since some vertex in the paths attached to $h_1$ and $h_2$ must be in $D$, we can replace these with $h_1$ and $h_2$. By observation~\ref{obs:ob2}, at least $(p+1)$ vertices of each group gadget must  be in $D$, which implies that exactly $(p+1)$ vertices are chosen from each group gadget since there are $tm(2ptr+1)$ group gadgets. Let $\hat{B}_i^{j}, 1 \leq i \leq t, 1 \leq j \leq m(2ptr+1)$ be a group gadget and let $P_k \in {\cal P} = \{P_1,\ldots,P_p\}$ be a path of $\hat{B}_i^{j}$. By observation~\ref{obs:ob1}, at least one vertex from each $P_k , 1 \leq k \leq p$ must be included in $D$. To dominate the vertex $x$ and $x'$ and their attached paths, at least one vertex from the set $\{\bar{x}_S | S \in {\cal S}\}$ must be selected. Therefore, the set of $p+1$ vertices in $D \cap \hat{B}_i^{j}$ includes:
  	\begin{itemize} [noitemsep, nolistsep]
  	\item $p$ vertices, one from each path $P_k , 1 \leq k \leq p$, which make up the set $S$.
  	\item the vertex $\bar{x}_S$ that corresponds to $x_S$ since $x_S$ is not dominated by $S$. 
  	\end{itemize}

 We say that the dominating set $D$ is \emph{consistent} with a set 
  of gadgets $\{\hat{B}_i\}_{i=1}^k$ iff $D \cap {\cal P}$ is the same for all
  $\hat{B}_i$. We show that there exits a number $ \ell \in \{0,1,\ldots,2rtp\}$
  such that $D$ is consistent with the set of gadget $\{\hat{B}_i^{m\ell+j}| 1 \leq j \leq m$\} for each $1 \leq i \leq t$. For two consecutive gadgets $\hat{B}^q_i$ and $\hat{B}^{q+1}_i$, if two vertices   $p^a_i$ and $p^b_i$ of the path $P_i$ in $\hat{B}^q_i$  and of its copy in $\hat{B}^{q+1}_i$, respectively, are selected, the distance between them must be less than $2r + 1$ (refer to Figure ~\ref{fig:consecutive}). Therefore, we have $b \leq a$. We call two consecutive gadgets $\hat{B}^q_i$ and $\hat{B}^{q+1}_i$ \emph{bad} pair if $b < a$. Since the distance between $p^a_i$ and $p^b_i$ is smaller than $2r+1$, there are at most $2pr$ consecutive bad pairs for each $i$ and for $t$ groups of variables $F_i, 1 \leq i \leq t$, the number of bad pairs is no larger than $2rpt$. By the pigeonhole principle, there exists a number $\ell \in \{0,1,\ldots, 2rtp\}$ such that $D$ is consistent with the set of gadgets $\hat{B}_i^{m\ell+j} , 1 \leq j \leq m$ for all $i$.

For each $i \in \{1,\ldots,t\}$, let $\{\hat{B}_i^{m\ell+j} | 1 \leq j \leq m \}$ for some $ \ell \in \{0,1,\ldots,2prt\}$ be the set of group gadgets that is consistent with $D$ and let $F_i$ be the corresponding group of variables. We assign to the variables of group $F_i$ the values of assignment corresponding to the selected set $S$. This assignment satisfies the clauses of $\phi$ that are connected to the vertices $\bar{x}_S$. Because all clauses of $\phi$ are $r-$dominated, the truth assignment of all groups $F_i , 1 \leq i \leq t,$ makes up a satisfying assignment of $\phi$.
\end{proof}
\begin{SCfigure}
  \centering
\input{figs/s-consecutive-gadgets.tex}
\caption{Two paths $P_i$ in two consecutive gadgets for $r=3$. Two circled vertices are in the $r-$dominating set $D$. The vertex $p_i^1$ of the gadget $\hat{B}^{j+1}_i$ is not dominated by the vertex $p_i^5$ of the same gadget but it is dominated by the vertex $p_i^5$ of $\hat{B}^{j}_i$. The distance between two circled vertices must be no larger than $7$ $\left( = 2r+1 \right)$. The numbering of the vertices of the horizontal path is shown in Figure~\ref{fig:frame}.}
  \label{fig:consecutive}
\end{SCfigure}
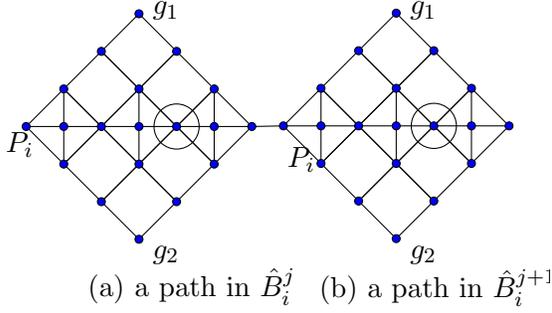

We prove the following bound on pathwidth using a \emph{mixed search game}~\cite{TUK95}.  Since the proof is similar to that of Lokshtanov et al.~\cite{LMS11}, we provide the proof in Appendix~\ref{App:AppendixB}.

\begin{lemma} \label{lm:pw-bound}
$\pw$(G) $\leq tp + O(rp)$.
\end{lemma}

Combined with Lemmas~\ref{lm:sat-graph} and~\ref{lm:graph-sat}, we get Theorem~\ref{thm:lb-ds}.

\section{Lower Bound for Connected $r$-dominating Set}\label{sec:lb-cds}
In this section, we prove Theorem~\ref{thm:dp-thm-cds}. The main idea is similar to that of the previous section: a reduction from $n_0$-variable, $m$-clause SAT to an instance of connected $r$-dominating set in a graph of pathwidth \[\pw \leq \frac{n_0p}{\lfloor p\log(2r+2) \rfloor} + O\left((2r+2)^{2p}\right)\mbox{ for any integer $p$.}\] Given this reduction, the final argument for Theorem~\ref{thm:lb-cds} is similar to the argument at the beginning of Section \ref{sec:lb-ds}.  Let $\phi$ be a SAT formula with $n_0$ variables.  For a given integer $p$, we assume that $n_0$ is divisible by $\lfloor p \log(2r+2) \rfloor$. We partition $\phi$'s variables into $t = \frac{n}{\lfloor p \log(2r+2)\rfloor}$ groups of variables $\{F_1,F_2, \ldots,F_t\}$ each of size $\lfloor p \log(2r+2) \rfloor$.  We will speak of an $r$-dominating {\em tree} as opposed to a connected $r$-dominating set: the tree is simply a witness to the connectedness of the $r$-dominating set.  We treat the problem as rooted: our construction has a global root vertex, $r_T$, which we will require to be in the \crds\ solution.  This can be forced by attaching a path of length $r$ to $r_T$.  In all the figures below, the dashed lines represent paths of length $r+1$ connected to $r_T$.

In the following, the length of a path is given by the number of edges in the path.

\paragraph{Core} A {\em core} is illustrated below. It is composed of a path with $2r+3$ vertices $a_1, a_2, \ldots, a_{2r+3}$, $2r+2$ edges $s_1, s_2, \ldots, s_{2r+2}$ (called {\em segments}), consecutive odd-indexed vertices connected by a subdivided edge and consecutive even-indexed vertices connected by a subdivided edge.  The even indexed vertices $a_2, a_4, \ldots, a_{2r+2}$ are connected the root $r_T$ via paths of length $r+1$.
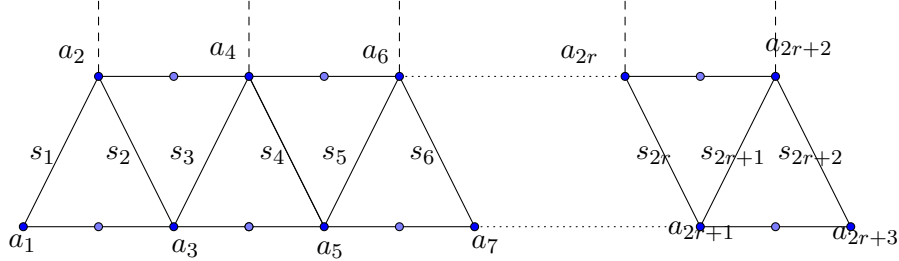
\vspace{-3em}
\begin{center}
  \input{figs/s-r-core-gadget.tex}
\end{center}
\vspace{-3em}

\paragraph{Pattern} A {\em pattern} $P_r(m)$ is illustrated below.  It is a tree-like graph with $m$ leaves and a single root $r_P$ such that the distance from the root to leaves is $r$.  The structure depends on the parity of $r$; if $r$ is even, the children of vertex $h$ are connected by a clique (indicated by the oval).  The dotted lines represent paths of length $\frac{r-1}{2}$ for $r$ even and $\frac{r}{2}-1$ for $r$ odd.  In future figures, we represent a pattern by a shaded box.
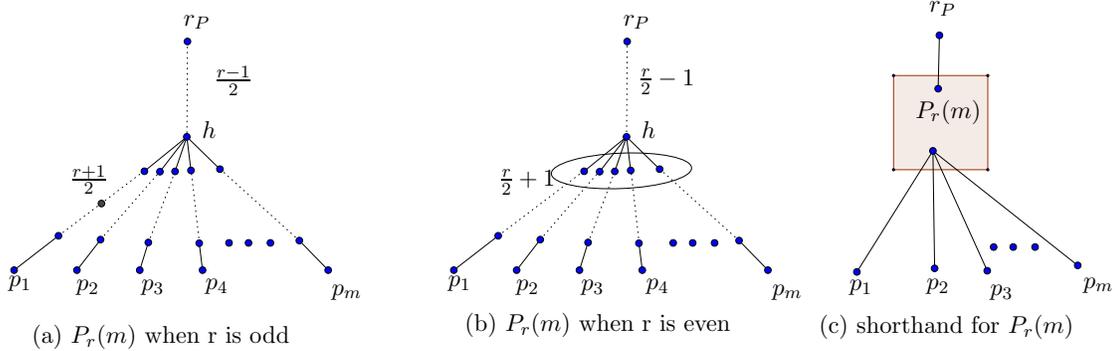
\vspace{-3em}
\begin{center}
\resizebox{1.0\textwidth}{!}{
 \input{figs/s-pattern.tex}
  }
\end{center}
\vspace{-3em}

\begin{observation} \label{obs:pattern}
A leaf of a pattern $r$-dominates all but the other leaves of the pattern.
\end{observation}
Given a set of $m$ vertices $X$, we say that pattern $P_r(m)$ is attached to set $X$ if the leaves of $P_r(m)$ are identified with $X$.

\paragraph{Core gadget} We connect patterns to the core in such a way as to force a minimum solution to contain a path from $r_T$ to the core, ending with a segment edge.  To each core that we use in the construction (these are not illustrated), we attach one pattern $P_r(r+1)$ to the odd-indexed vertices $a_1, a_2, \ldots, a_{2r+1}$ (but not $a_{2r+3}$) and another pattern $P_r(r+1)$ to the even-indexed vertices $a_2, a_4, \ldots, a_{2r+2}$.  In order to $r$-dominate the roots of these patterns, the dominating tree must contain a path from $r_T$ to an odd-indexed vertex and to an even-indexed vertex.  We attach additional \emph{path-forcing} patterns to guarantee that, even after adding the rest of the construction, this path will stay in the dominating tree.  For $i = 1, \ldots, r$, for the $r+1$ vertices that are $i$ hops from $r_T$, we attach a pattern $P_r(r+1)$.  As a result, at least one vertex at each distance from $r_T$ must be included in the dominating tree.  A core gadget is a subgraph of the larger construction such that edges from the remaining construction only attach to the vertices $a_1, a_2, \ldots, a_{2r+3}$ and $r_T$.  The previous observations guarantee:
\begin{observation}\label{obs:core-path}
  The part of a \crds\ that intersects a core gadget can be modified
  to contain a path from $r_T$ to an odd-indexed vertex (via a segment edge) without
  increasing its size.
\end{observation}

\paragraph{Group gadget} For each group $F_i$ of variables, we construct a {\em group gadget} which consists of $p$ cores $\{C_1,C_2, \ldots, C_p\}$ and is illustrated in Figure~\ref{fig:r-group-gadget}.  Let $S$ be a set of $p$ segments, one from each core, and let $\cal S$ be a collection of all possible such sets $S$; therefore $|{\cal S}| = (2r+2)^p$. Since a group represents $\lfloor p\log(2r+2) \rfloor$ variables, there are at most $2^{p \log (2r+2)} = (2r+2)^p$ truth assignments to each group of variables. We injectively map each set in ${\cal S}$ to a particular truth assignment for the corresponding group of variables.  Since the number of sets in ${\cal S}$ maybe larger than the number of truth assignments, we remove the sets that are not mapped to any truth assignment.
For each set $S \in {\cal S}$, we connect a corresponding {\em set} pattern $P_r(2rp)$ to the cores as follows.  For each $i = 1, \ldots, p$, $P_r(2rp)$ is attached to 
\begin{itemize} [noitemsep,nolistsep]
\item the vertices $a_1, a_2, \ldots, a_{2r+2}$ of $C_i$ {\em except} the endpoints of $s_j$ if $s_j \in S$
\item the vertices $a_2, a_3, \ldots, a_{2r+1}$ if $s_{2r+2} \in S$
\end{itemize}
We label the root of this pattern by the set vertex $x_S$.  We then connect these set patterns together.  For each $S \in {\cal S}$, we connect:
\begin{itemize} [noitemsep,nolistsep]
\item $x_S$ to a new vertex $\bar{x}_S$ via a path of length $r-1$
\item $\bar{x}_S$ to the root $r_T$ via paths of length $r+1$
\item $\bar{x}_S$ to a common vertex $x$, and 
\item $x$ to a path of length $r-1$.
\end{itemize} 
 Similarly, as in the core gadget construction, for the set of vertices $\{\bar{x}_S | S  \in {\cal S}\}$, we add {\em path forcing} patterns $P_r(|{\cal S|})$ to each level of vertices along the paths from $r$ to $\bar{x}_S, S \in {\cal S}$. 
\begin{figure}[!h] 
  \centering
  \input{figs/s-r-group-gadget.tex}
  \vspace{-20mm}
  \caption{A group gadget. Red segments are segments of the set $S$. The vertex $x_S$ is the root of the pattern shown.} \label{fig:r-group-gadget}
\end{figure}
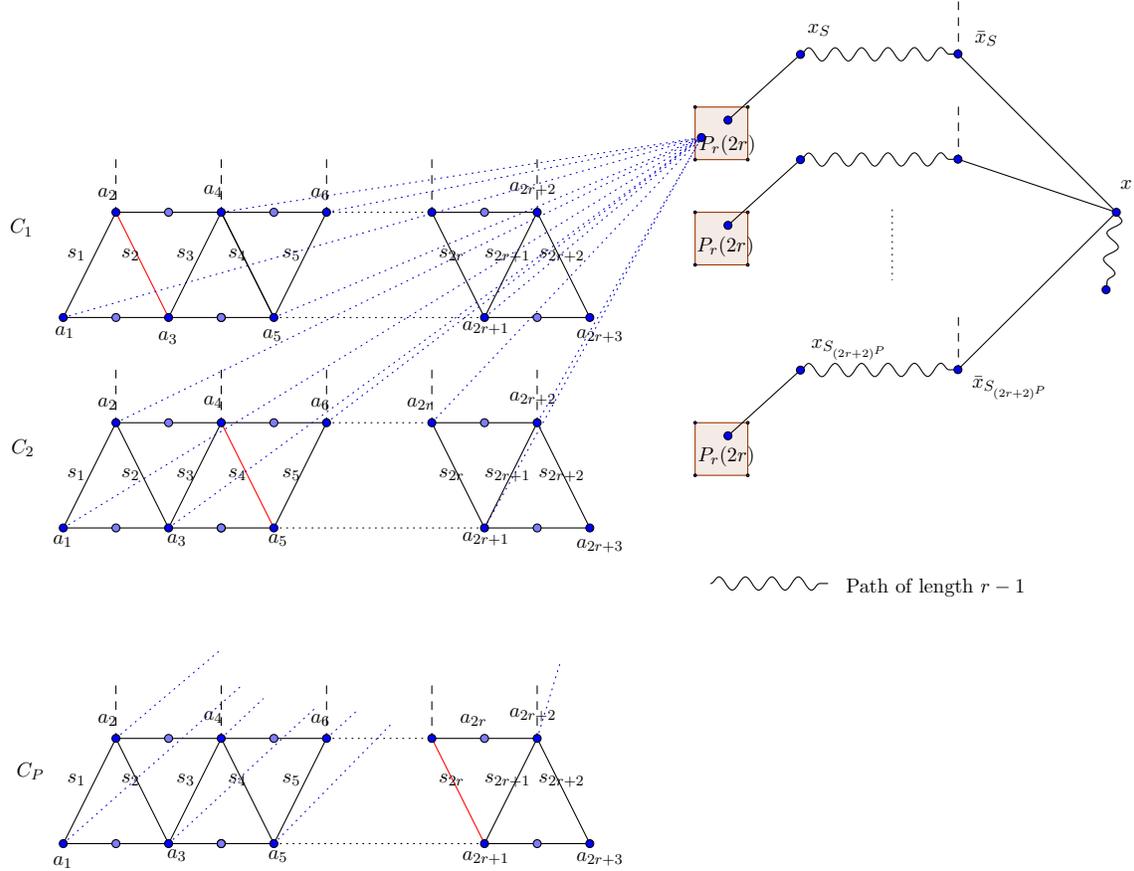
\begin{observation}\label{obs:group-path}
  The part of a \crds\ that intersects a group gadget can be modified
  to contain a path from $r_T$ to a vertex in the set $\{\bar{x}_S | S
  \in {\cal S}\}$ without increasing its size.
\end{observation}



\paragraph{Super-path} 
A \emph{super path} $F_i$ is a graph that consists of $X = m \left( (2r+1)pt + 1 \right)$ copies of the group gadget $B^1_i, B^2_i,\ldots,B^X_i$, which are assembled into a line ($m$ is the number of clauses). Vertex $a_{2r+s}$ of every core gadget of the group gadget $B^j_i$ is identified with the vertex $a_1$ of the corresponding core gadget of the group gadget $B^{j+1}_i$. The vertices $a_1$ and $a_{2r+3}$ of the cores gadgets of $B^1_i$ and $B^X_i$ are directedly connected to the root $r_T$.  In order to dominate all the odd- and even-indexed vertices of the cores (without spanning more than one segment edge per core), we must have:
\begin{observation}\label{obs:r-super-path}
  If an endpoint of segment edge $s_j$ in the $t^{th}$ core of $B^k_i$ is in the \crds, then there must be an endpoint of a segment  $s_{j'}$ in the $t^{th}$ core of $B^{k+1}_i$ that is also in the \crds, for $j' \le j$.
\end{observation}


\paragraph{Representing clauses} 
For each clause $C_j$ of $\phi$, we introduce $((2r+1)pt+1)$ clause vertices $c_j^{\ell}$, $0 \leq \ell \leq (2r+1)pt$. (There are $m((2r+1)pt+1)$ clause vertices in total.)  For a fixed $i$ ($1\leq i \leq t$) and for each $c_j^{\ell}$, ($1 \leq j \leq m, 0 \leq \ell\leq (2r+1)pt$), we connect $c_j^{\ell}$ to $B_i^{m\ell + j}$  by connecting it directly to the subset of vertices in the set $\{\bar{x}_{S}| S \in {\cal S}\}$ of $B_i^{m\ell + j}$ such that the truth assignments of the corresponding subsets in the collection ${\cal S}$ satisfy the clause $C_j$.  Each clause vertex is attached to a path of length $r-1$

Denote the final constructed graph as $G$.


\begin{lemma} \label{thm:r-domiating-number}
If $\phi$ has a satisfying assignment, $G$ has a connected $r$-dominating set of $((r+2)p+r+1)tm((2r+1)tp + 1) + 1$ vertices. 
\end{lemma}

\begin{proof}
Given a satisfying assignment of $\phi$, we construct an $r$-dominating tree $T$ as follows.  For group $i$, let $S_i$ be the set of $p$ segments  which corresponds to the truth assignment of variables of $F_i$.  In addition to the root, $T$ contains:
\begin{itemize}[nolistsep,noitemsep]
\item The path of length $r+2$ from $r_T$ that ends in each segment of $S_i$ for every group in the construction.  Each such path contains $r+2$ vertices in addition to the root.  As there are $tm((2r+1)tp + 1)$ groups and $p$ cores per group, this takes $(r+2)ptm((2r+1)tp + 1)$ vertices.  By Observation~\ref{obs:pattern}, this set of vertices will $r$-dominate all of the non-leaf vertices of all the patterns in the core gadget, since all these patterns include a leaf in one of these paths.  This set of vertices will also dominate $x_S$ for every $S \ne S_i$ since $S$ will connect to the endpoints of at least one segment edge that is not in $S_i$.
\item For each group, the path of length $r+1$ from $r_T$ to $\bar x_{S_i}$.  Each such path contains $r+1$ vertices (not including the root).  As there are $tm((2r+1)tp + 1)$ groups, this takes $(r+1)tm((2r+1)tp + 1)$ vertices (not including the root).  The vertex $\bar{x}_{S_i}$ $r$-dominates:
  \begin{itemize}[nolistsep,noitemsep]
  \item The vertices on the path from $x_{S_i}$ to $\bar{x}_{S_i}$.
  \item The vertex $x$ and the path attached to it.
  \item The clause vertex connected to $\bar{x}_{S_i}$ and its attached path.
  \item The vertices on the path from $x_{S}$ to $\bar{x}_{S}$, {\em not including} $x_S$, for every $S \ne S_i$.\qedhere
  \end{itemize}
\end{itemize}

\end{proof}

\begin{lemma} \label{lm:r-satisfying-assignment}
If $G$ has an $r$-dominating tree of $((r+2)p+r+1)tm((2r+1)tp + 1) + 1$ vertices, then $\phi$ has a satisfying assignment.
\end{lemma}

\begin{proof}
Let $T$ be the $r$-dominating tree;  $T$ contains the root $r_T$. By Observations~\ref{obs:core-path} and~\ref{obs:group-path}, each group gadget requires at least $(r+2)p + r+1$ vertices (not including the root) in the dominating tree. Since the number of copies of group gadget is $tm((2r+1)tp + 1)$, this implies exactly $(r+2)p + r+1$ vertices of each group gadget will be selected in which:
\begin{itemize}[nolistsep,noitemsep]
\item For each core gadget, exactly one segment $s_i$ in the set $\{s_1,s_2,\ldots,s_{2r+2}\}$ and a path connecting it to the root $r_T$ are selected which totals $(r+2)p$ vertices for $p$ cores. Denote the set of $p$ selected segments by $S$.
\item $r+1$ vertices on the path from $\bar{x}_S$ to the root $r_T$.
\end{itemize}
We say that $T$ is consistent with a set of group gadgets iff the set of segments in $T$ are the same for every group gadget.  If two segments $s_a$ and $s_b$ of two consecutive cores in group gadgets $B^q_i$ and $B_i^{q+1}$, respectively, are in $T$, by Observation~\ref{obs:r-super-path}, we have $b \leq a$. If $b < a$, we call $s_a$ and $s_b$ a \emph{bad pair}. Since there are $p$ cores in which there can be a bad pair, and each core has $2r+2$ segments, for each super-path, there can be at most $(2r+1)p$ consecutive bad pairs. Since we have $t$ super-paths, there are at most $tp(2r+1)$ bad pairs.  By the pigeonhole principle, there exists a number $\ell \in \{0,1,\ldots,tp(2r+1)\}$ such that $T$ is consistent with the set of gadgets $\{B_i^{m\ell+j} | 1 \leq i \leq t, 1 \leq j \leq m\}$.

Let $\{B_i^{m\ell+j} | 1 \leq i \leq t, 1 \leq j \leq m\}$ be the set of group gadgets which is consistent with $T$. For each group gadget $B_i^{m\ell+j}$, we assign the truth assignment corresponding to the set of segments $S \in T \cap B_i^{m\ell+j}$ to variables in the group $F_i$. The assignment of variables in all groups $F_i$ makes up a satisfying assignment of $\phi$, since all clause vertices are $r$-dominated by $T$.
\end{proof}

\begin{lemma} \label{lm:pw-cds}
$\pw(G) \leq tp + O((2r+2)^{2p})$.
\end{lemma}

The proof of Lemma~\ref{lm:pw-cds} is given in the Appendix~\ref{App:AppendixD} which completes the proof of Theorem~\ref{thm:lb-cds}.

\paragraph{Acknowledgments.} This material is based upon work supported by the National Science Foundation under Grant Nos.\ CCF-0963921 and CCF-1252833.

\newpage
\nocite{DEBF05,CNPPRW11,FLS14,LW10,DH2005,LMPRS13,RBR09,FT03,DFHT03,Kloks94,AN02,CK87,TP93,RS91,BA01, WL99, VV86}
 \bibliographystyle{plain}
 \bibliography{hung}
\newpage
\appendix
\input{Appendix.tex}
\end{document}

%% file: figs/s-frame-path.tex
\definecolor{qqqqff}{rgb}{0.0,0.0,1.0}
\begin{tikzpicture}[line cap=round,line join=round,>=triangle 45,x=0.5cm,y=0.5cm,every node/.style={scale=0.7}]
\clip(-7,1) rectangle (26,12);
\draw (9.011867490161439,7.868250574610374)-- (10.011867490161439,7.868250574610374);
\draw (10.011867490161439,7.868250574610374)-- (11.011867490161443,7.868250574610374);
\draw (10.011867490161439,7.868250574610374)-- (10.011867490161439,8.868250574610375);
\draw (10.011867490161439,8.868250574610375)-- (11.011867490161443,7.868250574610374);
\draw (10.011867490161439,8.868250574610375)-- (9.011867490161439,7.868250574610374);
\draw (11.011867490161443,7.868250574610374)-- (12.011867490161443,7.868250574610374);
\draw (12.011867490161443,7.868250574610374)-- (13.011867490161446,7.868250574610374);
\draw (12.011867490161443,7.868250574610374)-- (12.011867490161443,8.868250574610375);
\draw (12.011867490161443,8.868250574610375)-- (13.011867490161446,7.868250574610374);
\draw (12.011867490161443,8.868250574610375)-- (11.011867490161443,7.868250574610374);
\draw (13.011867490161446,7.868250574610374)-- (14.011867490161444,7.868250574610374);
\draw (14.011867490161444,7.868250574610374)-- (15.011867490161448,7.868250574610374);
\draw (14.011867490161444,7.868250574610374)-- (14.011867490161444,8.868250574610375);
\draw (14.011867490161444,8.868250574610375)-- (15.011867490161448,7.868250574610374);
\draw (14.011867490161444,8.868250574610375)-- (13.011867490161446,7.868250574610374);
\draw (11.011867490161439,9.868250574610375)-- (12.011867490161443,8.868250574610375);
\draw (11.011867490161439,9.868250574610375)-- (10.011867490161439,8.868250574610375);
\draw (10.011867490161439,8.868250574610375)-- (11.011867490161443,7.868250574610374);
\draw (12.011867490161443,8.868250574610375)-- (11.011867490161443,7.868250574610374);
\draw (13.01186749016144,9.868250574610375)-- (14.011867490161444,8.868250574610375);
\draw (13.01186749016144,9.868250574610375)-- (12.011867490161443,8.868250574610375);
\draw (12.011867490161443,8.868250574610375)-- (13.011867490161446,7.868250574610374);
\draw (14.011867490161444,8.868250574610375)-- (13.011867490161446,7.868250574610374);
\draw (12.011867490161439,10.868250574610375)-- (13.01186749016144,9.868250574610375);
\draw (12.011867490161439,10.868250574610375)-- (11.011867490161439,9.868250574610375);
\draw (11.011867490161439,9.868250574610375)-- (12.011867490161443,8.868250574610375);
\draw (13.01186749016144,9.868250574610375)-- (12.011867490161443,8.868250574610375);
\draw (8.987995549576587,6.841757129461701)-- (9.987995549576587,6.841757129461701);
\draw (9.987995549576587,6.841757129461701)-- (10.98799554957659,6.841757129461701);
\draw (10.98799554957659,6.841757129461701)-- (11.98799554957659,6.841757129461701);
\draw (11.98799554957659,6.841757129461701)-- (12.987995549576594,6.841757129461701);
\draw (12.987995549576594,6.841757129461701)-- (13.987995549576592,6.841757129461701);
\draw (13.987995549576592,6.841757129461701)-- (14.987995549576597,6.841757129461701);
\draw (9.01013427876043,5.873128623749571)-- (10.01013427876043,5.873128623749571);
\draw (10.01013427876043,5.873128623749571)-- (11.010134278760434,5.873128623749571);
\draw (10.01013427876043,5.873128623749571)-- (10.01013427876043,4.873128623749571);
\draw (9.01013427876043,5.873128623749571)-- (10.01013427876043,4.873128623749571);
\draw (11.010134278760434,5.873128623749571)-- (10.01013427876043,4.873128623749571);
\draw (11.010134278760434,5.873128623749571)-- (12.010134278760434,5.873128623749571);
\draw (12.010134278760434,5.873128623749571)-- (13.010134278760438,5.873128623749571);
\draw (12.010134278760434,5.873128623749571)-- (12.010134278760434,4.873128623749571);
\draw (11.010134278760434,5.873128623749571)-- (12.010134278760434,4.873128623749571);
\draw (13.010134278760438,5.873128623749571)-- (12.010134278760434,4.873128623749571);
\draw (13.010134278760438,5.873128623749571)-- (14.010134278760438,5.873128623749571);
\draw (14.010134278760438,5.873128623749571)-- (15.010134278760441,5.873128623749571);
\draw (14.010134278760438,5.873128623749571)-- (14.010134278760438,4.873128623749571);
\draw (13.010134278760438,5.873128623749571)-- (14.010134278760438,4.873128623749571);
\draw (15.010134278760441,5.873128623749571)-- (14.010134278760438,4.873128623749571);
\draw (11.010134278760434,5.873128623749571)-- (12.010134278760434,4.873128623749571);
\draw (11.010134278760434,5.873128623749571)-- (10.01013427876043,4.873128623749571);
\draw (10.01013427876043,4.873128623749571)-- (11.01013427876043,3.8731286237495706);
\draw (12.010134278760434,4.873128623749571)-- (11.01013427876043,3.8731286237495706);
\draw (13.010134278760438,5.873128623749571)-- (14.010134278760438,4.873128623749571);
\draw (13.010134278760438,5.873128623749571)-- (12.010134278760434,4.873128623749571);
\draw (12.010134278760434,4.873128623749571)-- (13.010134278760434,3.8731286237495706);
\draw (14.010134278760438,4.873128623749571)-- (13.010134278760434,3.8731286237495706);
\draw (12.010134278760434,4.873128623749571)-- (13.010134278760434,3.8731286237495706);
\draw (12.010134278760434,4.873128623749571)-- (11.01013427876043,3.8731286237495706);
\draw (11.01013427876043,3.8731286237495706)-- (12.01013427876043,2.8731286237495706);
\draw (13.010134278760434,3.8731286237495706)-- (12.01013427876043,2.8731286237495706);
\draw (17.56298457505404,7.091714210742007)-- (18.56298457505404,7.091714210742007);
\draw (18.56298457505404,7.091714210742007)-- (19.56298457505404,7.091714210742007);
\draw (18.56298457505404,7.091714210742007)-- (18.56298457505404,6.091714210742007);
\draw (18.56298457505404,7.091714210742007)-- (18.56298457505404,8.091714210742007);
\draw (18.56298457505404,8.091714210742007)-- (19.56298457505404,7.091714210742007);
\draw (18.56298457505404,8.091714210742007)-- (17.56298457505404,7.091714210742007);
\draw (17.56298457505404,7.091714210742007)-- (18.56298457505404,6.091714210742007);
\draw (19.56298457505404,7.091714210742007)-- (18.56298457505404,6.091714210742007);
\draw (19.56298457505404,7.091714210742007)-- (20.56298457505404,7.091714210742007);
\draw (20.56298457505404,7.091714210742007)-- (21.56298457505404,7.091714210742007);
\draw (20.56298457505404,7.091714210742007)-- (20.56298457505404,6.091714210742007);
\draw (20.56298457505404,7.091714210742007)-- (20.56298457505404,8.091714210742007);
\draw (20.56298457505404,8.091714210742007)-- (21.56298457505404,7.091714210742007);
\draw (20.56298457505404,8.091714210742007)-- (19.56298457505404,7.091714210742007);
\draw (19.56298457505404,7.091714210742007)-- (20.56298457505404,6.091714210742007);
\draw (21.56298457505404,7.091714210742007)-- (20.56298457505404,6.091714210742007);
\draw (21.56298457505404,7.091714210742007)-- (22.56298457505404,7.091714210742007);
\draw (22.56298457505404,7.091714210742007)-- (23.56298457505404,7.091714210742007);
\draw (22.56298457505404,7.091714210742007)-- (22.56298457505404,6.091714210742007);
\draw (22.56298457505404,7.091714210742007)-- (22.56298457505404,8.091714210742007);
\draw (22.56298457505404,8.091714210742007)-- (23.56298457505404,7.091714210742007);
\draw (22.56298457505404,8.091714210742007)-- (21.56298457505404,7.091714210742007);
\draw (21.56298457505404,7.091714210742007)-- (22.56298457505404,6.091714210742007);
\draw (23.56298457505404,7.091714210742007)-- (22.56298457505404,6.091714210742007);
\draw (19.56298457505404,9.09171421074201)-- (20.56298457505404,8.091714210742007);
\draw (19.56298457505404,9.09171421074201)-- (18.56298457505404,8.091714210742007);
\draw (18.56298457505404,8.091714210742007)-- (19.56298457505404,7.091714210742007);
\draw (20.56298457505404,8.091714210742007)-- (19.56298457505404,7.091714210742007);
\draw (21.56298457505404,9.09171421074201)-- (22.56298457505404,8.091714210742007);
\draw (21.56298457505404,9.09171421074201)-- (20.56298457505404,8.091714210742007);
\draw (20.56298457505404,8.091714210742007)-- (21.56298457505404,7.091714210742007);
\draw (22.56298457505404,8.091714210742007)-- (21.56298457505404,7.091714210742007);
\draw (20.56298457505404,10.09171421074201)-- (21.56298457505404,9.09171421074201);
\draw (20.56298457505404,10.09171421074201)-- (19.56298457505404,9.09171421074201);
\draw (19.56298457505404,9.09171421074201)-- (20.56298457505404,8.091714210742007);
\draw (21.56298457505404,9.09171421074201)-- (20.56298457505404,8.091714210742007);
\draw (19.56298457505404,7.091714210742007)-- (20.56298457505404,6.091714210742007);
\draw (19.56298457505404,7.091714210742007)-- (18.56298457505404,6.091714210742007);
\draw (18.56298457505404,6.091714210742007)-- (19.56298457505404,5.091714210742007);
\draw (20.56298457505404,6.091714210742007)-- (19.56298457505404,5.091714210742007);
\draw (21.56298457505404,7.091714210742007)-- (22.56298457505404,6.091714210742007);
\draw (21.56298457505404,7.091714210742007)-- (20.56298457505404,6.091714210742007);
\draw (20.56298457505404,6.091714210742007)-- (21.56298457505404,5.091714210742007);
\draw (22.56298457505404,6.091714210742007)-- (21.56298457505404,5.091714210742007);
\draw (20.56298457505404,6.091714210742007)-- (21.56298457505404,5.091714210742007);
\draw (20.56298457505404,6.091714210742007)-- (19.56298457505404,5.091714210742007);
\draw (19.56298457505404,5.091714210742007)-- (20.56298457505404,4.091714210742007);
\draw (21.56298457505404,5.091714210742007)-- (20.56298457505404,4.091714210742007);
\draw (11.759211216353243,11.721497393184648) node[anchor=north west] {$g_1$};
\draw (20.29370443805676,11.159345145048482) node[anchor=north west] {$g_1$};
\draw (20.447018687548443,4.0513659192783505) node[anchor=north west] {$g_2$};
\draw (8.181878728214045,7.575384157925895) node[anchor=north west] {$P$};
\draw (1.0,6.0)-- (2.0,6.0);
\draw (2.0,6.0)-- (3.0,6.0);
\draw (2.0,6.0)-- (2.0,7.0);
\draw (2.0,7.0)-- (3.0,6.0);
\draw (2.0,7.0)-- (1.0,6.0);
\draw (3.0,6.0)-- (4.0,6.0);
\draw (4.0,6.0)-- (5.0,6.0);
\draw (4.0,6.0)-- (4.0,7.0);
\draw (4.0,7.0)-- (3.0,6.0);
\draw (5.0,6.0)-- (6.0,6.0);
\draw (6.0,6.0)-- (7.0,6.0);
\draw (6.0,6.0)-- (6.0,7.0);
\draw (6.0,7.0)-- (7.0,6.0);
\draw (3.0,8.0)-- (4.0,7.0);
\draw (3.0,8.0)-- (2.0,7.0);
\draw (2.0,7.0)-- (3.0,6.0);
\draw (4.0,7.0)-- (3.0,6.0);
\draw (5.0,8.0)-- (6.0,7.0);
\draw (5.0,8.0)-- (4.0,7.0);
\draw[dashed] (4.0,7.0)-- (5.0,6.0);
\draw[dashed] (6.0,7.0)-- (5.0,6.0);
\draw (4.0,9.0)-- (5.0,8.0);
\draw (4.0,9.0)-- (3.0,8.0);
\draw (3.0,8.0)-- (4.0,7.0);
\draw (5.0,8.0)-- (4.0,7.0);
\draw (-6.5,6.0)-- (-5.5,6.0);
\draw (-5.5,6.0)-- (-4.5,6.0);
\draw (-5.5,6.0)-- (-5.5,7.0);
\draw (-5.5,7.0)-- (-4.5,6.0);
\draw (-5.5,7.0)-- (-6.5,6.0);
\draw (-4.5,6.0)-- (-3.5,6.0);
\draw (-3.5,6.0)-- (-2.5,6.0);
\draw (-3.5,6.0)-- (-3.5,7.0);
\draw (-3.5,7.0)-- (-2.5,6.0);
\draw (-3.5,7.0)-- (-4.5,6.0);
\draw (-2.5,6.0)-- (-1.5,6.0);
\draw (-1.5,6.0)-- (-0.5,6.0);
\draw (-1.5,6.0)-- (-1.5,7.0);
\draw (-1.5,7.0)-- (-0.5,6.0);
\draw (-1.5,7.0)-- (-2.5,6.0);
\draw (-4.5,8.0)-- (-3.5,7.0);
\draw (-4.5,8.0)-- (-5.5,7.0);
\draw (-5.5,7.0)-- (-4.5,6.0);
\draw (-3.5,7.0)-- (-4.5,6.0);
\draw (-2.5,8.0)-- (-1.5,7.0);
\draw (-2.5,8.0)-- (-3.5,7.0);
\draw (-3.5,7.0)-- (-2.5,6.0);
\draw (-1.5,7.0)-- (-2.5,6.0);
\draw (-3.5,9.0)-- (-2.5,8.0);
\draw (-3.5,9.0)-- (-4.5,8.0);
\draw (-4.5,8.0)-- (-3.5,7.0);
\draw (-2.5,8.0)-- (-3.5,7.0);
\draw (11.657001716692124,2.829270922667143) node[anchor=north west] {$g_2$};
\draw (-6.689603472478911,5.597765412161886) node[anchor=north west] {$p_1$};
\draw (-5.667508475867711,5.598870161992446) node[anchor=north west] {$p_2$};
\draw (-4.696518229087071,5.598870161992446) node[anchor=north west] {$p_3$};
\draw (-0.8081382426422728,5.598870161992446) node[anchor=north west] {$p_{2r+1}$};
\draw (0.8227947526134068,5.598870161992446) node[anchor=north west] {$p_1$};
\draw (1.8959944990551667,5.599974911823007) node[anchor=north west] {$p_2$};
\draw (2.8158799960052465,5.599974911823007) node[anchor=north west] {$p_3$};
\draw (6.755364732280605,5.599974911823007) node[anchor=north west] {$p_{2r+1}$};
\draw (4.911174739058206,5.591079661653567) node[anchor=north west] {$p$};
\draw (-3.7437989200075776,10.092248859406792) node[anchor=north west] {$A$};
\draw (3.6538735519679415,10.03315690649204) node[anchor=north west] {$A$};
\draw (-7.231273085367465,7.02271823531315) node[anchor=north west] {$B$};
\draw (-0.39723280173293807,7.003633894159189) node[anchor=north west] {$C$};
\draw (0.44821548077854984,7.003633894159189) node[anchor=north west] {$B$};
\draw (6.859531623157333,7.003633894159189) node[anchor=north west] {$C$};
\begin{scriptsize}
\draw [fill=qqqqff] (9.011867490161439,7.868250574610374) circle (1.5pt);
\draw [fill=qqqqff] (10.011867490161439,7.868250574610374) circle (1.5pt);
\draw [fill=qqqqff] (11.011867490161443,7.868250574610374) circle (1.5pt);
\draw [fill=qqqqff] (10.011867490161439,8.868250574610375) circle (1.5pt);
\draw [fill=qqqqff] (12.011867490161443,7.868250574610374) circle (1.5pt);
\draw [fill=qqqqff] (13.011867490161446,7.868250574610374) circle (1.5pt);
\draw [fill=qqqqff] (12.011867490161443,8.868250574610375) circle (1.5pt);
\draw [fill=qqqqff] (14.011867490161444,7.868250574610374) circle (1.5pt);
\draw [fill=qqqqff] (15.011867490161448,7.868250574610374) circle (1.5pt);
\draw [fill=qqqqff] (14.011867490161444,8.868250574610375) circle (1.5pt);
\draw [fill=qqqqff] (11.011867490161439,9.868250574610375) circle (1.5pt);
\draw [fill=qqqqff] (13.01186749016144,9.868250574610375) circle (1.5pt);
\draw [fill=qqqqff] (12.011867490161439,10.868250574610375) circle (1.5pt);
\draw [fill=qqqqff] (8.987995549576587,6.841757129461701) circle (1.5pt);
\draw [fill=qqqqff] (9.987995549576587,6.841757129461701) circle (1.5pt);
\draw [fill=qqqqff] (10.98799554957659,6.841757129461701) circle (1.5pt);
\draw [fill=qqqqff] (11.98799554957659,6.841757129461701) circle (1.5pt);
\draw [fill=qqqqff] (12.987995549576594,6.841757129461701) circle (1.5pt);
\draw [fill=qqqqff] (13.987995549576592,6.841757129461701) circle (1.5pt);
\draw [fill=qqqqff] (14.987995549576597,6.841757129461701) circle (1.5pt);
\draw [fill=qqqqff] (9.01013427876043,5.873128623749571) circle (1.5pt);
\draw [fill=qqqqff] (10.01013427876043,5.873128623749571) circle (1.5pt);
\draw [fill=qqqqff] (11.010134278760434,5.873128623749571) circle (1.5pt);
\draw [fill=qqqqff] (10.01013427876043,4.873128623749571) circle (1.5pt);
\draw [fill=qqqqff] (12.010134278760434,5.873128623749571) circle (1.5pt);
\draw [fill=qqqqff] (13.010134278760438,5.873128623749571) circle (1.5pt);
\draw [fill=qqqqff] (12.010134278760434,4.873128623749571) circle (1.5pt);
\draw [fill=qqqqff] (14.010134278760438,5.873128623749571) circle (1.5pt);
\draw [fill=qqqqff] (15.010134278760441,5.873128623749571) circle (1.5pt);
\draw [fill=qqqqff] (14.010134278760438,4.873128623749571) circle (1.5pt);
\draw [fill=qqqqff] (11.01013427876043,3.8731286237495706) circle (1.5pt);
\draw [fill=qqqqff] (13.010134278760434,3.8731286237495706) circle (1.5pt);
\draw [fill=qqqqff] (12.01013427876043,2.8731286237495706) circle (1.5pt);
\draw [fill=qqqqff] (17.56298457505404,7.091714210742007) circle (1.5pt);
\draw [fill=qqqqff] (18.56298457505404,7.091714210742007) circle (1.5pt);
\draw [fill=qqqqff] (19.56298457505404,7.091714210742007) circle (1.5pt);
\draw [fill=qqqqff] (18.56298457505404,6.091714210742007) circle (1.5pt);
\draw [fill=qqqqff] (18.56298457505404,8.091714210742007) circle (1.5pt);
\draw [fill=qqqqff] (20.56298457505404,7.091714210742007) circle (1.5pt);
\draw [fill=qqqqff] (21.56298457505404,7.091714210742007) circle (1.5pt);
\draw [fill=qqqqff] (20.56298457505404,6.091714210742007) circle (1.5pt);
\draw [fill=qqqqff] (20.56298457505404,8.091714210742007) circle (1.5pt);
\draw [fill=qqqqff] (22.56298457505404,7.091714210742007) circle (1.5pt);
\draw [fill=qqqqff] (23.56298457505404,7.091714210742007) circle (1.5pt);
\draw [fill=qqqqff] (22.56298457505404,6.091714210742007) circle (1.5pt);
\draw [fill=qqqqff] (22.56298457505404,8.091714210742007) circle (1.5pt);
\draw [fill=qqqqff] (19.56298457505404,9.09171421074201) circle (1.5pt);
\draw [fill=qqqqff] (21.56298457505404,9.09171421074201) circle (1.5pt);
\draw [fill=qqqqff] (20.56298457505404,10.09171421074201) circle (1.5pt);
\draw [fill=qqqqff] (19.56298457505404,5.091714210742007) circle (1.5pt);
\draw [fill=qqqqff] (21.56298457505404,5.091714210742007) circle (1.5pt);
\draw [fill=qqqqff] (20.56298457505404,4.091714210742007) circle (1.5pt);
\draw [fill=qqqqff] (1.0,6.0) circle (1.5pt);
\draw [fill=qqqqff] (2.0,6.0) circle (1.5pt);
\draw [fill=qqqqff] (3.0,6.0) circle (1.5pt);
\draw [fill=qqqqff] (2.0,7.0) circle (1.5pt);
\draw [fill=qqqqff] (4.0,6.0) circle (1.5pt);
\draw [fill=qqqqff] (5.0,6.0) circle (1.5pt);
\draw [fill=qqqqff] (4.0,7.0) circle (1.5pt);
\draw [fill=qqqqff] (6.0,6.0) circle (1.5pt);
\draw [fill=qqqqff] (7.0,6.0) circle (1.5pt);
\draw [fill=qqqqff] (6.0,7.0) circle (1.5pt);
\draw [fill=qqqqff] (3.0,8.0) circle (1.5pt);
\draw [fill=qqqqff] (5.0,8.0) circle (1.5pt);
\draw [fill=qqqqff] (4.0,9.0) circle (1.5pt);
\draw [fill=qqqqff] (-6.5,6.0) circle (1.5pt);
\draw [fill=qqqqff] (-5.5,6.0) circle (1.5pt);
\draw [fill=qqqqff] (-4.5,6.0) circle (1.5pt);
\draw [fill=qqqqff] (-5.5,7.0) circle (1.5pt);
\draw [fill=qqqqff] (-3.5,6.0) circle (1.5pt);
\draw [fill=qqqqff] (-2.5,6.0) circle (1.5pt);
\draw [fill=qqqqff] (-3.5,7.0) circle (1.5pt);
\draw [fill=qqqqff] (-1.5,6.0) circle (1.5pt);
\draw [fill=qqqqff] (-0.5,6.0) circle (1.5pt);
\draw [fill=qqqqff] (-1.5,7.0) circle (1.5pt);
\draw [fill=qqqqff] (-4.5,8.0) circle (1.5pt);
\draw [fill=qqqqff] (-2.5,8.0) circle (1.5pt);
\draw [fill=qqqqff] (-3.5,9.0) circle (1.5pt);
\end{scriptsize}
\end{tikzpicture}

%% file: figs/s-group-gadgets.tex
\definecolor{sqsqsq}{rgb}{0.12549019607843137,0.12549019607843137,0.12549019607843137}
\definecolor{xfqqff}{rgb}{0.4980392156862745,0.0,1.0}
\definecolor{ffqqqq}{rgb}{1.0,0.0,0.0}
\definecolor{qqqqff}{rgb}{0.0,0.0,1.0}
\begin{tikzpicture}[line cap=round,line join=round,>=triangle 45,x=0.4cm,y=0.4cm,every node/.style={scale=0.7}]
\clip(-0.5,-6) rectangle (21.7982682428186,20);
\draw (1.0907928159748779,6.934887148264213)-- (2.090792815974877,6.934887148264213);
\draw (2.090792815974877,6.934887148264213)-- (3.0907928159748805,6.934887148264213);
\draw (2.090792815974877,6.934887148264213)-- (2.090792815974877,5.934887148264213);
\draw (2.090792815974877,6.934887148264213)-- (2.090792815974877,7.934887148264213);
\draw (2.090792815974877,7.934887148264213)-- (3.0907928159748805,6.934887148264213);
\draw (2.090792815974877,7.934887148264213)-- (1.0907928159748779,6.934887148264213);
\draw (1.0907928159748779,6.934887148264213)-- (2.090792815974877,5.934887148264213);
\draw (3.0907928159748805,6.934887148264213)-- (2.090792815974877,5.934887148264213);
\draw (3.0907928159748805,6.934887148264213)-- (4.090792815974881,6.934887148264213);
\draw (4.090792815974881,6.934887148264213)-- (5.0907928159748845,6.934887148264213);
\draw (4.090792815974881,6.934887148264213)-- (4.090792815974881,5.934887148264213);
\draw (4.090792815974881,6.934887148264213)-- (4.090792815974881,7.934887148264213);
\draw (4.090792815974881,7.934887148264213)-- (5.0907928159748845,6.934887148264213);
\draw (4.090792815974881,7.934887148264213)-- (3.0907928159748805,6.934887148264213);
\draw (3.0907928159748805,6.934887148264213)-- (4.090792815974881,5.934887148264213);
\draw (5.0907928159748845,6.934887148264213)-- (4.090792815974881,5.934887148264213);
\draw (5.0907928159748845,6.934887148264213)-- (6.090792815974884,6.934887148264213);
\draw (6.090792815974884,6.934887148264213)-- (7.090792815974886,6.934887148264213);
\draw (6.090792815974884,6.934887148264213)-- (6.090792815974884,5.934887148264213);
\draw (6.090792815974884,6.934887148264213)-- (6.090792815974884,7.934887148264213);
\draw (6.090792815974884,7.934887148264213)-- (7.090792815974886,6.934887148264213);
\draw (6.090792815974884,7.934887148264213)-- (5.0907928159748845,6.934887148264213);
\draw (5.0907928159748845,6.934887148264213)-- (6.090792815974884,5.934887148264213);
\draw (7.090792815974886,6.934887148264213)-- (6.090792815974884,5.934887148264213);
\draw (3.090792815974877,8.934887148264213)-- (4.090792815974881,7.934887148264213);
\draw (3.090792815974877,8.934887148264213)-- (2.090792815974877,7.934887148264213);
\draw (2.090792815974877,7.934887148264213)-- (3.0907928159748805,6.934887148264213);
\draw (4.090792815974881,7.934887148264213)-- (3.0907928159748805,6.934887148264213);
\draw (5.09079281597488,8.934887148264213)-- (6.090792815974884,7.934887148264213);
\draw (5.09079281597488,8.934887148264213)-- (4.090792815974881,7.934887148264213);
\draw (4.090792815974881,7.934887148264213)-- (5.0907928159748845,6.934887148264213);
\draw (6.090792815974884,7.934887148264213)-- (5.0907928159748845,6.934887148264213);
\draw (4.090792815974877,9.934887148264213)-- (5.09079281597488,8.934887148264213);
\draw (4.090792815974877,9.934887148264213)-- (3.090792815974877,8.934887148264213);
\draw (3.090792815974877,8.934887148264213)-- (4.090792815974881,7.934887148264213);
\draw (5.09079281597488,8.934887148264213)-- (4.090792815974881,7.934887148264213);
\draw (3.0907928159748805,6.934887148264213)-- (4.090792815974881,5.934887148264213);
\draw (3.0907928159748805,6.934887148264213)-- (2.090792815974877,5.934887148264213);
\draw (2.090792815974877,5.934887148264213)-- (3.090792815974876,4.934887148264213);
\draw (4.090792815974881,5.934887148264213)-- (3.090792815974876,4.934887148264213);
\draw (5.0907928159748845,6.934887148264213)-- (6.090792815974884,5.934887148264213);
\draw (5.0907928159748845,6.934887148264213)-- (4.090792815974881,5.934887148264213);
\draw (4.090792815974881,5.934887148264213)-- (5.09079281597488,4.934887148264213);
\draw (6.090792815974884,5.934887148264213)-- (5.09079281597488,4.934887148264213);
\draw (4.090792815974881,5.934887148264213)-- (5.09079281597488,4.934887148264213);
\draw (4.090792815974881,5.934887148264213)-- (3.090792815974876,4.934887148264213);
\draw (3.090792815974876,4.934887148264213)-- (4.090792815974876,3.934887148264213);
\draw (5.09079281597488,4.934887148264213)-- (4.090792815974876,3.934887148264213);
\draw (1.0907928159748779,6.934887148264213)-- (2.090792815974877,6.934887148264213);
\draw (2.090792815974877,6.934887148264213)-- (3.0907928159748805,6.934887148264213);
\draw [color=ffqqqq] (2.090792815974877,6.934887148264213)-- (3.0907928159748814,7.934887148264213);
\draw [color=ffqqqq] (3.0907928159748814,7.934887148264213)-- (1.0907928159748779,6.934887148264213);
\draw (3.0907928159748805,6.934887148264213)-- (4.090792815974881,6.934887148264213);
\draw (4.090792815974881,6.934887148264213)-- (5.0907928159748845,6.934887148264213);
\draw [color=ffqqqq] (4.090792815974881,6.934887148264213)-- (5.090792815974892,7.934887148264213);
\draw (5.0907928159748845,6.934887148264213)-- (6.090792815974884,6.934887148264213);
\draw (6.090792815974884,6.934887148264213)-- (7.090792815974886,6.934887148264213);
\draw [color=ffqqqq] (6.090792815974884,6.934887148264213)-- (7.090792815974909,7.934887148264213);
\draw [color=ffqqqq] (5.929047980289921,9.528831621540125)-- (3.0907928159748814,7.934887148264213);
\draw [color=ffqqqq] (3.0907928159748814,7.934887148264213)-- (3.0907928159748805,6.934887148264213);
\draw [color=ffqqqq] (5.090792815974892,7.934887148264213)-- (3.0907928159748805,6.934887148264213);
\draw [color=ffqqqq] (7.090792815974895,8.934887148264213)-- (7.090792815974909,7.934887148264213);
\draw [color=ffqqqq] (11.023321241378655,19.395781552781557)-- (7.090792815974895,8.934887148264213);
\draw [color=ffqqqq] (11.023321241378655,19.395781552781557)-- (5.929047980289921,9.528831621540125);
\draw [color=ffqqqq] (5.929047980289921,9.528831621540125)-- (5.090792815974892,7.934887148264213);
\draw [color=ffqqqq] (7.090792815974895,8.934887148264213)-- (5.090792815974892,7.934887148264213);
\draw (0.6904279347905997,-0.6349597804678995)-- (1.6904279347905997,-0.6349597804678995);
\draw (1.6904279347905997,-0.6349597804678995)-- (2.6904279347906015,-0.6349597804678995);
\draw (1.6904279347905997,-0.6349597804678995)-- (1.6904279347905997,-1.6349597804678977);
\draw (1.6904279347905997,-0.6349597804678995)-- (1.6904279347905997,0.3650402195321014);
\draw (1.6904279347905997,0.3650402195321014)-- (2.6904279347906015,-0.6349597804678995);
\draw (1.6904279347905997,0.3650402195321014)-- (0.6904279347905997,-0.6349597804678995);
\draw (0.6904279347905997,-0.6349597804678995)-- (1.6904279347905997,-1.6349597804678977);
\draw (2.6904279347906015,-0.6349597804678995)-- (1.6904279347905997,-1.6349597804678977);
\draw (2.6904279347906015,-0.6349597804678995)-- (3.690427934790601,-0.6349597804678995);
\draw (3.690427934790601,-0.6349597804678995)-- (4.690427934790607,-0.6349597804678995);
\draw (3.690427934790601,-0.6349597804678995)-- (3.690427934790601,-1.6349597804678977);
\draw (3.690427934790601,-0.6349597804678995)-- (3.690427934790601,0.3650402195321014);
\draw (3.690427934790601,0.3650402195321014)-- (4.690427934790607,-0.6349597804678995);
\draw (3.690427934790601,0.3650402195321014)-- (2.6904279347906015,-0.6349597804678995);
\draw (2.6904279347906015,-0.6349597804678995)-- (3.690427934790601,-1.6349597804678977);
\draw (4.690427934790607,-0.6349597804678995)-- (3.690427934790601,-1.6349597804678977);
\draw (4.690427934790607,-0.6349597804678995)-- (5.690427934790604,-0.6349597804678995);
\draw (5.690427934790604,-0.6349597804678995)-- (6.690427934790611,-0.6349597804678995);
\draw (5.690427934790604,-0.6349597804678995)-- (5.690427934790604,-1.6349597804678977);
\draw (5.690427934790604,-0.6349597804678995)-- (5.690427934790604,0.3650402195321014);
\draw (5.690427934790604,0.3650402195321014)-- (6.690427934790611,-0.6349597804678995);
\draw (5.690427934790604,0.3650402195321014)-- (4.690427934790607,-0.6349597804678995);
\draw (4.690427934790607,-0.6349597804678995)-- (5.690427934790604,-1.6349597804678977);
\draw (6.690427934790611,-0.6349597804678995)-- (5.690427934790604,-1.6349597804678977);
\draw (2.690427934790598,1.3650402195321005)-- (3.690427934790601,0.3650402195321014);
\draw (2.690427934790598,1.3650402195321005)-- (1.6904279347905997,0.3650402195321014);
\draw (1.6904279347905997,0.3650402195321014)-- (2.6904279347906015,-0.6349597804678995);
\draw (3.690427934790601,0.3650402195321014)-- (2.6904279347906015,-0.6349597804678995);
\draw (4.6904279347906,1.3650402195321005)-- (5.690427934790604,0.3650402195321014);
\draw (4.6904279347906,1.3650402195321005)-- (3.690427934790601,0.3650402195321014);
\draw (3.690427934790601,0.3650402195321014)-- (4.690427934790607,-0.6349597804678995);
\draw (5.690427934790604,0.3650402195321014)-- (4.690427934790607,-0.6349597804678995);
\draw (3.690427934790601,2.3650402195321005)-- (4.6904279347906,1.3650402195321005);
\draw (3.690427934790601,2.3650402195321005)-- (2.690427934790598,1.3650402195321005);
\draw (2.690427934790598,1.3650402195321005)-- (3.690427934790601,0.3650402195321014);
\draw (4.6904279347906,1.3650402195321005)-- (3.690427934790601,0.3650402195321014);
\draw (2.6904279347906015,-0.6349597804678995)-- (3.690427934790601,-1.6349597804678977);
\draw (2.6904279347906015,-0.6349597804678995)-- (1.6904279347905997,-1.6349597804678977);
\draw (1.6904279347905997,-1.6349597804678977)-- (2.690427934790598,-2.6349597804678986);
\draw (3.690427934790601,-1.6349597804678977)-- (2.690427934790598,-2.6349597804678986);
\draw (4.690427934790607,-0.6349597804678995)-- (5.690427934790604,-1.6349597804678977);
\draw (4.690427934790607,-0.6349597804678995)-- (3.690427934790601,-1.6349597804678977);
\draw (3.690427934790601,-1.6349597804678977)-- (4.6904279347906,-2.6349597804678986);
\draw (5.690427934790604,-1.6349597804678977)-- (4.6904279347906,-2.6349597804678986);
\draw (3.690427934790601,-1.6349597804678977)-- (4.6904279347906,-2.6349597804678986);
\draw (3.690427934790601,-1.6349597804678977)-- (2.690427934790598,-2.6349597804678986);
\draw (2.690427934790598,-2.6349597804678986)-- (3.6904279347905975,-3.6349597804679);
\draw (4.6904279347906,-2.6349597804678986)-- (3.6904279347905975,-3.6349597804679);
\draw (0.6904279347905997,-0.6349597804678995)-- (1.6904279347905997,-0.6349597804678995);
\draw (1.6904279347905997,-0.6349597804678995)-- (2.6904279347906015,-0.6349597804678995);
\draw [color=ffqqqq] (1.6904279347905997,-0.6349597804678995)-- (2.690427934790605,0.3650402195321014);
\draw [color=ffqqqq] (2.690427934790605,0.3650402195321014)-- (0.6904279347905997,-0.6349597804678995);
\draw (2.6904279347906015,-0.6349597804678995)-- (3.690427934790601,-0.6349597804678995);
\draw (3.690427934790601,-0.6349597804678995)-- (4.690427934790607,-0.6349597804678995);
\draw (4.690427934790607,-0.6349597804678995)-- (5.690427934790604,-0.6349597804678995);
\draw (5.690427934790604,-0.6349597804678995)-- (6.690427934790611,-0.6349597804678995);
\draw [color=ffqqqq] (5.690427934790604,-0.6349597804678995)-- (6.690427934790632,0.3650402195321014);
\draw [color=ffqqqq] (6.690427934790632,0.3650402195321014)-- (6.690427934790611,-0.6349597804678995);
\draw [color=ffqqqq] (5.258364589941163,2.1227746089655253)-- (2.690427934790605,0.3650402195321014);
\draw [color=ffqqqq] (2.690427934790605,0.3650402195321014)-- (2.6904279347906015,-0.6349597804678995);
\draw [color=ffqqqq] (4.690427934790614,0.3650402195321014)-- (2.6904279347906015,-0.6349597804678995);
\draw [color=ffqqqq] (6.690427934790618,1.3650402195321005)-- (6.690427934790632,0.3650402195321014);
\draw [color=ffqqqq] (11.082434828203658,19.425338346194057)-- (6.690427934790618,1.3650402195321005);
\draw [color=ffqqqq] (11.082434828203658,19.425338346194057)-- (5.258364589941163,2.1227746089655253);
\draw [color=ffqqqq] (5.258364589941163,2.1227746089655253)-- (4.690427934790614,0.3650402195321014);
\draw [color=ffqqqq] (6.690427934790618,1.3650402195321005)-- (4.690427934790614,0.3650402195321014);
\draw (1.1633366293906229,14.793686380857883)-- (2.163336629390623,14.793686380857883);
\draw (2.163336629390623,14.793686380857883)-- (3.1633366293906247,14.793686380857883);
\draw (2.163336629390623,14.793686380857883)-- (2.163336629390623,13.793686380857881);
\draw (2.163336629390623,14.793686380857883)-- (2.163336629390623,15.79368638085789);
\draw (2.163336629390623,15.79368638085789)-- (3.1633366293906247,14.793686380857883);
\draw (2.163336629390623,15.79368638085789)-- (1.1633366293906229,14.793686380857883);
\draw (1.1633366293906229,14.793686380857883)-- (2.163336629390623,13.793686380857881);
\draw (3.1633366293906247,14.793686380857883)-- (2.163336629390623,13.793686380857881);
\draw (3.1633366293906247,14.793686380857883)-- (4.1633366293906295,14.793686380857883);
\draw (4.1633366293906295,14.793686380857883)-- (5.163336629390626,14.793686380857883);
\draw (4.1633366293906295,14.793686380857883)-- (4.1633366293906295,13.793686380857881);
\draw [color=sqsqsq] (4.1633366293906295,14.793686380857883)-- (4.1633366293906295,15.79368638085789);
\draw [color = red] (3.16,15.79) -- (3.16,14.79);
\draw (4.1633366293906295,15.79368638085789)-- (5.163336629390626,14.793686380857883);
\draw (3.1633366293906247,14.793686380857883)-- (4.1633366293906295,13.793686380857881);
\draw (5.163336629390626,14.793686380857883)-- (4.1633366293906295,13.793686380857881);
\draw (5.163336629390626,14.793686380857883)-- (6.163336629390628,14.793686380857883);
\draw (6.163336629390628,14.793686380857883)-- (7.163336629390631,14.793686380857883);
\draw (6.163336629390628,14.793686380857883)-- (6.163336629390628,13.793686380857881);
\draw (6.163336629390628,14.793686380857883)-- (6.163336629390628,15.79368638085789);
\draw (6.163336629390628,15.79368638085789)-- (7.163336629390631,14.793686380857883);
\draw (6.163336629390628,15.79368638085789)-- (5.163336629390626,14.793686380857883);
\draw (5.163336629390626,14.793686380857883)-- (6.163336629390628,13.793686380857881);
\draw (7.163336629390631,14.793686380857883)-- (6.163336629390628,13.793686380857881);
\draw (3.163336629390621,16.793686380857892)-- (4.1633366293906295,15.79368638085789);
\draw (3.163336629390621,16.793686380857892)-- (2.163336629390623,15.79368638085789);
\draw (2.163336629390623,15.79368638085789)-- (3.1633366293906247,14.793686380857883);
\draw (4.1633366293906295,15.79368638085789)-- (3.1633366293906247,14.793686380857883);
\draw (5.163336629390622,16.793686380857892)-- (6.163336629390628,15.79368638085789);
\draw (5.163336629390622,16.793686380857892)-- (4.1633366293906295,15.79368638085789);
\draw (4.1633366293906295,15.79368638085789)-- (5.163336629390626,14.793686380857883);
\draw (6.163336629390628,15.79368638085789)-- (5.163336629390626,14.793686380857883);
\draw (4.163336629390622,17.793686380857892)-- (5.163336629390622,16.793686380857892);
\draw (4.163336629390622,17.793686380857892)-- (3.163336629390621,16.793686380857892);
\draw (3.163336629390621,16.793686380857892)-- (4.1633366293906295,15.79368638085789);
\draw (5.163336629390622,16.793686380857892)-- (4.1633366293906295,15.79368638085789);
\draw (3.1633366293906247,14.793686380857883)-- (4.1633366293906295,13.793686380857881);
\draw (3.1633366293906247,14.793686380857883)-- (2.163336629390623,13.793686380857881);
\draw (2.163336629390623,13.793686380857881)-- (3.163336629390621,12.79368638085788);
\draw (4.1633366293906295,13.793686380857881)-- (3.163336629390621,12.79368638085788);
\draw (5.163336629390626,14.793686380857883)-- (6.163336629390628,13.793686380857881);
\draw (5.163336629390626,14.793686380857883)-- (4.1633366293906295,13.793686380857881);
\draw (4.1633366293906295,13.793686380857881)-- (5.163336629390622,12.79368638085788);
\draw (6.163336629390628,13.793686380857881)-- (5.163336629390622,12.79368638085788);
\draw (4.1633366293906295,13.793686380857881)-- (5.163336629390622,12.79368638085788);
\draw (4.1633366293906295,13.793686380857881)-- (3.163336629390621,12.79368638085788);
\draw (3.163336629390621,12.79368638085788)-- (4.163336629390622,11.79368638085788);
\draw (5.163336629390622,12.79368638085788)-- (4.163336629390622,11.79368638085788);
\draw (1.1633366293906229,14.793686380857883)-- (2.163336629390623,14.793686380857883);
\draw (2.163336629390623,14.793686380857883)-- (3.1633366293906247,14.793686380857883);
\draw [color=ffqqqq] (2.163336629390623,14.793686380857883)-- (3.163336629390628,15.79368638085789);
\draw [color=ffqqqq] (3.163336629390628,15.79368638085789)-- (3.1633366293906247,14.793686380857883);
\draw [color=ffqqqq] (3.163336629390628,15.79368638085789)-- (1.1633366293906229,14.793686380857883);
\draw (3.1633366293906247,14.793686380857883)-- (4.1633366293906295,14.793686380857883);
\draw (4.1633366293906295,14.793686380857883)-- (5.163336629390626,14.793686380857883);
\draw [color=ffqqqq] (4.1633366293906295,14.793686380857883)-- (5.163336629390633,15.79368638085789);
\draw (5.163336629390626,14.793686380857883)-- (6.163336629390628,14.793686380857883);
\draw (6.163336629390628,14.793686380857883)-- (7.163336629390631,14.793686380857883);
\draw [color=ffqqqq] (6.163336629390628,14.793686380857883)-- (7.163336629390656,15.79368638085789);
\draw [color=ffqqqq] (7.163336629390656,15.79368638085789)-- (7.163336629390631,14.793686380857883);
\draw [color=ffqqqq] (6.0,17.0)-- (3.163336629390628,15.79368638085789);
\draw [color=ffqqqq] (5.163336629390633,15.79368638085789)-- (3.1633366293906247,14.793686380857883);
\draw [color=ffqqqq] (7.163336629390642,16.793686380857892)-- (7.163336629390656,15.79368638085789);
\draw [color=ffqqqq] (11.052878034791156,19.395781552781557)-- (7.163336629390642,16.793686380857892);
\draw [color=ffqqqq] (11.052878034791156,19.395781552781557)-- (6.0,17.0);
\draw [color=ffqqqq] (6.0,17.0)-- (5.163336629390633,15.79368638085789);
\draw [color=ffqqqq] (7.163336629390642,16.793686380857892)-- (5.163336629390633,15.79368638085789);
\draw(5.163336629390626,14.793686380857883) circle (0.3cm);
\draw [color=ffqqqq] (5.090792815974892,7.934887148264213)-- (5.0907928159748845,6.934887148264213);
\draw [color=ffqqqq] (5.0907928159748845,6.934887148264213)-- (7.090792815974909,7.934887148264213);
\draw(7.090792815974886,6.934887148264213) circle (0.3cm);
\draw [color=ffqqqq] (4.690427934790614,0.3650402195321014)-- (4.690427934790607,-0.6349597804678995);
\draw [color=ffqqqq] (4.690427934790607,-0.6349597804678995)-- (6.690427934790632,0.3650402195321014);
\draw(3.690427934790601,-0.6349597804678995) circle (0.3cm);
\draw (14.375855214233557,15.81352407943479)-- (17.195158582793063,13.946193276882392);
\draw (14.412469543695368,17.64424055252538)-- (17.195158582793063,13.946193276882392);
\draw (14.5677117634359,19.55162769438255)-- (17.195158582793063,13.946193276882392);
\draw (14.375855214233557,14.165879253653262)-- (17.195158582793063,13.946193276882392);
\draw (-10.077935929882187,-2.3508235558782355)-- (-9.077935929882187,-2.3508235558782355);
\draw (-9.077935929882187,-2.3508235558782355)-- (-8.077935929882184,-2.3508235558782355);
\draw (-9.077935929882187,-2.3508235558782355)-- (-9.077935929882187,-1.350823555878233);
\draw (-9.077935929882187,-1.350823555878233)-- (-8.077935929882184,-2.3508235558782355);
\draw (-9.077935929882187,-1.350823555878233)-- (-10.077935929882187,-2.3508235558782355);
\draw (-8.077935929882184,-2.3508235558782355)-- (-7.077935929882184,-2.3508235558782355);
\draw (-7.077935929882184,-2.3508235558782355)-- (-6.07793592988218,-2.3508235558782355);
\draw (-7.077935929882184,-2.3508235558782355)-- (-7.077935929882184,-1.350823555878233);
\draw (-7.077935929882184,-1.350823555878233)-- (-8.077935929882184,-2.3508235558782355);
\draw (-6.07793592988218,-2.3508235558782355)-- (-5.077935929882181,-2.3508235558782355);
\draw (-5.077935929882181,-2.3508235558782355)-- (-4.0779359298821785,-2.3508235558782355);
\draw (-5.077935929882181,-2.3508235558782355)-- (-5.077935929882181,-1.350823555878233);
\draw (-5.077935929882181,-1.350823555878233)-- (-4.0779359298821785,-2.3508235558782355);
\draw (-8.077935929882187,-0.35082355587822933)-- (-7.077935929882184,-1.350823555878233);
\draw (-8.077935929882187,-0.35082355587822933)-- (-9.077935929882187,-1.350823555878233);
\draw (-9.077935929882187,-1.350823555878233)-- (-8.077935929882184,-2.3508235558782355);
\draw (-7.077935929882184,-1.350823555878233)-- (-8.077935929882184,-2.3508235558782355);
\draw (-6.077935929882186,-0.35082355587822933)-- (-5.077935929882181,-1.350823555878233);
\draw (-6.077935929882186,-0.35082355587822933)-- (-7.077935929882184,-1.350823555878233);
\draw (-7.077935929882187,0.6491764441217671)-- (-6.077935929882186,-0.35082355587822933);
\draw (-7.077935929882187,0.6491764441217671)-- (-8.077935929882187,-0.35082355587822933);
\draw (-8.077935929882187,-0.35082355587822933)-- (-7.077935929882184,-1.350823555878233);
\draw (-6.077935929882186,-0.35082355587822933)-- (-7.077935929882184,-1.350823555878233);
\draw (-10.077935929882187,-2.3508235558782355)-- (-9.077935929882187,-2.3508235558782355);
\draw (-9.077935929882187,-2.3508235558782355)-- (-8.077935929882184,-2.3508235558782355);
\draw (-8.077935929882184,-2.3508235558782355)-- (-7.077935929882184,-2.3508235558782355);
\draw (-7.077935929882184,-2.3508235558782355)-- (-6.07793592988218,-2.3508235558782355);
\draw (-6.07793592988218,-2.3508235558782355)-- (-5.077935929882181,-2.3508235558782355);
\draw (-5.077935929882181,-2.3508235558782355)-- (-4.0779359298821785,-2.3508235558782355);
\draw (17.66,17.34) -- (14.57,19.55);
\draw (17.66,17.34) -- (14.41,17.64);
\draw (17.66,17.34) -- (14.38,15.81);
\draw (17.66,17.34) -- (14.38,14.17);
\draw(-6.07793592988218,-2.3508235558782355) circle (0.3cm);
\draw(11,2) circle (0.3cm);
\draw[dashed] (-6.08,-2.35) -- (-5.08,-1.35);
\draw[dashed] (-6.08,-2.35) -- (-7.08,-1.35);

\draw (-7.403554426070674,2.045340402255375) node[anchor=north west] {$x_{S}$};
\draw (-10.802925087385974,-1.6407241702551862) node[anchor=north west] {P};
\draw (-0.15429410013322398,16.011429504767612) node[anchor=north west] {$P_1$};
\draw (-0.40003173830059513,8.14782508341175) node[anchor=north west] {$P_2$};
\draw (-0.8095944685795471,0.4070894811395701) node[anchor=north west] {$P_p$};
\draw (10.412424341063737,20.885225995087133) node[anchor=north west] {$x_S$};
\draw (14.180401459630094,20.926182268115028) node[anchor=north west] {$\bar{x}_S$};
\draw (11.026768436482165,18.79645607066448) node[anchor=north west] {$x_{S'}$};
\draw (14.508051643853255,18.79645607066448) node[anchor=north west] {$\bar{x}_{S'}$};
\draw (17.45690330186171,15.02847895209813) node[anchor=north west] {$x$};
\draw (3.8184643835726098,18.660281162776063) node[anchor=north west] {$g_1$};
\draw (3.859420656600505,11.779627294089681) node[anchor=north west] {$g_2$};
\draw (3.7365518375168194,11.101458106559673) node[anchor=north west] {$g_1$};
\draw (3.777508110544715,3.9160228727338165) node[anchor=north west] {$g_2$};
\draw (3.1222077420983916,3.442635050343285) node[anchor=north west] {$g_1$};
\draw (3.3269891072378677,-3.6428001834825714) node[anchor=north west] {$g_2$};
\draw (-8.79606770901911,-3.647581548622047) node[anchor=north west] {$x_S$};
\draw (-7.649292064238044,-4.220969371012579) node[anchor=north west] {$(a)$};
\draw (17.06,18.60) node[anchor=north west] {$x'$};
\draw (12.3,4.7) node[anchor=north west] {path of length $r-1$};
\draw (12.3,2.7) node[anchor=north west] {vertex in $S$};
\begin{scriptsize}
\draw [fill=qqqqff] (1.0907928159748779,6.934887148264213) circle (1.5pt);
\draw [fill=qqqqff] (2.090792815974877,6.934887148264213) circle (1.5pt);
\draw [fill=qqqqff] (3.0907928159748805,6.934887148264213) circle (1.5pt);
\draw [fill=qqqqff] (2.090792815974877,5.934887148264213) circle (1.5pt);
\draw [fill=qqqqff] (2.090792815974877,7.934887148264213) circle (1.5pt);
\draw [fill=qqqqff] (4.090792815974881,6.934887148264213) circle (1.5pt);
\draw [fill=qqqqff] (5.0907928159748845,6.934887148264213) circle (1.5pt);
\draw [fill=qqqqff] (4.090792815974881,5.934887148264213) circle (1.5pt);
\draw [fill=qqqqff] (4.090792815974881,7.934887148264213) circle (1.5pt);
\draw [fill=qqqqff] (6.090792815974884,6.934887148264213) circle (1.5pt);
\draw [fill=qqqqff] (7.090792815974886,6.934887148264213) circle (1.5pt);
\draw [fill=qqqqff] (6.090792815974884,5.934887148264213) circle (1.5pt);
\draw [fill=qqqqff] (6.090792815974884,7.934887148264213) circle (1.5pt);
\draw [fill=qqqqff] (3.090792815974877,8.934887148264213) circle (1.5pt);
\draw [fill=qqqqff] (5.09079281597488,8.934887148264213) circle (1.5pt);
\draw [fill=qqqqff] (4.090792815974877,9.934887148264213) circle (1.5pt);
\draw [fill=qqqqff] (3.090792815974876,4.934887148264213) circle (1.5pt);
\draw [fill=qqqqff] (5.09079281597488,4.934887148264213) circle (1.5pt);
\draw [fill=qqqqff] (4.090792815974876,3.934887148264213) circle (1.5pt);
\draw [fill=qqqqff] (3.0907928159748814,7.934887148264213) circle (1.5pt);
\draw [fill=qqqqff] (5.090792815974892,7.934887148264213) circle (1.5pt);
\draw [fill=qqqqff] (7.090792815974909,7.934887148264213) circle (1.5pt);
\draw [fill=qqqqff] (5.929047980289921,9.528831621540125) circle (1.5pt);
\draw [fill=qqqqff] (7.090792815974895,8.934887148264213) circle (1.5pt);
\draw [fill=qqqqff] (11.023321241378655,19.395781552781557) circle (1.5pt);
\draw [fill=qqqqff] (0.6904279347905997,-0.6349597804678995) circle (1.5pt);
\draw [fill=qqqqff] (1.6904279347905997,-0.6349597804678995) circle (1.5pt);
\draw [fill=qqqqff] (2.6904279347906015,-0.6349597804678995) circle (1.5pt);
\draw [fill=qqqqff] (1.6904279347905997,-1.6349597804678977) circle (1.5pt);
\draw [fill=qqqqff] (1.6904279347905997,0.3650402195321014) circle (1.5pt);
\draw [fill=qqqqff] (3.690427934790601,-0.6349597804678995) circle (1.5pt);
\draw [fill=qqqqff] (4.690427934790607,-0.6349597804678995) circle (1.5pt);
\draw [fill=qqqqff] (3.690427934790601,-1.6349597804678977) circle (1.5pt);
\draw [fill=qqqqff] (3.690427934790601,0.3650402195321014) circle (1.5pt);
\draw [fill=qqqqff] (5.690427934790604,-0.6349597804678995) circle (1.5pt);
\draw [fill=qqqqff] (6.690427934790611,-0.6349597804678995) circle (1.5pt);
\draw [fill=qqqqff] (5.690427934790604,-1.6349597804678977) circle (1.5pt);
\draw [fill=qqqqff] (5.690427934790604,0.3650402195321014) circle (1.5pt);
\draw [fill=qqqqff] (2.690427934790598,1.3650402195321005) circle (1.5pt);
\draw [fill=qqqqff] (4.6904279347906,1.3650402195321005) circle (1.5pt);
\draw [fill=qqqqff] (3.690427934790601,2.3650402195321005) circle (1.5pt);
\draw [fill=qqqqff] (2.690427934790598,-2.6349597804678986) circle (1.5pt);
\draw [fill=qqqqff] (4.6904279347906,-2.6349597804678986) circle (1.5pt);
\draw [fill=qqqqff] (3.6904279347905975,-3.6349597804679) circle (1.5pt);
\draw [fill=qqqqff] (2.690427934790605,0.3650402195321014) circle (1.5pt);
\draw [fill=qqqqff] (4.690427934790614,0.3650402195321014) circle (1.5pt);
\draw [fill=qqqqff] (6.690427934790632,0.3650402195321014) circle (1.5pt);
\draw [fill=qqqqff] (5.258364589941163,2.1227746089655253) circle (1.5pt);
\draw [fill=qqqqff] (6.690427934790618,1.3650402195321005) circle (1.5pt);
\draw [fill=qqqqff] (11.082434828203658,19.425338346194057) circle (1.5pt);
\draw [fill=qqqqff] (1.1633366293906229,14.793686380857883) circle (1.5pt);
\draw [fill=qqqqff] (2.163336629390623,14.793686380857883) circle (1.5pt);
\draw [fill=qqqqff] (3.1633366293906247,14.793686380857883) circle (1.5pt);
\draw [fill=qqqqff] (2.163336629390623,13.793686380857881) circle (1.5pt);
\draw [fill=qqqqff] (2.163336629390623,15.79368638085789) circle (1.5pt);
\draw [fill=qqqqff] (4.1633366293906295,14.793686380857883) circle (1.5pt);
\draw [fill=qqqqff] (5.163336629390626,14.793686380857883) circle (1.5pt);
\draw [fill=qqqqff] (4.1633366293906295,13.793686380857881) circle (1.5pt);
\draw [fill=qqqqff] (4.1633366293906295,15.79368638085789) circle (1.5pt);
\draw [fill=qqqqff] (6.163336629390628,14.793686380857883) circle (1.5pt);
\draw [fill=qqqqff] (7.163336629390631,14.793686380857883) circle (1.5pt);
\draw [fill=qqqqff] (6.163336629390628,13.793686380857881) circle (1.5pt);
\draw [fill=qqqqff] (6.163336629390628,15.79368638085789) circle (1.5pt);
\draw [fill=qqqqff] (3.163336629390621,16.793686380857892) circle (1.5pt);
\draw [fill=qqqqff] (5.163336629390622,16.793686380857892) circle (1.5pt);
\draw [fill=qqqqff] (4.163336629390622,17.793686380857892) circle (1.5pt);
\draw [fill=qqqqff] (3.163336629390621,12.79368638085788) circle (1.5pt);
\draw [fill=qqqqff] (5.163336629390622,12.79368638085788) circle (1.5pt);
\draw [fill=qqqqff] (4.163336629390622,11.79368638085788) circle (1.5pt);
\draw [fill=qqqqff] (3.163336629390628,15.79368638085789) circle (1.5pt);
\draw [fill=qqqqff] (5.163336629390633,15.79368638085789) circle (1.5pt);
\draw [fill=qqqqff] (7.163336629390656,15.79368638085789) circle (1.5pt);
\draw [fill=qqqqff] (6.0,17.0) circle (1.5pt);
\draw [fill=qqqqff] (7.163336629390642,16.793686380857892) circle (1.5pt);
\draw [fill=qqqqff] (11.052878034791156,19.395781552781557) circle (1.5pt);
\draw [fill=qqqqff] (14.5677117634359,19.55162769438255) circle (1.5pt);
\draw [fill=qqqqff] (10.934108244823248,17.571011893601757) circle (1.5pt);
\draw [fill=qqqqff] (14.412469543695368,17.64424055252538) circle (1.5pt);
\draw [fill=qqqqff] (10.860879585899625,15.740295420511167) circle (1.5pt);
\draw [fill=qqqqff] (14.375855214233557,15.81352407943479) circle (1.5pt);
\draw [fill=qqqqff] (17.195158582793063,13.946193276882392) circle (1.5pt);
\draw [fill=qqqqff] (14.375855214233557,14.165879253653262) circle (1.5pt);
\draw [fill=qqqqff] (18.659731761265533,12.115476803791802) circle (1.5pt);
\draw [fill=qqqqff] (-10.077935929882187,-2.3508235558782355) circle (1.5pt);
\draw [fill=qqqqff] (-9.077935929882187,-2.3508235558782355) circle (1.5pt);
\draw [fill=qqqqff] (-8.077935929882184,-2.3508235558782355) circle (1.5pt);
\draw [fill=qqqqff] (-9.077935929882187,-1.350823555878233) circle (1.5pt);
\draw [fill=qqqqff] (-7.077935929882184,-2.3508235558782355) circle (1.5pt);
\draw [fill=qqqqff] (-6.07793592988218,-2.3508235558782355) circle (1.5pt);
\draw [fill=qqqqff] (-7.077935929882184,-1.350823555878233) circle (1.5pt);
\draw [fill=qqqqff] (-5.077935929882181,-2.3508235558782355) circle (1.5pt);
\draw [fill=qqqqff] (-4.0779359298821785,-2.3508235558782355) circle (1.5pt);
\draw [fill=qqqqff] (-5.077935929882181,-1.350823555878233) circle (1.5pt);
\draw [fill=qqqqff] (-8.077935929882187,-0.35082355587822933) circle (1.5pt);
\draw [fill=qqqqff] (-6.077935929882186,-0.35082355587822933) circle (1.5pt);
\draw [fill=qqqqff] (-7.077935929882187,0.6491764441217671) circle (1.5pt);
\draw [fill=qqqqff] (17.66,17.34) circle (1.5pt);
\draw [fill=qqqqff] (18.88,16.28) circle (1.5pt);
\draw [fill=qqqqff] (10,4) circle (1.5pt);
\draw [fill=qqqqff] (12,4) circle (1.5pt);
\draw [fill=qqqqff] (11,2) circle (1.5pt);

\end{scriptsize}
\draw[snake=snake=expanding waves]  (11.09,19.4) -- (14.57,19.55);
\draw[snake=snake=expanding waves]  (10.93,17.57) -- (14.41,17.64);
\draw[snake=snake=expanding waves]  (10.85,15.74) -- (14.38,15.81);
\draw[snake=snake=expanding waves]  (17.2,13.95) -- (18.66,12.12);
\draw[snake=snake=expanding waves]  (10,4) -- (12,4);
\draw[snake=snake=expanding waves]  (17.66,17.34) -- (18.88,16.28);

\end{tikzpicture}

%% file: figs/s-assemble.tex
\definecolor{xdxdff}{rgb}{0.49019607843137253,0.49019607843137253,1.0}
\definecolor{zzttqq}{rgb}{0.6,0.2,0.0}
\definecolor{qqqqff}{rgb}{0.0,0.0,1.0}
\definecolor{qqqqff}{rgb}{0.0,0.0,1.0}

\begin{tikzpicture}[line cap=round,line join=round,>=triangle 45,x=0.5cm,y=0.5cm]
\fill[color=zzttqq,fill=zzttqq,fill opacity=0.1] (1.0,2.0) -- (1.0,1.0) -- (1.52,1.0) -- (1.52,2.0) -- cycle;
\fill[color=zzttqq,fill=zzttqq,fill opacity=0.1] (2.0,2.0) -- (2.0,1.0) -- (2.52,1.0) -- (2.52,2.0) -- cycle;
\fill[color=zzttqq,fill=zzttqq,fill opacity=0.1] (3.0,2.0) -- (3.0,1.0) -- (3.52,1.0) -- (3.52,2.0) -- cycle;
\fill[color=zzttqq,fill=zzttqq,fill opacity=0.1] (4.0,2.0) -- (4.0,1.0) -- (4.52,1.0) -- (4.52,2.0) -- cycle;
\fill[color=zzttqq,fill=zzttqq,fill opacity=0.1] (5.0,2.0) -- (5.0,1.0) -- (5.52,1.0) -- (5.52,2.0) -- cycle;
\fill[color=zzttqq,fill=zzttqq,fill opacity=0.1] (1.0314404498753975,0.3090110819121666) -- (1.0314404498753975,-0.6909889180878335) -- (1.5514404498753978,-0.6909889180878335) -- (1.5514404498753978,0.3090110819121666) -- cycle;
\fill[color=zzttqq,fill=zzttqq,fill opacity=0.1] (2.0314404498753977,0.3090110819121666) -- (2.0314404498753977,-0.6909889180878335) -- (2.5514404498753973,-0.6909889180878335) -- (2.5514404498753973,0.3090110819121666) -- cycle;
\fill[color=zzttqq,fill=zzttqq,fill opacity=0.1] (3.031440449875397,0.3090110819121666) -- (3.031440449875397,-0.6909889180878335) -- (3.5514404498753964,-0.6909889180878335) -- (3.5514404498753964,0.3090110819121666) -- cycle;
\fill[color=zzttqq,fill=zzttqq,fill opacity=0.1] (4.031440449875397,0.3090110819121666) -- (4.031440449875397,-0.6909889180878335) -- (4.551440449875397,-0.6909889180878335) -- (4.551440449875397,0.3090110819121666) -- cycle;
\fill[color=zzttqq,fill=zzttqq,fill opacity=0.1] (5.031440449875397,0.3090110819121666) -- (5.031440449875397,-0.6909889180878335) -- (5.551440449875397,-0.6909889180878335) -- (5.551440449875397,0.3090110819121666) -- cycle;
\fill[color=zzttqq,fill=zzttqq,fill opacity=0.1] (1.0314404498753975,-1.3596193363773281) -- (1.0314404498753975,-2.35961933637733) -- (1.5514404498753975,-2.35961933637733) -- (1.5514404498753975,-1.3596193363773281) -- cycle;
\fill[color=zzttqq,fill=zzttqq,fill opacity=0.1] (2.031440449875398,-1.3596193363773281) -- (2.031440449875398,-2.35961933637733) -- (2.5514404498753986,-2.35961933637733) -- (2.5514404498753986,-1.3596193363773281) -- cycle;
\fill[color=zzttqq,fill=zzttqq,fill opacity=0.1] (3.031440449875398,-1.3596193363773281) -- (3.031440449875398,-2.35961933637733) -- (3.5514404498753986,-2.35961933637733) -- (3.5514404498753986,-1.3596193363773281) -- cycle;
\fill[color=zzttqq,fill=zzttqq,fill opacity=0.1] (4.031440449875398,-1.3596193363773281) -- (4.031440449875398,-2.35961933637733) -- (4.551440449875398,-2.35961933637733) -- (4.551440449875398,-1.3596193363773281) -- cycle;
\fill[color=zzttqq,fill=zzttqq,fill opacity=0.1] (5.031440449875398,-1.3596193363773281) -- (5.031440449875398,-2.35961933637733) -- (5.551440449875398,-2.35961933637733) -- (5.551440449875398,-1.3596193363773281) -- cycle;
\draw [color=zzttqq] (1.0,2.0)-- (1.0,1.0);
\draw [color=zzttqq] (1.0,1.0)-- (1.52,1.0);
\draw [color=zzttqq] (1.52,1.0)-- (1.52,2.0);
\draw [color=zzttqq] (1.52,2.0)-- (1.0,2.0);
\draw [color=zzttqq] (2.0,2.0)-- (2.0,1.0);
\draw [color=zzttqq] (2.0,1.0)-- (2.52,1.0);
\draw [color=zzttqq] (2.52,1.0)-- (2.52,2.0);
\draw [color=zzttqq] (2.52,2.0)-- (2.0,2.0);
\draw [color=zzttqq] (3.0,2.0)-- (3.0,1.0);
\draw [color=zzttqq] (3.0,1.0)-- (3.52,1.0);
\draw [color=zzttqq] (3.52,1.0)-- (3.52,2.0);
\draw [color=zzttqq] (3.52,2.0)-- (3.0,2.0);
\draw [color=zzttqq] (4.0,2.0)-- (4.0,1.0);
\draw [color=zzttqq] (4.0,1.0)-- (4.52,1.0);
\draw [color=zzttqq] (4.52,1.0)-- (4.52,2.0);
\draw [color=zzttqq] (4.52,2.0)-- (4.0,2.0);
\draw [color=zzttqq] (5.0,2.0)-- (5.0,1.0);
\draw [color=zzttqq] (5.0,1.0)-- (5.52,1.0);
\draw [color=zzttqq] (5.52,1.0)-- (5.52,2.0);
\draw [color=zzttqq] (5.52,2.0)-- (5.0,2.0);
\draw (1.52,1.6527814708516673)-- (2.0,1.6594361991275612);
\draw (1.52,1.466927469391784)-- (2.0,1.4560569373861258);
\draw (1.52,1.2686832011679086)-- (2.0,1.2730961402946521);
\draw (2.52,1.6527814708516673)-- (3.0,1.6594361991275612);
\draw (2.52,1.466927469391784)-- (3.0,1.4560569373861258);
\draw (2.52,1.2686832011679086)-- (3.0,1.2730961402946521);
\draw (3.52,1.6527814708516673)-- (4.0,1.6594361991275612);
\draw (3.52,1.466927469391784)-- (4.0,1.4560569373861258);
\draw (3.52,1.2686832011679086)-- (4.0,1.2730961402946521);
\draw (4.52,1.6527814708516673)-- (5.0,1.6594361991275612);
\draw (4.52,1.466927469391784)-- (5.0,1.4560569373861258);
\draw (4.52,1.2686832011679086)-- (5.0,1.2730961402946521);
\draw [color=zzttqq] (1.0314404498753975,0.3090110819121666)-- (1.0314404498753975,-0.6909889180878335);
\draw [color=zzttqq] (1.0314404498753975,-0.6909889180878335)-- (1.5514404498753978,-0.6909889180878335);
\draw [color=zzttqq] (1.5514404498753978,-0.6909889180878335)-- (1.5514404498753978,0.3090110819121666);
\draw [color=zzttqq] (1.5514404498753978,0.3090110819121666)-- (1.0314404498753975,0.3090110819121666);
\draw [color=zzttqq] (2.0314404498753977,0.3090110819121666)-- (2.0314404498753977,-0.6909889180878335);
\draw [color=zzttqq] (2.0314404498753977,-0.6909889180878335)-- (2.5514404498753973,-0.6909889180878335);
\draw [color=zzttqq] (2.5514404498753973,-0.6909889180878335)-- (2.5514404498753973,0.3090110819121666);
\draw [color=zzttqq] (2.5514404498753973,0.3090110819121666)-- (2.0314404498753977,0.3090110819121666);
\draw [color=zzttqq] (3.031440449875397,0.3090110819121666)-- (3.031440449875397,-0.6909889180878335);
\draw [color=zzttqq] (3.031440449875397,-0.6909889180878335)-- (3.5514404498753964,-0.6909889180878335);
\draw [color=zzttqq] (3.5514404498753964,-0.6909889180878335)-- (3.5514404498753964,0.3090110819121666);
\draw [color=zzttqq] (3.5514404498753964,0.3090110819121666)-- (3.031440449875397,0.3090110819121666);
\draw [color=zzttqq] (4.031440449875397,0.3090110819121666)-- (4.031440449875397,-0.6909889180878335);
\draw [color=zzttqq] (4.031440449875397,-0.6909889180878335)-- (4.551440449875397,-0.6909889180878335);
\draw [color=zzttqq] (4.551440449875397,-0.6909889180878335)-- (4.551440449875397,0.3090110819121666);
\draw [color=zzttqq] (4.551440449875397,0.3090110819121666)-- (4.031440449875397,0.3090110819121666);
\draw [color=zzttqq] (5.031440449875397,0.3090110819121666)-- (5.031440449875397,-0.6909889180878335);
\draw [color=zzttqq] (5.031440449875397,-0.6909889180878335)-- (5.551440449875397,-0.6909889180878335);
\draw [color=zzttqq] (5.551440449875397,-0.6909889180878335)-- (5.551440449875397,0.3090110819121666);
\draw [color=zzttqq] (5.551440449875397,0.3090110819121666)-- (5.031440449875397,0.3090110819121666);
\draw (1.5514404498753978,-0.0382074472361662)-- (2.0314404498753977,-0.031552718960272175);
\draw (1.5514404498753978,-0.2240614486960495)-- (2.0314404498753977,-0.23493198070170757);
\draw (1.5514404498753978,-0.4223057169199249)-- (2.0314404498753977,-0.4178927777931813);
\draw (2.5514404498753973,-0.0382074472361662)-- (3.031440449875397,-0.031552718960272175);
\draw (2.5514404498753973,-0.2240614486960495)-- (3.031440449875397,-0.23493198070170757);
\draw (2.5514404498753973,-0.4223057169199249)-- (3.031440449875397,-0.4178927777931813);
\draw (3.5514404498753964,-0.0382074472361662)-- (4.031440449875397,-0.031552718960272175);
\draw (3.5514404498753964,-0.2240614486960495)-- (4.031440449875397,-0.23493198070170757);
\draw (3.5514404498753964,-0.4223057169199249)-- (4.031440449875397,-0.4178927777931813);
\draw (4.551440449875397,-0.0382074472361662)-- (5.031440449875397,-0.031552718960272175);
\draw (4.551440449875397,-0.2240614486960495)-- (5.031440449875397,-0.23493198070170757);
\draw (4.551440449875397,-0.4223057169199249)-- (5.031440449875397,-0.4178927777931813);
\draw [color=zzttqq] (1.0314404498753975,-1.3596193363773281)-- (1.0314404498753975,-2.35961933637733);
\draw [color=zzttqq] (1.0314404498753975,-2.35961933637733)-- (1.5514404498753975,-2.35961933637733);
\draw [color=zzttqq] (1.5514404498753975,-2.35961933637733)-- (1.5514404498753975,-1.3596193363773281);
\draw [color=zzttqq] (1.5514404498753975,-1.3596193363773281)-- (1.0314404498753975,-1.3596193363773281);
\draw [color=zzttqq] (2.031440449875398,-1.3596193363773281)-- (2.031440449875398,-2.35961933637733);
\draw [color=zzttqq] (2.031440449875398,-2.35961933637733)-- (2.5514404498753986,-2.35961933637733);
\draw [color=zzttqq] (2.5514404498753986,-2.35961933637733)-- (2.5514404498753986,-1.3596193363773281);
\draw [color=zzttqq] (2.5514404498753986,-1.3596193363773281)-- (2.031440449875398,-1.3596193363773281);
\draw [color=zzttqq] (3.031440449875398,-1.3596193363773281)-- (3.031440449875398,-2.35961933637733);
\draw [color=zzttqq] (3.031440449875398,-2.35961933637733)-- (3.5514404498753986,-2.35961933637733);
\draw [color=zzttqq] (3.5514404498753986,-2.35961933637733)-- (3.5514404498753986,-1.3596193363773281);
\draw [color=zzttqq] (3.5514404498753986,-1.3596193363773281)-- (3.031440449875398,-1.3596193363773281);
\draw [color=zzttqq] (4.031440449875398,-1.3596193363773281)-- (4.031440449875398,-2.35961933637733);
\draw [color=zzttqq] (4.031440449875398,-2.35961933637733)-- (4.551440449875398,-2.35961933637733);
\draw [color=zzttqq] (4.551440449875398,-2.35961933637733)-- (4.551440449875398,-1.3596193363773281);
\draw [color=zzttqq] (4.551440449875398,-1.3596193363773281)-- (4.031440449875398,-1.3596193363773281);
\draw [color=zzttqq] (5.031440449875398,-1.3596193363773281)-- (5.031440449875398,-2.35961933637733);
\draw [color=zzttqq] (5.031440449875398,-2.35961933637733)-- (5.551440449875398,-2.35961933637733);
\draw [color=zzttqq] (5.551440449875398,-2.35961933637733)-- (5.551440449875398,-1.3596193363773281);
\draw [color=zzttqq] (5.551440449875398,-1.3596193363773281)-- (5.031440449875398,-1.3596193363773281);
\draw (1.5514404498753975,-1.7068378655256615)-- (2.031440449875398,-1.7001831372497676);
\draw (1.5514404498753975,-1.892691866985545)-- (2.031440449875398,-1.9035623989912032);
\draw (1.5514404498753975,-2.0909361352094207)-- (2.031440449875398,-2.0865231960826773);
\draw (2.5514404498753986,-1.7068378655256615)-- (3.031440449875398,-1.7001831372497676);
\draw (2.5514404498753986,-1.892691866985545)-- (3.031440449875398,-1.9035623989912032);
\draw (2.5514404498753986,-2.0909361352094207)-- (3.031440449875398,-2.0865231960826773);
\draw (3.5514404498753986,-1.7068378655256615)-- (4.031440449875398,-1.7001831372497676);
\draw (3.5514404498753986,-1.892691866985545)-- (4.031440449875398,-1.9035623989912032);
\draw (3.5514404498753986,-2.0909361352094207)-- (4.031440449875398,-2.0865231960826773);
\draw (4.551440449875398,-1.7068378655256615)-- (5.031440449875398,-1.7001831372497676);
\draw (4.551440449875398,-1.892691866985545)-- (5.031440449875398,-1.9035623989912032);
\draw (4.551440449875398,-2.0909361352094207)-- (5.031440449875398,-2.0865231960826773);
\draw (5.52,1.6046332341991933)-- (7.0,0.0);
\draw (5.52,1.3883498735046529)-- (7.0,0.0);
\draw (5.52,1.2169845602862757)-- (7.0,0.0);
\draw (5.551440449875397,-0.07183213275774891)-- (7.0,0.0);
\draw (5.551440449875397,-0.23293421938825198)-- (7.0,0.0);
\draw (5.551440449875397,-0.4262567233448557)-- (7.0,0.0);
\draw (5.551440449875398,-1.6935931381714797)-- (7.0,0.0);
\draw (5.551440449875398,-1.9191360594541842)-- (7.0,0.0);
\draw (5.551440449875398,-2.1017184243020877)-- (7.0,0.0);
\draw (1.0,1.638304014372265)-- (-0.2790882644569448,-0.15355154713478716);
\draw (1.0,1.3932639375849758)-- (-0.2790882644569448,-0.15355154713478716);
\draw (1.0,1.2382685200987213)-- (-0.2790882644569448,-0.15355154713478716);
\draw (1.0314404498753975,-0.08175573979215783)-- (-0.2790882644569448,-0.15355154713478716);
\draw (1.0314404498753975,-0.21645209284224748)-- (-0.2790882644569448,-0.15355154713478716);
\draw (1.0314404498753975,-0.4184966224173821)-- (-0.2790882644569448,-0.15355154713478716);
\draw (-0.2790882644569448,-0.15355154713478716)-- (1.0314404498753975,-1.6642810082969484);
\draw (-0.2790882644569448,-0.15355154713478716)-- (1.0314404498753975,-1.8749713406894417);
\draw (-0.2790882644569448,-0.15355154713478716)-- (1.0314404498753975,-2.0931863278102387);
\draw (1.2,-3.6) node[anchor=north west] {path of length r-1};
\draw (3.0478563656023256,0.015937049915998376) node[anchor=north west] {$\hat{B}^{x}_i$};
\draw (3.7675342193392893,3.3619717465112395) node[anchor=north west] {$c_j^\ell$};
\draw (-0.90,-0.15) node[anchor=north west] {$h_1$};
\draw (6.74,-0.1) node[anchor=north west] {$h_2$};

\draw (3.5617817579441997,3.105476877029485)-- (3.2401934818667044,1.7074478671531645);
\draw (3.5617817579441997,3.105476877029485)-- (3.220870958274342,0.09401714719090909);
\draw (3.5617817579441997,3.105476877029485)-- (3.2788385290514293,-0.28277206286015655);
\draw (3.5617817579441997,3.105476877029485)-- (3.2953627409658535,-1.823274937069931);
\begin{scriptsize}
\draw [fill=xdxdff] (5.52,1.6046332341991933) circle (1.5pt);
\draw [fill=xdxdff] (5.52,1.3883498735046529) circle (1.5pt);
\draw [fill=xdxdff] (5.52,1.2169845602862757) circle (1.5pt);
\draw [fill=xdxdff] (5.551440449875397,-0.07183213275774891) circle (1.5pt);
\draw [fill=xdxdff] (5.551440449875397,-0.23293421938825198) circle (1.5pt);
\draw [fill=xdxdff] (5.551440449875397,-0.4262567233448557) circle (1.5pt);
\draw [fill=xdxdff] (5.551440449875398,-1.6935931381714797) circle (1.5pt);
\draw [fill=xdxdff] (5.551440449875398,-1.9191360594541842) circle (1.5pt);
\draw [fill=xdxdff] (5.551440449875398,-2.1017184243020877) circle (1.5pt);
\draw [fill=xdxdff] (1.0,1.638304014372265) circle (1.5pt);
\draw [fill=xdxdff] (1.0,1.3932639375849758) circle (1.5pt);
\draw [fill=xdxdff] (1.0,1.2382685200987213) circle (1.5pt);
\draw [fill=xdxdff] (1.0314404498753975,-0.08175573979215783) circle (1.5pt);
\draw [fill=xdxdff] (1.0314404498753975,-0.21645209284224748) circle (1.5pt);
\draw [fill=xdxdff] (1.0314404498753975,-0.4184966224173821) circle (1.5pt);
\draw [fill=xdxdff] (1.0314404498753975,-1.6642810082969484) circle (1.5pt);
\draw [fill=xdxdff] (1.0314404498753975,-1.8749713406894417) circle (1.5pt);
\draw [fill=xdxdff] (1.0314404498753975,-2.0931863278102387) circle (1.5pt);
\draw [fill=qqqqff] (3.5617817579441997,3.105476877029485) circle (1.5pt);
\draw [fill=qqqqff] (3.2401934818667044,1.7074478671531645) circle (1.5pt);
\draw [fill=qqqqff] (3.220870958274342,0.09401714719090909) circle (1.5pt);
\draw [fill=qqqqff] (3.2788385290514293,-0.28277206286015655) circle (1.5pt);
\draw [fill=qqqqff] (3.2953627409658535,-1.823274937069931) circle (1.5pt);
\draw [fill=qqqqff] (-1.29,-0.15) circle (1.5pt);
\draw [fill=qqqqff] (-0.28,-0.15) circle (1.5pt);
\draw [fill=qqqqff] (7,0) circle (1.5pt);
\draw [fill=qqqqff] (7.8,0) circle (1.5pt);
\draw [fill=qqqqff] (4.3, 3.9) circle (1.5pt);
\draw [fill=qqqqff] (0, -4) circle (1.5pt);
\draw [fill=qqqqff] (0.8, -4) circle (1.5pt);

\end{scriptsize}
\draw[thick]  (-0.28,-0.15) -- (-1.29,-0.15);
\draw[snake=snake=expanding waves]  (-1.29,-0.15) -- (-2.09,-0.15);
\draw [fill=qqqqff] (-2.09,-0.15) circle (1.5pt);

\draw[thick]  (7,0) -- (7.8,0);
\draw[snake=snake=expanding waves]  (7.8,0) -- (8.6,0);
\draw [fill=qqqqff] (8.6,0) circle (1.5pt);

\draw[snake=snake=expanding waves]  (0,-4) -- (0.8,-4);
\draw[snake=snake=expanding waves] (3.5617817579441997,3.105476877029485) -- (4.3, 3.9);

\end{tikzpicture}

%% file: figs/s-consecutive-gadgets.tex
\definecolor{qqqqff}{rgb}{0.0,0.0,1.0}
\begin{tikzpicture}[line cap=round,line join=round,>=triangle 45,x=0.5cm,y=0.5cm]
\clip(0,2) rectangle (15,10.5);
\draw (1.0907928159748779,6.934887148264213)-- (2.090792815974877,6.934887148264213);
\draw (2.090792815974877,6.934887148264213)-- (3.0907928159748805,6.934887148264213);
\draw (2.090792815974877,6.934887148264213)-- (2.090792815974877,5.934887148264213);
\draw (2.090792815974877,6.934887148264213)-- (2.090792815974877,7.934887148264213);
\draw (2.090792815974877,7.934887148264213)-- (3.0907928159748805,6.934887148264213);
\draw (2.090792815974877,7.934887148264213)-- (1.0907928159748779,6.934887148264213);
\draw (1.0907928159748779,6.934887148264213)-- (2.090792815974877,5.934887148264213);
\draw (3.0907928159748805,6.934887148264213)-- (2.090792815974877,5.934887148264213);
\draw (3.0907928159748805,6.934887148264213)-- (4.090792815974881,6.934887148264213);
\draw (4.090792815974881,6.934887148264213)-- (5.0907928159748845,6.934887148264213);
\draw (4.090792815974881,6.934887148264213)-- (4.090792815974881,5.934887148264213);
\draw (4.090792815974881,6.934887148264213)-- (4.090792815974881,7.934887148264213);
\draw (4.090792815974881,7.934887148264213)-- (5.0907928159748845,6.934887148264213);
\draw (4.090792815974881,7.934887148264213)-- (3.0907928159748805,6.934887148264213);
\draw (3.0907928159748805,6.934887148264213)-- (4.090792815974881,5.934887148264213);
\draw (5.0907928159748845,6.934887148264213)-- (4.090792815974881,5.934887148264213);
\draw (5.0907928159748845,6.934887148264213)-- (6.090792815974884,6.934887148264213);
\draw (6.090792815974884,6.934887148264213)-- (7.090792815974886,6.934887148264213);
\draw (6.090792815974884,6.934887148264213)-- (6.090792815974884,5.934887148264213);
\draw (6.090792815974884,6.934887148264213)-- (6.090792815974884,7.934887148264213);
\draw (6.090792815974884,7.934887148264213)-- (7.090792815974886,6.934887148264213);
\draw (6.090792815974884,7.934887148264213)-- (5.0907928159748845,6.934887148264213);
\draw (5.0907928159748845,6.934887148264213)-- (6.090792815974884,5.934887148264213);
\draw (7.090792815974886,6.934887148264213)-- (6.090792815974884,5.934887148264213);
\draw (3.090792815974877,8.934887148264213)-- (4.090792815974881,7.934887148264213);
\draw (3.090792815974877,8.934887148264213)-- (2.090792815974877,7.934887148264213);
\draw (2.090792815974877,7.934887148264213)-- (3.0907928159748805,6.934887148264213);
\draw (4.090792815974881,7.934887148264213)-- (3.0907928159748805,6.934887148264213);
\draw (5.09079281597488,8.934887148264213)-- (6.090792815974884,7.934887148264213);
\draw (5.09079281597488,8.934887148264213)-- (4.090792815974881,7.934887148264213);
\draw (4.090792815974881,7.934887148264213)-- (5.0907928159748845,6.934887148264213);
\draw (6.090792815974884,7.934887148264213)-- (5.0907928159748845,6.934887148264213);
\draw (4.090792815974877,9.934887148264213)-- (5.09079281597488,8.934887148264213);
\draw (4.090792815974877,9.934887148264213)-- (3.090792815974877,8.934887148264213);
\draw (3.090792815974877,8.934887148264213)-- (4.090792815974881,7.934887148264213);
\draw (5.09079281597488,8.934887148264213)-- (4.090792815974881,7.934887148264213);
\draw (3.0907928159748805,6.934887148264213)-- (4.090792815974881,5.934887148264213);
\draw (3.0907928159748805,6.934887148264213)-- (2.090792815974877,5.934887148264213);
\draw (2.090792815974877,5.934887148264213)-- (3.090792815974876,4.934887148264213);
\draw (4.090792815974881,5.934887148264213)-- (3.090792815974876,4.934887148264213);
\draw (5.0907928159748845,6.934887148264213)-- (6.090792815974884,5.934887148264213);
\draw (5.0907928159748845,6.934887148264213)-- (4.090792815974881,5.934887148264213);
\draw (4.090792815974881,5.934887148264213)-- (5.09079281597488,4.934887148264213);
\draw (6.090792815974884,5.934887148264213)-- (5.09079281597488,4.934887148264213);
\draw (4.090792815974881,5.934887148264213)-- (5.09079281597488,4.934887148264213);
\draw (4.090792815974881,5.934887148264213)-- (3.090792815974876,4.934887148264213);
\draw (3.090792815974876,4.934887148264213)-- (4.090792815974876,3.934887148264213);
\draw (5.09079281597488,4.934887148264213)-- (4.090792815974876,3.934887148264213);
\draw(5.0907928159748845,6.934887148264213) circle (0.3cm);
\draw (7.090792815974886,6.934887148264213)-- (7.930468581667604,6.951568056152041);
\draw (7.930468581667604,6.951568056152041)-- (8.930468581667604,6.951568056152041);
\draw (8.930468581667604,6.951568056152041)-- (9.930468581667608,6.951568056152041);
\draw (8.930468581667604,6.951568056152041)-- (8.930468581667604,5.951568056152041);
\draw (8.930468581667604,6.951568056152041)-- (8.930468581667604,7.951568056152041);
\draw (8.930468581667604,7.951568056152041)-- (9.930468581667608,6.951568056152041);
\draw (8.930468581667604,7.951568056152041)-- (7.930468581667604,6.951568056152041);
\draw (7.930468581667604,6.951568056152041)-- (8.930468581667604,5.951568056152041);
\draw (9.930468581667608,6.951568056152041)-- (8.930468581667604,5.951568056152041);
\draw (9.930468581667608,6.951568056152041)-- (10.930468581667608,6.951568056152041);
\draw (10.930468581667608,6.951568056152041)-- (11.930468581667611,6.951568056152041);
\draw (10.930468581667608,6.951568056152041)-- (10.930468581667608,5.951568056152041);
\draw (10.930468581667608,6.951568056152041)-- (10.930468581667608,7.951568056152041);
\draw (10.930468581667608,7.951568056152041)-- (11.930468581667611,6.951568056152041);
\draw (10.930468581667608,7.951568056152041)-- (9.930468581667608,6.951568056152041);
\draw (9.930468581667608,6.951568056152041)-- (10.930468581667608,5.951568056152041);
\draw (11.930468581667611,6.951568056152041)-- (10.930468581667608,5.951568056152041);
\draw (11.930468581667611,6.951568056152041)-- (12.930468581667611,6.951568056152041);
\draw (12.930468581667611,6.951568056152041)-- (13.930468581667615,6.951568056152041);
\draw (12.930468581667611,6.951568056152041)-- (12.930468581667611,5.951568056152041);
\draw (12.930468581667611,6.951568056152041)-- (12.930468581667611,7.951568056152041);
\draw (12.930468581667611,7.951568056152041)-- (13.930468581667615,6.951568056152041);
\draw (12.930468581667611,7.951568056152041)-- (11.930468581667611,6.951568056152041);
\draw (11.930468581667611,6.951568056152041)-- (12.930468581667611,5.951568056152041);
\draw (13.930468581667615,6.951568056152041)-- (12.930468581667611,5.951568056152041);
\draw (9.930468581667604,8.951568056152041)-- (10.930468581667608,7.951568056152041);
\draw (9.930468581667604,8.951568056152041)-- (8.930468581667604,7.951568056152041);
\draw (8.930468581667604,7.951568056152041)-- (9.930468581667608,6.951568056152041);
\draw (10.930468581667608,7.951568056152041)-- (9.930468581667608,6.951568056152041);
\draw (11.930468581667608,8.951568056152041)-- (12.930468581667611,7.951568056152041);
\draw (11.930468581667608,8.951568056152041)-- (10.930468581667608,7.951568056152041);
\draw (10.930468581667608,7.951568056152041)-- (11.930468581667611,6.951568056152041);
\draw (12.930468581667611,7.951568056152041)-- (11.930468581667611,6.951568056152041);
\draw (10.930468581667604,9.951568056152045)-- (11.930468581667608,8.951568056152041);
\draw (10.930468581667604,9.951568056152045)-- (9.930468581667604,8.951568056152041);
\draw (9.930468581667604,8.951568056152041)-- (10.930468581667608,7.951568056152041);
\draw (11.930468581667608,8.951568056152041)-- (10.930468581667608,7.951568056152041);
\draw (9.930468581667608,6.951568056152041)-- (10.930468581667608,5.951568056152041);
\draw (9.930468581667608,6.951568056152041)-- (8.930468581667604,5.951568056152041);
\draw (8.930468581667604,5.951568056152041)-- (9.930468581667604,4.951568056152042);
\draw (10.930468581667608,5.951568056152041)-- (9.930468581667604,4.951568056152042);
\draw (11.930468581667611,6.951568056152041)-- (12.930468581667611,5.951568056152041);
\draw (11.930468581667611,6.951568056152041)-- (10.930468581667608,5.951568056152041);
\draw (10.930468581667608,5.951568056152041)-- (11.930468581667608,4.951568056152042);
\draw (12.930468581667611,5.951568056152041)-- (11.930468581667608,4.951568056152042);
\draw (10.930468581667608,5.951568056152041)-- (11.930468581667608,4.951568056152042);
\draw (10.930468581667608,5.951568056152041)-- (9.930468581667604,4.951568056152042);
\draw (9.930468581667604,4.951568056152042)-- (10.930468581667604,3.9515680561520417);
\draw (11.930468581667608,4.951568056152042)-- (10.930468581667604,3.9515680561520417);
\draw(11.930468581667611,6.951568056152041) circle (0.3cm);
\draw (0.2859607146895809,7.046859768193645) node[anchor=north west] {$P_i$};
\draw (7.770493642803446,6.656183598869021) node[anchor=north west] {$P_i$};
\draw (2.465522501448014,3.448526629677367) node[anchor=north west] {(a) a path in  $\hat{B}^j_i$};
\draw (8.61353169239869,3.407402822380038) node[anchor=north west] {(b) a path in  $\hat{B}^{j+1}_i$};
\draw (4.19,10.53) node[anchor=north west] {$g_1$};
\draw (4.19,4.03) node[anchor=north west] {$g_2$};
\draw (11.03,10.53) node[anchor=north west] {$g_1$};
\draw (11.03,4.03) node[anchor=north west] {$g_2$};
\begin{scriptsize}
\draw [fill=qqqqff] (1.0907928159748779,6.934887148264213) circle (1.5pt);
\draw [fill=qqqqff] (2.090792815974877,6.934887148264213) circle (1.5pt);
\draw [fill=qqqqff] (3.0907928159748805,6.934887148264213) circle (1.5pt);
\draw [fill=qqqqff] (2.090792815974877,5.934887148264213) circle (1.5pt);
\draw [fill=qqqqff] (2.090792815974877,7.934887148264213) circle (1.5pt);
\draw [fill=qqqqff] (4.090792815974881,6.934887148264213) circle (1.5pt);
\draw [fill=qqqqff] (5.0907928159748845,6.934887148264213) circle (1.5pt);
\draw [fill=qqqqff] (4.090792815974881,5.934887148264213) circle (1.5pt);
\draw [fill=qqqqff] (4.090792815974881,7.934887148264213) circle (1.5pt);
\draw [fill=qqqqff] (6.090792815974884,6.934887148264213) circle (1.5pt);
\draw [fill=qqqqff] (7.090792815974886,6.934887148264213) circle (1.5pt);
\draw [fill=qqqqff] (6.090792815974884,5.934887148264213) circle (1.5pt);
\draw [fill=qqqqff] (6.090792815974884,7.934887148264213) circle (1.5pt);
\draw [fill=qqqqff] (3.090792815974877,8.934887148264213) circle (1.5pt);
\draw [fill=qqqqff] (5.09079281597488,8.934887148264213) circle (1.5pt);
\draw [fill=qqqqff] (4.090792815974877,9.934887148264213) circle (1.5pt);
\draw [fill=qqqqff] (3.090792815974876,4.934887148264213) circle (1.5pt);
\draw [fill=qqqqff] (5.09079281597488,4.934887148264213) circle (1.5pt);
\draw [fill=qqqqff] (4.090792815974876,3.934887148264213) circle (1.5pt);
\draw [fill=qqqqff] (7.930468581667604,6.951568056152041) circle (1.5pt);
\draw [fill=qqqqff] (8.930468581667604,6.951568056152041) circle (1.5pt);
\draw [fill=qqqqff] (9.930468581667608,6.951568056152041) circle (1.5pt);
\draw [fill=qqqqff] (8.930468581667604,5.951568056152041) circle (1.5pt);
\draw [fill=qqqqff] (8.930468581667604,7.951568056152041) circle (1.5pt);
\draw [fill=qqqqff] (10.930468581667608,6.951568056152041) circle (1.5pt);
\draw [fill=qqqqff] (11.930468581667611,6.951568056152041) circle (1.5pt);
\draw [fill=qqqqff] (10.930468581667608,5.951568056152041) circle (1.5pt);
\draw [fill=qqqqff] (10.930468581667608,7.951568056152041) circle (1.5pt);
\draw [fill=qqqqff] (12.930468581667611,6.951568056152041) circle (1.5pt);
\draw [fill=qqqqff] (13.930468581667615,6.951568056152041) circle (1.5pt);
\draw [fill=qqqqff] (12.930468581667611,5.951568056152041) circle (1.5pt);
\draw [fill=qqqqff] (12.930468581667611,7.951568056152041) circle (1.5pt);
\draw [fill=qqqqff] (9.930468581667604,8.951568056152041) circle (1.5pt);
\draw [fill=qqqqff] (11.930468581667608,8.951568056152041) circle (1.5pt);
\draw [fill=qqqqff] (10.930468581667604,9.951568056152045) circle (1.5pt);
\draw [fill=qqqqff] (9.930468581667604,4.951568056152042) circle (1.5pt);
\draw [fill=qqqqff] (11.930468581667608,4.951568056152042) circle (1.5pt);
\draw [fill=qqqqff] (10.930468581667604,3.9515680561520417) circle (1.5pt);
\end{scriptsize}
\end{tikzpicture}

%% file: figs/s-r-core-gadget.tex
\definecolor{xdxdff}{rgb}{0.49019607843137253,0.49019607843137253,1.0}
\definecolor{qqqqff}{rgb}{0.0,0.0,1.0}
\begin{tikzpicture}[line cap=round,line join=round,>=triangle 45,x=2.0cm,y=2.0cm]
\clip(1.5,0.5) rectangle (8.0,3.0);
\draw (2.0,1.0)-- (3.0,1.0);
\draw (3.0,1.0)-- (4.0,1.0);
\draw (2.5,2.0)-- (3.5,2.0);
\draw (2.5,2.0)-- (2.0,1.0);
\draw (2.5,2.0)-- (3.0,1.0);
\draw (3.5,2.0)-- (3.0,1.0);
\draw (3.5,2.0)-- (4.0,1.0);
\draw (4.0,1.0)-- (5.0,1.0);
\draw (3.5,2.0)-- (4.5,2.0);
\draw (3.5,2.0)-- (4.0,1.0);
\draw (4.5,2.0)-- (4.0,1.0);
\draw (4.5,2.0)-- (5.0,1.0);
\draw (6.5,1.0)-- (7.5,1.0);
\draw (6.0,2.0)-- (7.0,2.0);
\draw (6.0,2.0)-- (6.5,1.0);
\draw (7.0,2.0)-- (6.5,1.0);
\draw (7.0,2.0)-- (7.5,1.0);
\draw[dashed] (2.5,2.5)-- (2.5,2.0);
\draw[dashed] (3.5,2.5)-- (3.5,2.0);
\draw[dashed] (4.5,2.5)-- (4.5,2.0);
\draw[dashed] (6.0,2.5)-- (6.0,2.0);
\draw[dashed] (7.0,2.5)-- (7.0,2.0);
\draw (1.9723189990216796,1.59951203345955077) node[anchor=north west] {$s_1$};
\draw (2.479337890704167,1.5995120334595507) node[anchor=north west] {$s_2$};
\draw (2.9035438634606364,1.5995120334595507) node[anchor=north west] {$s_3$};
\draw (3.5019692146061323,1.5995120334595507) node[anchor=north west] {$s_4$};
\draw (3.917581646825611,1.5995120334595507) node[anchor=north west] {$s_5$};
\draw (4.498819916897124,1.5995120334595507) node[anchor=north west] {$s_6$};
\draw (6.947978969939647,1.5995120334595507) node[anchor=north west] {$s_{2r+2}$};
\draw (6.432366537720168,1.5995120334595507) node[anchor=north west] {$s_{2r+1}$};
\draw (6.002689510870603,1.5995120334595507) node[anchor=north west] {$s_{2r}$};

\draw (1.8348223504298191,1.0151512769441415) node[anchor=north west] {$a_1$};
\draw (2.1644993772793844,2.2698081953448725) node[anchor=north west] {$a_2$};
\draw (2.917608458090724,0.9807771147961764) node[anchor=north west] {$a_3$};
\draw (3.1699436201073677,2.2955888169558465) node[anchor=north west] {$a_4$};
\draw (3.880084998233751,0.9721835742591851) node[anchor=north west] {$a_5$};
\draw (4.1839814034723425,2.2698081953448725) node[anchor=north west] {$a_6$};
\draw (4.911309862672708,1.0065577364071503) node[anchor=north west] {$a_7$};
\draw (5.501915592075908,2.2698081953448725) node[anchor=north west] {$a_{2r}$};
\draw (6.217528024295387,1.1010866823140546) node[anchor=north west] {$a_{2r+1}$};
\draw (6.864392132218935,2.3213694385668204) node[anchor=north west] {$a_{2r+2}$};
\draw (7.308907672493283,1.0667125201660894) node[anchor=north west] {$a_{2r+3}$};

\draw[dotted] (4.5,2) -- (6,2);
\draw[dotted] (5,1) -- (6.5,1);
\begin{scriptsize}
\draw [fill=qqqqff] (2.0,1.0) circle (1.5pt);
\draw [fill=qqqqff] (3.0,1.0) circle (1.5pt);
\draw [fill=qqqqff] (4.0,1.0) circle (1.5pt);
\draw [fill=qqqqff] (2.5,2.0) circle (1.5pt);
\draw [fill=qqqqff] (3.5,2.0) circle (1.5pt);
\draw [fill=xdxdff] (3.0,2.0) circle (1.5pt);
\draw [fill=xdxdff] (2.5,1.0) circle (1.5pt);
\draw [fill=xdxdff] (3.5,1.0) circle (1.5pt);
\draw [fill=qqqqff] (5.0,1.0) circle (1.5pt);
\draw [fill=qqqqff] (4.5,2.0) circle (1.5pt);
\draw [fill=xdxdff] (4.0,2.0) circle (1.5pt);
\draw [fill=xdxdff] (3.5,1.0) circle (1.5pt);
\draw [fill=xdxdff] (4.5,1.0) circle (1.5pt);
\draw [fill=qqqqff] (6.5,1.0) circle (1.5pt);
\draw [fill=qqqqff] (7.5,1.0) circle (1.5pt);
\draw [fill=qqqqff] (6.0,2.0) circle (1.5pt);
\draw [fill=qqqqff] (7.0,2.0) circle (1.5pt);
\draw [fill=xdxdff] (6.5,2.0) circle (1.5pt);
\draw [fill=xdxdff] (7.0,1.0) circle (1.5pt);
\end{scriptsize}
\end{tikzpicture}

%% file: figs/s-pattern.tex
\definecolor{uuuuuu}{rgb}{0.26666666666666666,0.26666666666666666,0.26666666666666666}
\definecolor{zzttqq}{rgb}{0.6,0.2,0.0}
\definecolor{qqqqff}{rgb}{0.0,0.0,1.0}
\begin{tikzpicture}[line cap=round,line join=round,>=triangle 45,x=1.0cm,y=1.0cm]
\clip(1,0) rectangle (20,7.5);
\fill[color=zzttqq,fill=zzttqq,fill opacity=0.1] (16.0,5.1) -- (16.0,3.6) -- (17.5,3.5999999999999996) -- (17.5,5.1) -- cycle;
\draw (4.749245739207646,4.123662981011211)-- (4.074137296031698,3.5753415578018797);
\draw (2.0,2.0)-- (2.706291042054311,2.547360456769595);
\draw (3.0,2.0)-- (3.3776445369211,2.486498147377332);
\draw (4.749245739207646,4.123662981011211)-- (4.321229165570373,3.5686633991656995);
\draw (4.749245739207646,4.123662981011211)-- (4.561847439267197,3.5728710256385616);
\draw (4.749245739207646,4.123662981011211)-- (4.815412904647723,3.5886978750742404);
\draw (5.0,2.0)-- (4.9466948273879625,2.428730142036299);
\draw (4.0,2.0)-- (4.136555312772506,2.440851415333135);
\draw (4.749245739207646,4.123662981011211)-- (5.276205850544171,3.6087323509827813);
\draw (7.0,2.0)-- (6.5406165337669195,2.4712825700292664);
\draw (1.7781408238222904,2.045246404283426) node[anchor=north west] {$p_1$};
\draw (2.85844681553497,1.9843840948911629) node[anchor=north west] {$p_2$};
\draw (3.893106075203452,1.9843840948911629) node[anchor=north west] {$p_3$};
\draw (4.912549757523868,1.9995996722392286) node[anchor=north west] {$p_4$};
\draw (6.9210059674685676,1.9083062081508342) node[anchor=north west] {$p_m$};
\draw (4.858196489568065,4.51236742914522) node[anchor=north west] {$h$};
\draw (4.5625914785183515,6.203494009255662) node[anchor=north west] {$r_P$};
\draw (11.749245739207646,4.123662981011207)-- (11.074137296031697,3.575341557801878);
\draw (9.0,2.0)-- (9.70629104205431,2.547360456769595);
\draw (10.0,2.0)-- (10.375776445369212,2.4864981473773327);
\draw (11.749245739207646,4.123662981011207)-- (11.321229165570372,3.5686633991656973);
\draw (11.749245739207646,4.123662981011207)-- (11.561847439267197,3.5728710256385594);
\draw (11.749245739207646,4.123662981011207)-- (11.815412904647722,3.588697875074239);
\draw (12.0,2.0)-- (11.946694827387962,2.428730142036299);
\draw (11.0,2.0)-- (11.136555312772506,2.4408514153331353);
\draw (11.749245739207646,4.123662981011207)-- (12.276205850544171,3.608732350982779);
\draw (14.0,2.0)-- (13.540616533766922,2.471282570029267);
\draw (8.777306403932608,2.045246404283426) node[anchor=north west] {$p_1$};
\draw (9.857612395645289,1.9843840948911629) node[anchor=north west] {$p_2$};
\draw (10.89227165531377,1.9843840948911629) node[anchor=north west] {$p_3$};
\draw (11.926930914982252,1.9995996722392286) node[anchor=north west] {$p_4$};
\draw (13.935387124926951,1.9083062081508342) node[anchor=north west] {$p_m$};
\draw (11.857362069678384,4.51236742914522) node[anchor=north west] {$h$};
\draw (11.56175705862867,6.203494009255662) node[anchor=north west] {$r_P$};
\draw [rotate around={2.4195092166563152:(11.668266100065097,3.582019716438063)}] (11.668266100065097,3.582019716438063) ellipse (1.1174208408654462cm and 0.2819302674838223cm);
\draw (2.1512716641501804,1.2692519595320735) node[anchor=north west] {$\mbox{(a) } P_r(m) \mbox{ when r is odd}$};
\draw (2.7519377740985087,3.8102533766590505) node[anchor=north west] {$\frac{r+1}{2}$};
\draw (5.049489953656461,5.392673420857887) node[anchor=north west] {$\frac{r-1}{2}$};
\draw (9.583732003380101,3.749391067266788) node[anchor=north west] {$\frac{r}{2} + 1$};
\draw (11.800978191225374,5.347026688813689) node[anchor=north west] {$\frac{r}{2} - 1$};
\draw (9.042313376308254,1.4518388877088622) node[anchor=north west] {$\mbox{(b) } P_r(m) \mbox{ when r is even}$};
\draw[dotted](4.76,5.64)--(4.75,4.12);
\draw[dotted] (2.71,2.55) -- (4.07,3.58);
\draw[dotted] (3.38,2.49) -- (4.32,3.57);
\draw[dotted] (4.14,2.44) -- (4.56,3.57);
\draw[dotted] (4.95,2.43) -- (4.82,3.57);
\draw[dotted] (6.54,2.47) -- (5.28,3.61);

\draw[dotted] (9.71,2.55)--(11.07,3.58);
\draw[dotted] (10.38,2.49)--(11.32,3.58);
\draw[dotted] (11.14,2.44) -- (11.56,3.57);
\draw[dotted] (11.95,2.43) -- (11.82,3.59);
\draw[dotted] (13.54,2.47)--(12.28,3.61);
\draw[dotted] (11.75,4.12) -- (11.76,5.64);

\draw [color=zzttqq] (16.0,5.1)-- (16.0,3.6);
\draw [color=zzttqq] (16.0,3.6)-- (17.5,3.5999999999999996);
\draw [color=zzttqq] (17.5,3.5999999999999996)-- (17.5,5.1);
\draw [color=zzttqq] (17.5,5.1)-- (16.0,5.1);
\draw (16.710070830410896,4.890796022876442)-- (16.72987847740728,5.741009794952475);
\draw (16.626627149699726,3.9002637039570796)-- (15.416229170541264,1.9714074361196023);
\draw (16.626627149699726,3.9002637039570796)-- (16.653724894400554,2.028522623374646);
\draw (16.626627149699726,3.9002637039570796)-- (17.491414307474532,1.9904458318712837);
\draw (16.626627149699726,3.9002637039570796)-- (18.931754340996136,2.075371344180775);
\draw (15.164840412896673,1.9786310052535157) node[anchor=north west] {$p_1$};
\draw (16.393100488356698,1.9984416516319028) node[anchor=north west] {$p_2$};
\draw (17.403443453654457,1.8993884197399666) node[anchor=north west] {$p_3$};
\draw (18.889241932033517,2.01825229801029) node[anchor=north west] {$p_m$};
\draw (16.21480467095121,4.851174730119667) node[anchor=north west] {$P_r(m)$};
\draw (14.679917218734747,1.4239329066586726) node[anchor=north west] {$\mbox{(c) } \mbox{shorthand for }P_r(m)$};
\draw (16.43272178111347,6.416215794012261) node[anchor=north west] {$r_P$};


\begin{scriptsize}
\draw [fill=qqqqff] (2.0,2.0) circle (1.5pt);
\draw [fill=qqqqff] (3.0,2.0) circle (1.5pt);
\draw [fill=qqqqff] (4.0,2.0) circle (1.5pt);
\draw [fill=qqqqff] (5.0,2.0) circle (1.5pt);
\draw [fill=qqqqff] (7.0,2.0) circle (1.5pt);
\draw [fill=qqqqff] (4.749245739207646,4.123662981011211) circle (1.5pt);
\draw [fill=qqqqff] (4.074137296031698,3.5753415578018797) circle (1.5pt);
\draw [fill=qqqqff] (2.706291042054311,2.547360456769595) circle (1.5pt);
\draw [fill=qqqqff] (3.375776445369211,2.486498147377332) circle (1.5pt);
\draw [fill=qqqqff] (4.321229165570373,3.5686633991656995) circle (1.5pt);
\draw [fill=qqqqff] (4.561847439267197,3.5728710256385616) circle (1.5pt);
\draw [fill=qqqqff] (4.815412904647723,3.5886978750742404) circle (1.5pt);
\draw [fill=qqqqff] (4.9466948273879625,2.428730142036299) circle (1.5pt);
\draw [fill=qqqqff] (4.136555312772506,2.440851415333135) circle (1.5pt);
\draw [fill=qqqqff] (5.276205850544171,3.6087323509827813) circle (1.5pt);
\draw [fill=qqqqff] (6.5406165337669195,2.4712825700292664) circle (1.5pt);
\draw [fill=qqqqff] (4.759471773707591,5.642725896882655) circle (1.5pt);
\draw [fill=qqqqff] (9.0,2.0) circle (1.5pt);
\draw [fill=qqqqff] (10.0,2.0) circle (1.5pt);
\draw [fill=qqqqff] (11.0,2.0) circle (1.5pt);
\draw [fill=qqqqff] (12.0,2.0) circle (1.5pt);
\draw [fill=qqqqff] (14.0,2.0) circle (1.5pt);
\draw [fill=qqqqff] (11.749245739207646,4.123662981011207) circle (1.5pt);
\draw [fill=qqqqff] (11.074137296031697,3.575341557801878) circle (1.5pt);
\draw [fill=qqqqff] (9.70629104205431,2.547360456769595) circle (1.5pt);
\draw [fill=qqqqff] (10.375776445369212,2.4864981473773327) circle (1.5pt);
\draw [fill=qqqqff] (11.321229165570372,3.5686633991656973) circle (1.5pt);
\draw [fill=qqqqff] (11.561847439267197,3.5728710256385594) circle (1.5pt);
\draw [fill=qqqqff] (11.815412904647722,3.588697875074239) circle (1.5pt);
\draw [fill=qqqqff] (11.946694827387962,2.428730142036299) circle (1.5pt);
\draw [fill=qqqqff] (11.136555312772506,2.4408514153331353) circle (1.5pt);
\draw [fill=qqqqff] (12.276205850544171,3.608732350982779) circle (1.5pt);
\draw [fill=qqqqff] (13.540616533766922,2.471282570029267) circle (1.5pt);
\draw [fill=qqqqff] (11.759471773707592,5.64272589688265) circle (1.5pt);
\draw [fill=uuuuuu] (3.3902141690430048,3.0613510072857375) circle (1.5pt);

\draw [fill=qqqqff] (16.0,5.1) circle (0.5pt);
\draw [fill=qqqqff] (16.0,3.6) circle (0.5pt);
\draw [fill=uuuuuu] (17.5,3.5999999999999996) circle (0.5pt);
\draw [fill=uuuuuu] (17.5,5.1) circle (0.5pt);
\draw [fill=qqqqff] (16.710070830410896,4.890796022876442) circle (1.5pt);
\draw [fill=qqqqff] (16.72987847740728,5.741009794952475) circle (1.5pt);
\draw [fill=qqqqff] (16.626627149699726,3.9002637039570796) circle (1.5pt);
\draw [fill=qqqqff] (15.416229170541264,1.9714074361196023) circle (1.5pt);
\draw [fill=qqqqff] (16.653724894400554,2.028522623374646) circle (1.5pt);
\draw [fill=qqqqff] (17.491414307474532,1.9904458318712837) circle (1.5pt);
\draw [fill=qqqqff] (18.931754340996136,2.075371344180775) circle (1.5pt);

\draw [fill=qqqqff] (5.420974565113976,2.4339286035853734) circle (1.5pt);
\draw [fill=qqqqff] (5.731939486995937,2.4339286035853734) circle (1.5pt);
\draw [fill=qqqqff] (6.080220199503733,2.4339286035853734) circle (1.5pt);
\draw [fill=qqqqff] (12.491676346383464,2.4334373785875995) circle (1.5pt);
\draw [fill=qqqqff] (12.802641268265425,2.4334373785875995) circle (1.5pt);
\draw [fill=qqqqff] (13.15092198077322,2.4334373785875995) circle (1.5pt);
\draw [fill=qqqqff] (17.593543448512737,2.3679037260975457) circle (1.5pt);
\draw [fill=qqqqff] (17.904508370394698,2.3679037260975457) circle (1.5pt);
\draw [fill=qqqqff] (18.252789082902495,2.3679037260975457) circle (1.5pt);

\end{scriptsize}
\end{tikzpicture}

%% file: figs/s-r-group-gadget.tex
\definecolor{qqqqcc}{rgb}{0.0,0.0,0.8}
\definecolor{uuuuuu}{rgb}{0.26666666666666666,0.26666666666666666,0.26666666666666666}
\definecolor{zzttqq}{rgb}{0.6,0.2,0.0}
\definecolor{ffqqqq}{rgb}{1.0,0.0,0.0}
\definecolor{xdxdff}{rgb}{0.49019607843137253,0.49019607843137253,1.0}
\definecolor{qqqqff}{rgb}{0.0,0.0,1.0}
\begin{tikzpicture}[line cap=round,line join=round,>=triangle 45,x=1.4cm,y=1.4cm,every node/.style={scale=0.7}]
\clip(1,-3) rectangle (13,6);
\fill[color=zzttqq,fill=zzttqq,fill opacity=0.1] (8.0,5.0) -- (8.0,4.5) -- (8.5,4.5) -- (8.5,5.0) -- cycle;
\fill[color=zzttqq,fill=zzttqq,fill opacity=0.1] (8.0,4.0) -- (8.0,3.5) -- (8.5,3.5) -- (8.5,4.0) -- cycle;
\fill[color=zzttqq,fill=zzttqq,fill opacity=0.1] (8.0,2.0) -- (8.0,1.5) -- (8.5,1.5) -- (8.5,2.0) -- cycle;
\draw (2.0,1.0)-- (3.0,1.0);
\draw (3.0,1.0)-- (4.0,1.0);
\draw (2.5,2.0)-- (3.5,2.0);
\draw (2.5,2.0)-- (2.0,1.0);
\draw (2.5,2.0)-- (3.0,1.0);
\draw (3.5,2.0)-- (3.0,1.0);
\draw (3.5,2.0)-- (4.5,2.0);
\draw [color=ffqqqq] (3.5,2.0)-- (4.0,1.0);
\draw (4.5,2.0)-- (4.0,1.0);
\draw (6.0,1.0)-- (7.0,1.0);
\draw (5.5,2.0)-- (6.5,2.0);
\draw (5.5,2.0)-- (6.0,1.0);
\draw (6.5,2.0)-- (6.0,1.0);
\draw (6.5,2.0)-- (7.0,1.0);
\draw (1.9700712429124096,1.6318579963671391) node[anchor=north west] {$s_1$};
\draw (2.4837186646689884,1.6318579963671391) node[anchor=north west] {$s_2$};
\draw (3.0085323347246233,1.6318579963671391) node[anchor=north west] {$s_3$};
\draw (3.49984725988309,1.6318579963671391) node[anchor=north west] {$s_4$};
\draw (4.013494681639669,1.6318579963671391) node[anchor=north west] {$s_5$};
\draw (6.447736810833891,1.6318579963671391) node[anchor=north west] {$s_{2r+2}$};
\draw (5.934089389077312,1.6318579963671391) node[anchor=north west] {$s_{2r+1}$};
\draw (5.498605705414125,1.6318579963671391) node[anchor=north west] {$s_{2r}$};
\draw (1.8360762633237369,1.0012130848162258) node[anchor=north west] {$a_1$};
\draw (2.260393698687867,2.301832894310501) node[anchor=north west] {$a_2$};
\draw (2.919202348332175,1.0065480916200018) node[anchor=north west] {$a_3$};
\draw (3.2653560456029127,2.324165390908613) node[anchor=north west] {$a_4$};
\draw (3.879499702050996,1.0065480916200018) node[anchor=north west] {$a_5$};
\draw (4.281484640817014,2.301832894310501) node[anchor=north west] {$a_6$};
\draw (5.1971170013396115,2.301832894310501) node[anchor=north west] {$a_{2r}$};
\draw (5.7219306713952465,1.029376822909618) node[anchor=north west] {$a_{2r+1}$};
\draw (6.168580603357489,2.357664135805781) node[anchor=north west] {$a_{2r+2}$};
\draw (6.805056756403684,0.96587807801245) node[anchor=north west] {$a_{2r+3}$};
\draw (2.0,3.0)-- (3.0,3.0);
\draw (3.0,3.0)-- (4.0,3.0);
\draw (2.5,4.0)-- (3.5,4.0);
\draw (2.5,4.0)-- (2.0,3.0);
\draw [color=ffqqqq] (2.5,4.0)-- (3.0,3.0);
\draw (3.5,4.0)-- (3.0,3.0);
\draw (3.5,4.0)-- (4.0,3.0);
\draw (3.5,4.0)-- (4.5,4.0);
\draw (3.5,4.0)-- (4.0,3.0);
\draw (4.5,4.0)-- (4.0,3.0);
\draw (1.9700712429124096,3.71994642829061) node[anchor=north west] {$s_1$};
\draw (2.4837186646689884,3.71994642829061) node[anchor=north west] {$s_2$};
\draw (3.0085323347246233,3.719946428290616) node[anchor=north west] {$s_3$};
\draw (3.49984725988309,3.71994642829061) node[anchor=north west] {$s_4$};
\draw (4.013494681639669,3.719946428290616) node[anchor=north west] {$s_5$};
\draw (2.260393698687867,4.30059133984153) node[anchor=north west] {$a_2$};
\draw (3.2653560456029127,4.322923836439641) node[anchor=north west] {$a_4$};
\draw (4.281484640817014,4.30059133984153) node[anchor=north west] {$a_6$};
\draw (6.0,3.0)-- (7.0,3.0);
\draw (5.5,4.0)-- (6.5,4.0);
\draw (5.5,4.0)-- (6.0,3.0);
\draw (6.5,4.0)-- (6.0,3.0);
\draw (6.5,4.0)-- (7.0,3.0);
\draw (6.447736810833891,3.71994642829061) node[anchor=north west] {$s_{2r+2}$};
\draw (5.934089389077312,3.719946428290616) node[anchor=north west] {$s_{2r+1}$};
\draw (5.498605705414125,3.71994642829061) node[anchor=north west] {$s_{2r}$};
\draw (5.7219306713952465,3.028135268440647) node[anchor=north west] {$a_{2r+1}$};
\draw (6.168580603357489,4.35642258133681) node[anchor=north west] {$a_{2r+2}$};
\draw (6.805056756403684,2.9546365235434787) node[anchor=north west] {$a_{2r+3}$};
\draw (2.0,-2.0)-- (3.0,-2.0);
\draw (3.0,-2.0)-- (4.0,-2.0);
\draw (2.5,-1.0)-- (3.5,-1.0);
\draw (2.5,-1.0)-- (2.0,-2.0);
\draw (2.5,-1.0)-- (3.0,-2.0);
\draw (3.5,-1.0)-- (3.0,-2.0);
\draw (3.5,-1.0)-- (4.5,-1.0);
\draw (3.5,-1.0)-- (4.0,-2.0);
\draw (4.5,-1.0)-- (4.0,-2.0);
\draw (6.0,-2.0)-- (7.0,-2.0);
\draw (5.5,-1.0)-- (6.5,-1.0);
\draw [color=ffqqqq] (5.5,-1.0)-- (6.0,-2.0);
\draw (6.5,-1.0)-- (6.0,-2.0);
\draw (6.5,-1.0)-- (7.0,-2.0);
\draw (1.9700712429124096,-1.2713665613874277) node[anchor=north west] {$s_1$};
\draw (2.4837186646689884,-1.2713665613874277) node[anchor=north west] {$s_2$};
\draw (3.0085323347246233,-1.2713665613874277) node[anchor=north west] {$s_3$};
\draw (3.49984725988309,-1.2713665613874277) node[anchor=north west] {$s_4$};
\draw (4.013494681639669,-1.2713665613874277) node[anchor=north west] {$s_5$};
\draw (6.447736810833891,-1.2713665613874277) node[anchor=north west] {$s_{2r+2}$};
\draw (5.934089389077312,-1.2713665613874277) node[anchor=north west] {$s_{2r+1}$};
\draw (5.498605705414125,-1.2713665613874277) node[anchor=north west] {$s_{2r}$};
\draw (1.8360762633237369,-2.0525077076298454) node[anchor=north west] {$a_1$};
\draw (2.260393698687867,-0.7018878981355703) node[anchor=north west] {$a_2$};
\draw (2.919202348332175,-1.9860064525270134) node[anchor=north west] {$a_3$};
\draw (3.2653560456029127,-0.6683891532384022) node[anchor=north west] {$a_4$};
\draw (3.879499702050996,-1.9971727008260693) node[anchor=north west] {$a_5$};
\draw (4.281484640817014,-0.7018878981355703) node[anchor=north west] {$a_6$};
\draw (5.699598174797134,-0.7018878981355703) node[anchor=north west] {$a_{2r}$};
\draw (5.7219306713952465,-2.003177721237397) node[anchor=north west] {$a_{2r+1}$};
\draw (6.168580603357489,-0.6460566566402902) node[anchor=north west] {$a_{2r+2}$};
\draw (6.805056756403684,-2.01078427144336212) node[anchor=north west] {$a_{2r+3}$};
\draw [color=zzttqq] (8.0,5.0)-- (8.0,4.5);
\draw [color=zzttqq] (8.0,4.5)-- (8.5,4.5);
\draw [color=zzttqq] (8.5,4.5)-- (8.5,5.0);
\draw [color=zzttqq] (8.5,5.0)-- (8.0,5.0);
\draw (7.966842814197019,4.804064982682059) node[anchor=north west] {$P_r(2r)$};
\draw (8.311728445601325,4.875648756236932)-- (9.0,5.5);
\draw (9.0,5.853692322793326) node[anchor=north west] {$x_S$};
\draw [color=zzttqq] (8.0,4.0)-- (8.0,3.5);
\draw [color=zzttqq] (8.0,3.5)-- (8.5,3.5);
\draw [color=zzttqq] (8.5,3.5)-- (8.5,4.0);
\draw [color=zzttqq] (8.5,4.0)-- (8.0,4.0);
\draw (7.950341559094188,3.843271394271737) node[anchor=north west] {$P_r(2r)$};
\draw (8.311728445601323,3.8756487562369317)-- (9.0,4.5);
\draw [color=zzttqq] (8.0,2.0)-- (8.0,1.5);
\draw [color=zzttqq] (8.0,1.5)-- (8.5,1.5);
\draw [color=zzttqq] (8.5,1.5)-- (8.5,2.0);
\draw [color=zzttqq] (8.5,2.0)-- (8.0,2.0);
\draw (7.958009062496076,1.844512948740708) node[anchor=north west] {$P_r(2r)$};
\draw (8.311728445601325,1.875648756236932)-- (9.0,2.5);
\draw (9.037313946831888,2.8611377786463105) node[anchor=north west] {$x_{S_{(2r+2)^P}}$};
\draw (10.495891678161739,5.504412675073232)-- (12.0,4.0);
\draw (10.495891678161739,4.504412675073233)-- (12.0,4.0);
\draw (10.49589167816174,2.5044126750732327)-- (12.0,4.0);
\draw [color=qqqqcc,dotted] (2.0,3.0)-- (8.060474322735429,4.710321843298829);
\draw [color=qqqqcc,dotted] (3.5,4.0)-- (8.060474322735429,4.710321843298829);
\draw [color=qqqqcc,dotted] (4.5,4.0)-- (8.060474322735429,4.710321843298829);
\draw [color=qqqqcc,dotted] (4.0,3.0)-- (8.060474322735429,4.710321843298829);
\draw [color=qqqqcc,dotted] (6.0,3.0)-- (8.060474322735429,4.710321843298829);
\draw [color=qqqqcc,dotted] (6.5,2.0)-- (8.060474322735429,4.710321843298829);
\draw [color=qqqqcc,dotted] (5.5,2.0)-- (8.060474322735429,4.710321843298829);
\draw [color=qqqqcc,dotted] (6.0,1.0)-- (8.060474322735429,4.710321843298829);
\draw [color=qqqqcc,dotted] (2.0,1.0)-- (8.060474322735429,4.710321843298829);
\draw [color=qqqqcc,dotted] (2.5,2.0)-- (8.060474322735429,4.710321843298829);
\draw [color=qqqqcc,dotted] (3.0,1.0)-- (8.060474322735429,4.710321843298829);
\draw [color=qqqqcc,dotted] (4.5,2.0)-- (8.060474322735429,4.710321843298829);
\draw [color=qqqqcc,dotted] (2.0,-2.0)-- (3.6772311488606046,-0.5118364554352106);
\draw [color=qqqqcc,dotted] (2.5,-1.0)-- (3.476314882726506,-0.172790256333921);
\draw [color=qqqqcc,dotted] (3.5,-1.0)-- (3.8974225635955495,-0.6239599848039722);
\draw [color=qqqqcc,dotted] (4.5,-1.0)-- (4.804963941784331,-0.72280112500275);
\draw [color=qqqqcc,dotted] (3.0,-2.0)-- (4.409599380989219,-0.6958444504030764);
\draw [color=qqqqcc,dotted] (4.0,-2.0)-- (5.128444036980332,-0.8485989398011875);
\draw [color=qqqqcc,dotted] (6.5,-1.0)-- (6.720814356466352,-0.27293496166711495);
\draw [dashed](2.5,4.5)-- (2.5,4.0);
\draw [dashed](3.5,4.5)-- (3.5,4.0);
\draw [dashed](4.5,4.5)-- (4.5,4.0);
\draw [dashed](5.5,4.5)-- (5.5,4.0);
\draw [dashed](6.5,4.5)-- (6.5,4.0);
\draw [dashed](2.5,2.5)-- (2.5,2.0);
\draw [dashed](3.5,2.5)-- (3.5,2.0);
\draw [dashed](4.5,2.5)-- (4.5,2.0);
\draw [dashed](5.5,2.5)-- (5.5,2.0);
\draw [dashed](6.5,2.5)-- (6.5,2.0);
\draw [dashed](2.5,-0.5)-- (2.5,-1.0);
\draw [dashed](3.5,-0.5)-- (3.5,-1.0);
\draw [dashed](4.5,-0.5)-- (4.5,-1.0);
\draw [dashed](5.5,-0.5)-- (5.5,-1.0);
\draw [dashed](6.5,-0.5)-- (6.5,-1.0);
\draw [dotted](9.872134515538475,4.023597414610263)-- (9.872134515538475,3.3581928765117923);
\draw (1.8534461127961481,2.986334579062063) node[anchor=north west] {$a_1$};
\draw (2.83296853781331,2.9333874209530274) node[anchor=north west] {$a_3$};
\draw (3.8257277523577313,2.9598610000075452) node[anchor=north west] {$a_5$};
\draw (1.4252480155014815,4.002121546732141) node[anchor=north west] {$C_1$};
\draw (1.4381509601757811,1.9118445094956056) node[anchor=north west] {$C_2$};
\draw (1.4897627388729797,-1.1590563229877004) node[anchor=north west] {$C_P$};
\draw [dotted](4.5,4.0)-- (5.5,4.0);
\draw [dotted](4.0,3.0)-- (6.0,3.0);
\draw [dotted](4.5,2.0)-- (5.5,2.0);
\draw [dotted](4.0,1.0)-- (6.0,1.0);
\draw [dotted](4.5,-1.0)-- (5.5,-1.0);
\draw [dotted](4.0,-2.0)-- (6.0,-2.0);

\draw[snake=snake=expanding waves] (9.0,5.5)-- (10.495891678161739,5.504412675073232);
\draw[snake=snake=expanding waves] (9.0,4.5)-- (10.495891678161739,4.504412675073233);
\draw[snake=snake=expanding waves] (9.0,2.5)-- (10.49589167816174,2.5044126750732327);
\draw[snake=snake=expanding waves] (8.148076232686863,0.4721193771466411)-- (9.260165325933267,0.4721193771466411);
\draw (9.36743403513338,0.5867330618082274) node[anchor=north west] {Path of length $r-1$};
\draw[dashed] (10.5,6.0)-- (10.5,5.5);
\draw[dashed] (10.5,5.0)-- (10.5,4.5);
\draw[dashed] (10.5,3.0)-- (10.5,2.5);
\draw (10.588182184193208,5.788561520144711) node[anchor=north west] {$\bar{x}_S$};
\draw (10.562420414663274,2.5082295333330182) node[anchor=north west] {$\bar{x}_{S_{(2r+2)^P}}$};
\draw[snake=snake=expanding waves] (12.0,4.0)-- (11.902032430219887,3.2639081062111046);
\draw (11.970730482299713,4.371664195998299) node[anchor=north west] {$x$};

\begin{scriptsize}
\draw [fill=qqqqff] (2.0,1.0) circle (1.5pt);
\draw [fill=qqqqff] (3.0,1.0) circle (1.5pt);
\draw [fill=qqqqff] (4.0,1.0) circle (1.5pt);
\draw [fill=qqqqff] (2.5,2.0) circle (1.5pt);
\draw [fill=qqqqff] (3.5,2.0) circle (1.5pt);
\draw [fill=xdxdff] (3.0,2.0) circle (1.5pt);
\draw [fill=xdxdff] (2.5,1.0) circle (1.5pt);
\draw [fill=xdxdff] (3.5,1.0) circle (1.5pt);
\draw [fill=qqqqff] (4.5,2.0) circle (1.5pt);
\draw [fill=xdxdff] (4.0,2.0) circle (1.5pt);
\draw [fill=xdxdff] (3.5,1.0) circle (1.5pt);
\draw [fill=qqqqff] (6.0,1.0) circle (1.5pt);
\draw [fill=qqqqff] (7.0,1.0) circle (1.5pt);
\draw [fill=qqqqff] (5.5,2.0) circle (1.5pt);
\draw [fill=qqqqff] (6.5,2.0) circle (1.5pt);
\draw [fill=xdxdff] (6.0,2.0) circle (1.5pt);
\draw [fill=xdxdff] (6.5,1.0) circle (1.5pt);
\draw [fill=qqqqff] (2.0,3.0) circle (1.5pt);
\draw [fill=qqqqff] (3.0,3.0) circle (1.5pt);
\draw [fill=qqqqff] (4.0,3.0) circle (1.5pt);
\draw [fill=qqqqff] (2.5,4.0) circle (1.5pt);
\draw [fill=qqqqff] (3.5,4.0) circle (1.5pt);
\draw [fill=xdxdff] (3.0,4.0) circle (1.5pt);
\draw [fill=xdxdff] (2.5,3.0) circle (1.5pt);
\draw [fill=xdxdff] (3.5,3.0) circle (1.5pt);
\draw [fill=qqqqff] (4.5,4.0) circle (1.5pt);
\draw [fill=xdxdff] (4.0,4.0) circle (1.5pt);
\draw [fill=xdxdff] (3.5,3.0) circle (1.5pt);
\draw [fill=qqqqff] (6.0,3.0) circle (1.5pt);
\draw [fill=qqqqff] (7.0,3.0) circle (1.5pt);
\draw [fill=qqqqff] (5.5,4.0) circle (1.5pt);
\draw [fill=qqqqff] (6.5,4.0) circle (1.5pt);
\draw [fill=xdxdff] (6.0,4.0) circle (1.5pt);
\draw [fill=xdxdff] (6.5,3.0) circle (1.5pt);
\draw [fill=qqqqff] (2.0,-2.0) circle (1.5pt);
\draw [fill=qqqqff] (3.0,-2.0) circle (1.5pt);
\draw [fill=qqqqff] (4.0,-2.0) circle (1.5pt);
\draw [fill=qqqqff] (2.5,-1.0) circle (1.5pt);
\draw [fill=qqqqff] (3.5,-1.0) circle (1.5pt);
\draw [fill=xdxdff] (3.0,-1.0) circle (1.5pt);
\draw [fill=xdxdff] (2.5,-2.0) circle (1.5pt);
\draw [fill=xdxdff] (3.5,-2.0) circle (1.5pt);
\draw [fill=qqqqff] (4.5,-1.0) circle (1.5pt);
\draw [fill=xdxdff] (4.0,-1.0) circle (1.5pt);
\draw [fill=xdxdff] (3.5,-2.0) circle (1.5pt);
\draw [fill=qqqqff] (6.0,-2.0) circle (1.5pt);
\draw [fill=qqqqff] (7.0,-2.0) circle (1.5pt);
\draw [fill=qqqqff] (5.5,-1.0) circle (1.5pt);
\draw [fill=qqqqff] (6.5,-1.0) circle (1.5pt);
\draw [fill=xdxdff] (6.0,-1.0) circle (1.5pt);
\draw [fill=xdxdff] (6.5,-2.0) circle (1.5pt);
\draw [fill=qqqqff] (8.0,5.0) circle (0.5pt);
\draw [fill=qqqqff] (8.0,4.5) circle (0.5pt);
\draw [fill=uuuuuu] (8.5,4.5) circle (0.5pt);
\draw [fill=uuuuuu] (8.5,5.0) circle (0.5pt);
\draw [fill=qqqqff] (8.311728445601325,4.875648756236932) circle (1.5pt);
\draw [fill=qqqqff] (9.0,5.5) circle (1.5pt);
\draw [fill=qqqqff] (10.495891678161739,5.504412675073232) circle (1.5pt);
\draw [fill=qqqqff] (8.0,4.0) circle (0.5pt);
\draw [fill=qqqqff] (8.0,3.5) circle (0.5pt);
\draw [fill=uuuuuu] (8.5,3.5) circle (0.5pt);
\draw [fill=uuuuuu] (8.5,4.0) circle (0.5pt);
\draw [fill=qqqqff] (8.311728445601323,3.8756487562369317) circle (1.5pt);
\draw [fill=qqqqff] (9.0,4.5) circle (1.5pt);
\draw [fill=qqqqff] (10.495891678161739,4.504412675073233) circle (1.5pt);
\draw [fill=qqqqff] (8.0,2.0) circle (0.5pt);
\draw [fill=qqqqff] (8.0,1.5) circle (0.5pt);
\draw [fill=uuuuuu] (8.5,1.5) circle (0.5pt);
\draw [fill=uuuuuu] (8.5,2.0) circle (0.5pt);
\draw [fill=qqqqff] (8.311728445601325,1.875648756236932) circle (1.5pt);
\draw [fill=qqqqff] (9.0,2.5) circle (1.5pt);
\draw [fill=qqqqff] (10.49589167816174,2.5044126750732327) circle (1.5pt);
\draw [fill=qqqqff] (12.0,4.0) circle (1.5pt);
\draw [fill=qqqqff] (8.06,4.71) circle (1.5pt);
\draw [fill=qqqqff] (11.902032430219887,3.2639081062111046)circle (1.5pt);

\end{scriptsize}
\end{tikzpicture}

%% file: Appendix.tex
\begin{center} {\LARGE \bf Appendices}
\end{center}

\section{Notation}
We denote the input graph with vertex set $V$ and edge set $E$ by $G =
(V,E)$ and use $n$ to denote the number of vertices. Edges of the graph are undirected and unweighted. For two vertices $u,v$, we denote the shortest distance and path between them by $d_G(u,v)$ and $P_G(u,v)$. Given a subset of vertices $S$ and $u$ a vertex of $G$, we define $d_G(u,S) = \min_{v \in S}\{d_G(u,v)\}$ and $P_G(u,S) = P_G(u,v)$ for $v = \arg \min_{v \in S}\{d_G(u,v)\}$. We omit the subscript $G$ when $G$ is clear from context. $G[S]$ is the subgraph induced by $S$. For a vertex $u$ (or a set $S$), an open $r$-neighborhood of $u$ ($S$) is the set $N_G^r(u) =\{v \in G| d_G(u,v) < r \}$ ($N_G^r(S) = \{u | d_G(u,S) < r\}$) and a closed $r$-neighborhood of $u$ ($S$) is the set $N_G^r[u] =\{v \in G| d_G(u,v) \leq r \}$ ($N_G^r[S] = \{u | d_G(u,S) \leq r\}$) . A set of vertices $D$ is an $r$-dominating set if $N_G^r[D] = V$.

\begin{define}[Tree decomposition] \emph{A tree decomposition} of $G$ is a tree $T$ whose nodes are subsets $X_i$ (so-called bags) of $V$ satisfying the following conditions:
  \begin{enumerate} [noitemsep,nolistsep]
  \item The union of all sets $X_i$ is $V$.
  \item For each edge $(u,v) \in E$, there is a bag $X_i$ containing both $u,v$.
  \item For a vertex $v \in V$, all the bags containing $v$ make up a subtree of $T$.  
  \end{enumerate}
\end{define}

We denote the size of the bag $X_i$ by $n_i$. We use $V_i$ to denote the set of vertices in descendant bags of $X_i: V_i = \cup_{j: X_j \mbox{is a descendant of } X_i}X_j$.
The \emph{width} of a tree decomposition $T$ is $\max_{i \in T}|X_i| -1$ and the treewidth of $G$ is the minimum width among all possible tree decompositions of $G$. We will assume throughout that graph $G$ has treewidth $\tw$ and that we are given a tree decomposition of $G$ of width $\tw$.

\section{Algorithm for $r$-dominating Set} \label{App:AppendixA}
To simplify the dynamic program, we will use a {\em nice tree decomposition}.  Kloks shows how to make a tree decomposition {\em nice} in linear time with only a constant factor increase in space (Lemma 13.1.2~\cite{Kloks94}).

\begin{define}
(\emph{Nice tree decomposition}) A tree decomposition $T$ of $G$ is \emph{nice} if the following conditions hold:
\begin{itemize}[noitemsep,nolistsep]
	\item $T$ is rooted at node $X_0$.
	\item Every node has at most two children.
	\item Any node $X_i$ of $T$ is one of four following types:	
		\begin{description}
		\item{\bf {leaf node}} $X_i$ is a leaf of $T$,
		\item{\bf {forget node}} $X_i$ has only one child $X_j$ and  $X_i = X_j \backslash \{v\}$ ,
		\item{\bf {introduce node}} $X_i$ has only one child $X_j$ and  $X_j = X_i \backslash \{v\}$,
		\item{\bf {join node}} $X_i$  has two children $X_j, X_k$ and $X_i = X_j = X_k$.
		\end{description}
	\end{itemize}
\end{define}
The dynamic programming table $A_i$ for a node $X_i$ of the tree decomposition is indexed by bags of the tree decomposition and all possible {\em distance-labelings} of the vertices in that bag.  For a vertex $v$ in bag $X_i$, a positive distance label for $v$ indicates that $v$ is $r$-dominated at that distance by a vertex in $V_i$, and a negative distance label for $v$ indicates that $v$ should be $r$-dominated at that distance via a vertex in $V \setminus V_i$.  

For an $r$-dominating set $D$, we say that $D$ {\em induces} the labeling $c: X_i \rightarrow [-r,r]$ for bag $X_i$ such that:
\[
c(u) = \left\{
  \begin{array}[c]{ll}
    d_G(u,D) & \mbox{if }P_G(u,D) \subseteq G[V_i] \\
    -d_G(u,D) & \mbox{otherwise}
  \end{array}
\right.
\]

If $D$ induces the labeling $c$, the set $D \cap V_i$ is the partial solution associated with $c$. We limit ourselves to labelings that are locally valid; $c$ is {\em locally valid} if
$|c(u) - c(v)| \leq 1$ for any two adjacent vertices $u,v \in X_i$.  If a labeling $c$ is not locally valid, we define $A_i[c] = -\infty$.

We show how to populate $A_i$ from the populated tables for the child/children of $X_i$. 
Over the course of the dynamic programming, we maintain the following correctness invariant at all bags of the of the tree $T$:\\
{\bf Correctness Invariant:} For any locally valid labeling $c$
  of $X_i$, $A_i[c]$ is the minimum size of the partial solution associated with labeling $c$.

From the root bag $X_0$, we can extract the minimum size of an $r$-dominating set from the root's table.
This is the optimal answer by the correctness invariant and the definition of {\em induces}.

We will show how to handle each of the four types of nodes (leaf, forget, introduce and join) in turn.  We use $\# _0(X_i,c)$ to denote the number of vertices in $X_i$ that are assigned label $0$ in $c$ in populating the tables for leaf and join nodes.

We say $v$ \emph{positively resolves} $u$ if $c(v) = c(u)-1$ when $c(u) > 0$; we use this definition for leaf and introduce nodes.

We will use the following Ordering Lemma to reduce the number of cases
we need to consider in populating the table of an introduce node and
join node.  We define an ordering $\preceq$ on labels for single
vertices: $\ell_1\preceq \ell_2$ if $\ell_1 = \ell_2$ or $\ell_1 = -t,
\ell_2 = t$ for $t \ge 0$.  We extend this ordering to labelings $c,
c'$ for a bag of vertices $X_i$ by saying $c \preceq c'$ if $c(u)
\preceq c'(u)$ for all $u \in X_i$.

\begin{lemma}[Ordering Lemma] \label{lm:mono-lm}
If two labelings $c$ and $c'$ of $X$ satisfy $c' \preceq c$, then $A_i[c'] \leq A_i[c]$.
\end{lemma}

\begin{proof} Let $D$ and $D'$ be minimum $r$-dominating sets that induce labelings $c$ and $c'$ on $X_i$, respectively. Let $D_i =  D \cap V_i$ and $D_i' = D' \cap V_i$.  By the definition of $\preceq$ and inducing, $(D' \setminus D'_i)\cup D_i$ is also an $r-$dominating set of $G$. Since $D'$ is the minimum $r$-dominating set that induces $c'$ on $X_i$, $|D'| \le |(D' \setminus D'_i)\cup D_i|$, and so $|D_i'| \le |D_i|$, or equivalently, $A_i[c'] \leq A_i[c]$.
\qed \end{proof}

\subsection{Leaf Node} 

We populate the table $A_i$ for a leaf node $X_i$ as follows:
\begin{equation}
  \label{eq:leaf}
  A_i[c] = \left\{
    \begin{array}[c]{ll}
      0 & \mbox{if $c$ is locally valid and all negative}\\
      \# _0(X_i,c) &  \mbox{if $c$ is locally valid and all positive labels are positively resolved}\\
      \infty & \mbox{otherwise}
    \end{array}
  \right.
\end{equation}

Since we can check local validity and positive resolution in time proportional to the number of edges in $G[X_i]$ ($O(n_i^2)$) and there are $(2r+1)^{n_i}$ labelings of $X_i$, we get:

\begin{observation}
  The time to populate the table for a leaf node is $O(n_i^2(2r+1)^{n_i})$.
\end{observation}

\begin{lemma}
The correctness invariant is maintained at leaf nodes.
\end{lemma}

\begin{proof}
Clearly the correctness invariant is maintained for labels resulting in the first and third cases of Equation~(\ref{eq:leaf}).  For the second case, we argue the correctness invariant by induction on the label of vertices; let $S$ be the set of vertices labeled 0 by $c$.  For a vertex $u$ with label $c(u) = 1$, by the local validity, $u$ must have a neighbor $v$ in $X_i$ such that $c(v) = c(u) - 1 = 0$. Therefore, 
we have $d_G(u,S) = c_i(u)$ and $P_G(u,S) \subseteq G[V_i] = G[X_i]$. Suppose that for all vertices $v \in X_i$ which have label $t > 0$, we have $d_{G}(v,S) = t$ and $P_G(u,S) \subseteq G[V_i] = G[X_i]$. For a vertex $u$ that has label $t+1$, we prove that $d_{G}(u,S) = t + 1$ and $P_G(u,S) \subseteq G[V_i] = G[X_i]$ since, by the local validity of $c$, $u$ has a neighbor $v$ such that $c(v) = t = d_{G}(v,S)$. 
\qed \end{proof}

\subsection{Forget Node} Let $X_i$ be a forget node with child $X_j$ and $X_i = X_j \cup \{u\}$. Let $c_i$ be a labeling of $X_i$.  We consider extensions of $c_i$ to labelings $c_j = c_i \times d$ of $X_j$ as follows:
\[
c_j(v) = \left\{
\begin{array}[c]{ll}
  c_i(v) & \mbox{if } v \in X_i \\
  d & \mbox{otherwise}
\end{array}
\right.
\]
We populate the table $A_i$ for forget node $X_i$ as follows:
\begin{equation}
A_i[c_i] = \min
\begin{cases}
  A_j[c_i \times d] & \forall d < 0, \exists \mbox{ a $v$ in $X_i$ s.t. $c_i(v) = d+ d_G(u,v)$}\\
  A_j[c_i \times d] & \forall d \ge 0 
\end{cases}\label{eq:forget}
\end{equation}

In the first case above, we are considering those solutions in which $u$ will be dominated via a node in $V \setminus V_j$; in order to track feasibility, we must handle this constraint through vertices in $X_i$ whose label is closer to 0 than $u$'s.  Since the positive labels are directly inherited from the child table $A_j$, and we assume the correctness invariant for $A_j$ we get:
\begin{lemma}
The correctness invariant is maintained for forget nodes.
\end{lemma} 

For a given $c_i$ and negative values of $d$, we can check the condition for the first case of Equation~(\ref{eq:forget}) in time proportional to the degree of $u$ in $X_j$ ($O(n_i)$).  Therefore:
\begin{observation} \label{lm:lm-forget}
The time to populate the table for a forget node $X_i$ with child $X_j$ is $O(n_i(2r+1)^{n_j})$.
\end{observation}

\subsection{Introduce Node}
Let $X_i$ be an introduce node with its child $X_j$ and $X_i = X_j \cup \{u\}$. 
We show how to compute $A_i[c_i]$ where $c_i = c_j \times d$ is the extension of a labeling $c_j$ for $X_j$ to $X_i$ where $u$ is labeled $d$. We define a map $\phi$ applied to the label $c_j(v)$ of a vertex $v$: 	
	\[
 	\phi(c_j(v)) =
  	\begin{cases}
    -c_j(v) & \mbox{ if } c_j(v) > 0 \mbox{ and } d_{G[V_i]}(v, u) = d_{G}(u,v) = c_j(v) - d \\
   	c_j(v)  &     \mbox{otherwise}
  	\end{cases}
	\]

We use $\phi(c_j)$ to define the natural extension this map to a full labeling of $X_j$.  Clearly $\phi(c_j) \preceq c_j$.  This map corresponds to the lowest ordering label that is $\preceq c_j$ that we use in conjunction with the Ordering Lemma. Note that $\phi$ will be used only for $d \geq 0$.

There are three cases for computing $A_i[c_j \times d]$, depending on the value of $d$:\\
{\bf $\mathbf d = 0:$}  In this case, $u$ is in the dominating set.  If a vertex $v \in X_j$ is to be $r$-dominated by $u$ via a path contained in $G[V_i]$, it will be represented by the table entry in which $c_j(v) = -d_G(u,v)$. Therefore $A_j[c_j']+1$ corresponds to the size of a subset of $V_i$ that induces the positive labels of $c_j \times 0$ for any $c_j' \preceq c_j$ and where $c_j(v) = d_G(u,v) = d_{G[V_i]}(v, u)$.  The Ordering Lemma tells us that the best solution is given by the rule:
$A_i[c_j \times 0] = A_i[\phi(c_j) \times 0] = A_j[\phi(c_j)]+1$.\\
$\mathbf d > 0:$ In this case, $u$ is $r$-dominated is to be dominated by a vertex in $V_i$ via a path  contained by $G[V_i]$. Therefore, we require there be a neighbor $v$ of $u$ in $X_j$ (with a label) that positively resolves $u$; otherwise, the labeling $c_j \times d$ is infeasible. 
 Further, for other vertices $v'$ of $X_j$ which are $r$-dominated by a vertex of $V_i$ by a path through $u$ and contained by $G[V_i]$, the condition of the mapping $\phi$ must hold: 
$d_{G}(v',u) = d_{G[V_i]}(v',u) = c_j(v') - d$
As for the previous case, the Ordering Lemma tells us that the best solution is given by the rule:
 \[
 A_i[c \times \{t\}] =
 \begin{cases}
   A_j[\phi(c)]  & \mbox{ if } \exists v \in X_i \mbox{ s.t } v \mbox{ positively resolves } u  \\
   +\infty       & \mbox{otherwise}
 \end{cases}
 \] 
$\mathbf d < 0:$ In this case, $u$ is not $r$-dominated by a path contained entirely in $G[V_i]$.  Therefore, the table entries for $X_i$ are simply inherited from $X_j$ (as long as $c_j \times d$ is locally valid).  We get:
\[
A_i[c \times \{c_i(u)\}] =
\begin{cases}
  A_j[c]  & c \times \{c_i(u)\} \mbox{ is a locally valid labeling of } X_i \\
  +\infty     &   \mbox{otherwise}
\end{cases}
\] 	   
Since the correctness invariant holds for $X_j$, and by the arguments above (using the Ordering Lemma), we get:

\begin{lemma}
The correctness invariant is maintained for introduce nodes.
\end{lemma}

As we can assume that we have already computed $d_{G[V_i]}(u,v)$ and $d_G(u,v)$ for all $u,v \in X_i$ and all $i$, the time to compute $A_i[c_i]$ is the time to check the condition $d_{G[V_i]}(v, u) = d_{G}(u,v) = c_j(v) - d$ ($O(\tw)$).  We get:

\begin{observation} \label{lm:lm-intro}
The time required to populate the dynamic programming table for an introduce bag is  $O((2r+1)^{\tw}\tw)$.\end{observation}

\subsection{Join Nodes} In order to populate the table $A_i$ from populated tables $A_j$ and $A_k$, we use two types of intermediate tables (for all three nodes), called the \emph{indication table} ($N$) and the \emph{convolution table} ($\bar N$), such that these tables of the node $X_i$ can be efficiently populated from the tables of its children. In this section, we will refer to the tables of the previous sections as the original tables.  We will initialize $N_j$ and $N_k$ from $A_j$ and $A_k$, then compute $\bar N_j$ from $N_j$ and $\bar N_k$ from $N_k$, then combine $\bar N_j$ and $\bar N_k$ to give $\bar N_i$, then compute $N_i$ from $\bar N_i$ and finally $A_i$ from $N_i$.  The tables $\bar N_j$ can be used to count the number of $r$-dominating sets; we view our method as incorrectly counting so that we can more efficiently compute $A_i$ from $A_j$ and $A_k$ while still correctly computing $A_i$.

Let $X_i$ be a join node with two children $X_j$ and $X_k$ and $X_i = X_j = X_k$.  We say the labeling $c_i$ (for $X_i)$ is \emph{consistent} with labelings $c_j$ and $c_k$ (for $X_j$ and $X_k$, respectively) if for every $u \in X_i$:
	\begin{enumerate} [noitemsep, nolistsep]
	\item If $c_i(u) = 0$,  then $c_j(u) = 0$ and $c_k(u) = 0$.
	\item If $c_i(u) = t < 0$, then $c_j(u) = t$ and $c_k(u) = t$.
	\item If $c_i(u) = t > 0$, then $(c_j(u) = t)\wedge(c_k(u) = -t)$ or $(c_j(u) = -t)\wedge(c_k(u) = t)$ or $(c_j(u) = t)\wedge(c_k(u) = t)$.
	\end{enumerate}
\begin{equation} \label{eq: join-bag-cons}
	A_i[c_i] = min(A_j[c_j] + A_k[c_k] - \#_0(X_i,c_i) | c_i \mbox{ is consistent with } c_j \mbox{ and } c_k)
\end{equation}
\begin{proof}[Equation~(\ref{eq: join-bag-cons}) is correct]
For a node $u \in X_i$:
 	\begin{enumerate}[noitemsep, nolistsep]
 	\item If $c_i(u) = 0$, $u$ is in the dominating set. Therefore, $u$ must also assigned label $0$ in $X_j$ and $X_k$.
 	\item If $c_i(u) < 0$, $u$ is not $r$-dominated via a path contained in $G[V_i]$. Therefore, in $X_j$ and $X_k$, $u$ is also not $r$-dominated via a path in $G[V_j]$ and $G[V_k]$. Hence, $c_j(u) = c_k(u) = c_i(u)$   
 	\item If $c_i(u) > 0$, $u$ is $r$-dominated via a path in $G[V_i]$. There are two cases: (a) $u$ is $r$-dominated via a path in  $G[V_j]$ (b) $u$ is $r$-dominated via a path in $G[V_k]$.  The Ordering Lemma implies that the best result is given by $c_j(u) = c_i(u)$ and $c_k(u) = -c_i(u)$ (case (a)) or $c_j(u) = -c_i(u)$ and $c_k(u) = c_i(u)$ (case (b)).  That is the case $c_j(u) = c_i(u)$ and $c_k(u) = c_i(u)$ can be ignored.
 	\end{enumerate}  
\qed \end{proof}

The {\em indication table} indexes the solution $A_j[c]$ and is indexed by labellings and numbers from 0 to $n$.  We initialize the indication table $N_j$ for $X_j$ by: 
\begin{equation} \label{eq: eq-ind-N_i}
	N_j[c_j][x] = \left\{
  		\begin{array}{ll}
    		1 & \mbox{if} A_j[c_j] = x\\
    		0 & \mbox{otherwise}
  		\end{array}
		\right.	
	\end{equation}
We likewise initialize $N_k$.
 Using the indication table, Equation~\ref{eq: join-bag-cons} can be written as:
	\begin{equation} \label{eq:ind-converse}
	A_i[c_i] = \min\{\infty, \min \{x \ : \ N_i[c_i][x] > 0 \}\}
	\end{equation}
where we define
	\begin{equation}\label{eq:nsum}
	N_i[c_i][x] = \sum_{\substack{x_j, x_k : x_j + x_k - \#_0(X_i,c_i) = x\\
	c_i \text{ is consistent with } c_j \mbox{ and } c_k}}N_j[c_j][x_j]\cdot N_k[c_k][x_k]
	\end{equation}
We guarantee that $N_i[c_i][x]$ is non-zero only if there is a subset of $V_i$ of size $x$ that induces the positive labels of $c_i$.  This, along with the correctness invariant held for the children of $X_i$ gives us:
\begin{lemma}
The correctness invariant is maintained at join nodes.
\end{lemma}

The {\em convolution table} $\bar N_i$ for $X_i$  (and likewise $\bar N_j$ and $\bar N_k$ for $X_j$ and $X_k$) is also indexed by labellings and numbers from 0 to $n$.  However, we use a different labeling scheme.  To distinguish between the labeling schemes, we use the {\em bar}-labels $[-r, \ldots, -1, 0, \bar 1, \ldots, \bar r]$ for the convolution table and $\bar c$ to represent a {\em bar}-labeling of the vertices in a bag.  We define the convolution table in terms of the indication tables as:
\begin{equation}
\bar{N}_i[\bar c][x] = \sum_{c\ :\ \overline{|c(u)|}= \bar c(u) } N_i[c][x]\label{eq:bar-no-bar}
\end{equation}

The following observation, which is a corollary of Equation~(\ref{eq: join-bag-cons}) and the definition of {\em consistent}, is the key to our algorithm:
\begin{observation} \label{obs: obs-encode}
If the vertex $u \in X_i$ has label $\bar{t}$, its label in $X_j$ and $X_k$ must also $\bar{t}$. 
\end{observation}

\subsection{Running time analysis}

\begin{lemma} \label{lm: conversion-table}
Convolution tables can be computed from indication tables {\em and vice versa} in time $O(nn_i(2r+1)^{n_i})$.
\end{lemma}
\begin{proof}
Consider the indicator table $N_i$ and convolution table $\bar N_i$ for bag $X_i$; we order the vertices of $X_i$ arbitrarily.  We calculate Equation~(\ref{eq:bar-no-bar}) by dynamic programming over the vertices in this order.

We first describe how to compute $\bar N_i$ from $N_i$.   We initialize $\bar N_i[c] = N_i[c']$ where $c'(u) = c(u)$ if $c(u) \le 0$ and $\overline{c'(u)} = c(u)$ if $c'(u) > 0$.  We then correct the table by considering the barred labels of the vertices according to their order from left to right.  In particular, suppose $c = c_1\times \{\bar{t}\}\times c_2$, for some $t > 0$, be a labeling in which $c_1$ is a bar-labeling of the first $\ell$ vertices of $X_i$ and $c_2$ is a bar-labeling of the last $n_i - \ell - 1$ vertices.   We update $\bar N_i[c]$ in order from $\ell = 0,\ldots, n_i$ according to:
	\begin{equation} \label{eq: eq-trans}
	\bar{N}_i[c_1\times \{\bar{t}\}\times c_2][x] := \bar{N}_i[c_1\times \{\bar t\} \times c_2][x] + \bar{N}_i[c_1\times \{-t\} \times c_2][x]   
	\end{equation} 
It is easy to show that this process results in the same table as Equation~(\ref{eq:bar-no-bar}).

We now describe how to compute $N_i$ from $\bar N_i$, which is the same process, but in reverse.  We initialize $N_i[c] = \bar N_i[c']$ where $c(u) = c'(u)$ if $c(u) \le 0$ and $\overline{c(u)} = c'(u)$ if $c(u) > 0$.    We then update $N_i[c]$ in reverse order of the vertices of $X_i$, i.e.\ for $\ell = n_i, n_i-1, \ldots, 0$ according to:
	\begin{equation} \label{eq: rev-eq-trans}
	N_i[c_1\times \{t\}\times c_2][x] := N_i[c_1\times \{t\} \times c_2][x] - N_i[c_1\times \{-t\} \times c_2][x]   
	\end{equation}
Since the labeling $c$ has length of $n_i$, the number of operations for the both forward and backward conversion is bounded by $O(nn_i(2r+1)^{n_i})$.
\qed \end{proof}

\begin{lemma} \label{lm: convolution} The convolution tables for $X_i, X_j$ and $X_k$ satisfy: 
 \begin{equation} \label{eq: convolution}
 \bar{N}_i[\bar{c}][x] = \sum_{x_i, x_j\ : \ x_i + x_j - \#_0(X_i,\bar{c}) = x} \bar{N}_j[\bar{c}][x_i] \cdot \bar{N}_k[\bar{c}][x_j]
 \end{equation}
\end{lemma}

\begin{proof}
\begin{equation*}
	\begin{split}
	\bar{N}_i[\bar{c}][x] &= \sum_{c\ :\ \overline{|c(u)|}= \bar c(u) } N_i[c][x] \qquad \mbox{by Equation~(\ref{eq:bar-no-bar})} \\
	&= \sum_{c\ :\ \overline{|c(u)|}= \bar c(u) } \sum_{\substack{x_j, x_k\ : \ x_j + x_k - \#_0(X_i,c) = x\\
	c_j, c_k\ :\ c \text{ is consistent with } c_j \text{ and } c_k}}(N_j[c_j][x_j]\cdot N_k[c_k][x_k])  \qquad  \mbox{by Equation~(\ref{eq:nsum})}\\
	&= \sum_{x_j, x_k\ : \ x_j + x_k - \#_0(X_i,c) = x} \sum_{\substack{c\ :\ \overline{|c(u)|}= \bar c(u) \\
	c_j, c_k\ :\ c \mbox{ is consistent with } c_j \mbox{ and } c_k}}(N_j[c_j][x_j]\cdot N_k[c_k][x_k])  \\	&=  \sum_{x_j, x_k\ : \ x_j + x_k - \#_0(X_i,c) = x} \sum_{\substack{c_j\ :\ \overline{|c_j(u)|}= \bar c(u) \\c_k\ :\ \overline{|c_k(u)|}= \bar c(u) }}(N_j[c_j][x_j]\cdot N_k[c_k][x_k]) \qquad \text{see explanation below}\\
&=  \sum_{x_j, x_k\ : \ x_j + x_k - \#_0(X_i,c) = x} \sum_{c_j\ :\ \overline{|c_j(u)|}= \bar c(u)}\bar{N}_j[\bar{c}][x_1] 
 \sum_{c_k\ :\ \overline{|c_k(u)|}= \bar c(u) }N_k[c_k][x_k] \\
	&= \sum_{x_j, x_k\ : \ x_j + x_k - \#_0(X_i,c) = x} \bar{N}_j[\bar{c}][x_1]\cdot\bar{N}_k[\bar{c}][x_2] \qquad \mbox{by Equation~(\ref{eq:bar-no-bar})}
	\end{split}
	\end{equation*}
        To explain the conversion between a sum over all $c$ such that $\overline{|c(u)|}= \bar c(u)$ and all $c_j, c_k$ such that $c \text{ is consistent with } c_j \text{ and } c_k$ to a sum over all $c_j$ such that $\overline{|c_j(u)|}= \bar c(u)$ and all $c_k$ such that $\overline{|c_k(u)|}= \bar c(u)$, we note that the constraint $\overline{|c(u)|}= \bar c(u)$ is the same as {\em consistent with}.
\qed \end{proof}

\begin{lemma} \label{lm:lm-join}
The time required to populate the dynamic programming table for join node $X_i$ is $O(n^2(2r+1)^{n_i})$.
\end{lemma}
\begin{proof}
We update the table $A_i$ of the join node $X_i$ by following steps:
	\begin{enumerate} [nolistsep, noitemsep]
	\item Computing the indication tables $N_j[c_j][x]$ and $N_k[c_k][x]$ for all possible $c_j,c_k$ and $x$ by Equation~\eqref{eq: eq-ind-N_i} takes $O(n(2r+1)^{n_i})$ time.
	\item Computing the convolution tables $\bar{N}_j[\bar c][x]$ and $\bar{N}_k[\bar c][x]$ via  Lemma~\ref{lm: conversion-table} takes $O(nn_i(2r+1)^{n_i})$ time.
	\item Computing the table $\bar{N}_i[\bar c][x]$ of the join node $X_i$ via Lemma~\ref{lm: convolution} takes $O(n^2(2r+1)^{n_i})$ time.
	\item Computing the indication table $N_i[c_i][x]$ of the join node $X_i$ via Lemma~\ref{lm: conversion-table} takes $O(nn_i(2r+1)^{n_i})$ time.
	\item Computing the table $A_i[c_i]$ of the join node $X_i$ by Equation~\eqref{eq:ind-converse} takes $O(n(2r+1)^{n_i})$ time.
	\end{enumerate}
\qed \end{proof}

\paragraph{Proof of Theorem~\ref{thm:dp-thm-ds}} Theorem~\ref{thm:dp-thm-ds} follows from the correctness and running time analyses for each of the types of nodes of the nice tree decomposition.  Using the \emph{finte integer index} property~\cite{BA01,RBR09}, we can reduce the running time of Theorem~\ref{thm:dp-thm-ds} to $O((2r+1)^{\tw + 1}\tw^2n)$. 

For a given bag $X_i$, let $S_c$ be the minimum partial solution that is associated with a labeling $c$ of $X_i$; $|S_c| = A[c]$. Let $S_1$ be the minimum partial solution that is associated with the labeling $c_1 = \{1,1,\ldots,1\}$ of $X_i$  
\begin{lemma}[Claim 5.4~\cite{BA01} \textbf{finite integer index property}] \label{lm:finit-int-index}
For a given bag $X_i$, if the minimum partial solution $S_c$ can lead to an optimal solution of $G$, we have:	
	\begin{center}
	$||S_c| - |S_1|| \leq n_i + 1$
	\end{center}
\end{lemma}

Given Lemma~\ref{lm:finit-int-index}, we can prove following theorem.

\begin{theorem} \label{thm:linear-join-node}
We can populate the dynamic programming table for join node $X_i$ in $O(n_i^2(2r+1)^{n_i})$ time.
\end{theorem}

\begin{proof}
By Lemma~\ref{lm:finit-int-index}, for a fixed $\bar{c}$, there are at most $2n_i + 3$ values $x \in \{1,2,\ldots,n\}$ such that $\bar{N}[\bar{c}][x] \not= 0$. Therefore, by maintaining non-zero values only, we have:
	\begin{enumerate} [nolistsep, noitemsep]
	\item Computing the convolution tables $\bar{N}_j[\bar c][x]$ and $\bar{N}_k[\bar c][x]$ via  Lemma~\ref{lm: conversion-table} takes $O(n^2_i(2r+1)^{n_i})$ time.
	\item Computing the table $\bar{N}_i[\bar c][x]$ of the join node $X_i$ via Lemma~\ref{lm: convolution} takes $O(n_i^2(2r+1)^{n_i})$ time.
	\item Computing the indication table $N_i[c_i][x]$ of the join node $X_i$ via Lemma~\ref{lm: conversion-table} takes $O(n^2_i(2r+1)^{n_i})$ time.
	\item Computing the table $A_i[c_i]$ of the join node $X_i$ by Equation~\eqref{eq:ind-converse} takes $O(n_i(2r+1)^{n_i})$ time.
	\end{enumerate} 
\end{proof}

Clearly, Theorem~\ref{thm:linear-join-node} implies an $O((2r+1)^{\tw + 1}\tw^2n)$ time algorithm for $r$DS problem.

%
%
%

\section{Algorithm for Connected $r$-dominating Set} \label{App:AppendixC}
We apply the Cut\&Count technique by Cygan et al.~\cite{CNPPRW11} to design a randomized algorithm which decides whether there is a connected $r$-dominating set of a given size in graphs of treewidth at most $\tw$ in time $O((2r+2)^{\tw}n^{O(1)})$ with probability of false negative at most $\frac{1}{2}$ and no false positives.
\subsection{Overview of the cut-and-count technique}

Rather than doubly introduce notation, we give an overview of the Cut\&Count technique as applied to our connected $r$-dominating set problem.  The goal is, rather than search over the set of all possible {\em connected} $r$-dominating sets, which usually results in $\Omega(\tw^\tw)$ configurations for the dynamic programming table, to search over all possible $r$-dominating sets.  Formally, let $\cal S$ be the family of connected subsets of vertices that $r$-dominate the input graph and let ${\cal S}_k \subseteq {\cal S}$ be the subset of solutions of size $k$.  Likewise, let $\cal R$ be the family of (not-necessarily-connected) subsets of vertices that $r$-dominate the input graph and similarly define ${\cal R}_k$.   Note that ${\cal S}$ and ${\cal S}_k$ are subsets of ${\cal R}$ and ${\cal R}_k$, respectively.  We wish to determine, for a given $k$, whether ${\cal S}_k$ is empty.  We cannot, of course, simply determine whether ${\cal R}_k$ is empty.  Instead, for every subset of vertices $U$, we derive a family ${\cal C}(U)$ whose size is odd only if $G[U]$ is connected.  Further, we assign random weights $\omega$ to the vertices of the graph, so that, by the Isolation Lemma (formalized below), the subset of ${\cal S}_k$ contains a unique solution of minimum weight with high probability.  We can then determine, for a given $k$, the parity of $|\cup_{U \in {\cal R}_k \ : \ \omega(U) = W}{\cal C}(U)|$ for every $W$.  We will find at least one value of $W$ to result in odd parity if ${\cal S}_k$ is non-empty.

The {\em Isolation Lemma} was first introduced by Valiant and Vazirani~\cite{VV86}. Given a universe $\univ$ of $|\univ|$ elements and a weight function $\omega : \univ \rightarrow \mathbb{Z}$. For each subset $X \subseteq \univ$, we define $\omega(X) = \sum_{x\in X}\omega(x)$. Let $\mathcal{F}$ be a family of subsets of $\univ$. We say that $\omega$ {\em isolates} a family $\mathcal{F}$ if there is a \emph{unique} set in $\mathcal{F}$ that has minimum weight.
\begin{lemma} \label{lm:isolation}
(\textbf{Isolation Lemma}) For a set $\univ$, a random weight function $\omega: \univ \rightarrow \{1,2,\ldots,N\}$, and a family $\mathcal{F}$ of subsets of $\univ$:
	\[ \text{Prob}[\omega\text{ isolates }\mathcal{F}] \geq 1 - \frac{|\univ|}{N} \]
\end{lemma} 

Throughout the following, we fix a {\em root} vertex, $\rho$, of the graph $G = (V,E)$ and use a random assignment of weights to the vertices $\omega: V \rightarrow \{1,2, \ldots, 2n\}$.

\subsubsection{Cutting}

Given a graph $G= (V,E)$, we say that an ordered bipartition $(V_1,V_2)$ of $V$ is a {\em consistent cut} of $G$ if there are no edges in $G$ between $V_1$ and $V_2$ and $\rho \in V_1$.  We say that an ordered bipartition $(C_1,C_2)$ of a subset $C$ of $V$ is a {\em consistent subcut} if there are no edges in $G$ between $C_1$ and $C_2$, and, if $\rho \in C$ then $\rho \in C_1$.
\begin{lemma}[Lemma 3.3~\cite{CNPPRW11}] \label{lem:parity}Let $C$ be a subset of vertices that contains $\rho$.  The number of consistent cuts of $G[C]$ is $2^{cc(G[C])-1}$ where $cc(G[C])$ is the number of connected components of $G[C]$.
\end{lemma}
Recall the definitions of $\cal S$, ${\cal S}_k$, $\cal R$ and ${\cal R}_k$ from above.  We further 
let ${\cal S}_{k,W}$ be the subset of ${\cal S}_k$ with the further restriction of having weight $W$: ${\cal S}_{k,W} = \{U \in {\cal S}_k\ : \ \omega(U) = W\}$.  Similarly, we define ${\cal R}_{k,W}$.  Let ${\cal C}_{k,W}$ be the family of consistent cuts derived from ${\cal R}_{k,W}$ as:
\[{\cal C}_{k,W} = \{(C_1,C_2))\ :\ C \in {\cal R}_{k,W} \text{ and }(C_1,C_2) \text{ is a consistent cut of }G[C]\}\]
Since the number of consistent cuts of $G[C]$ for $C \in {\cal S}_{k,W}$ is odd by Lemma~\ref{lem:parity} and the number of of consistent cuts of $G[C]$ for $C \in {\cal R}_{k,W}\setminus {\cal S}_{k,W}$ is even by Lemma~\ref{lem:parity}, we get:
\begin{lemma}[Lemma 3.4~\cite{CNPPRW11}] \label{lm:cut} For every $W$, $|{\cal S}_{k,W}| \equiv |{\cal C}_{k,W}| \pmod 2 $.
\end{lemma}

\subsubsection{Counting}

Lemma~\ref{lm:cut} allows us to focus on computing $|{\cal C}_{k,W}| \pmod 2 $.  In the next section, we give an algorithm to compute $|{\cal C}_{k,W}|$ for all $k$ and $W$ ($W \in \{1,2,\ldots,2n^2\}$):
\begin{lemma} \label{lm:counting}
There is an algorithm which computes $|{\cal C}_{k,W}|$ for all $k$ and $W$ in time $O(k^2 n^4(2r+2)^{\tw})$.
\end{lemma}
Let $k^*$ be the size of the smallest connected $r$-dominating set.
 Since the range of $\omega$ has size $2n$, by the Isolation Lemma, the smallest value $W^*$ of $W$ such that ${\cal S}_{k^*,W}$ is non-empty also implies that $|{\cal S}_{k^*,W^*}| = 1$ with probability 1/2.  By Lemma~\ref{lm:cut}, $|{\cal C}_{k^*,W^*}|$ is also odd (with probability 1/2).  We can then find $|{\cal C}_{k^*,W^*}|$ by linear search over the possible values of $W$.  Thus Lemma~\ref{lm:counting} implies Theorem~\ref{thm:dp-thm-cds}

\subsection{ Counting Algorithm} 
To determine $|{\cal C}_{k,W}|$ for each $W$, we use dynamic programming given a tree decomposition $\mathbb{T}$ of $G$. To simplify the algorithm, we use an \emph{edge-nice} variant of $\mathbb{T}$. A tree decomposition $\mathbb{T}$ is edge-nice if each bag is one of following types:
\begin{quote}
  \begin{description}[nolistsep,noitemsep]
  \item[Leaf] a leaf $X_i$ of $\mathbb{T}$ with $X_i = \emptyset$
  \item[Introduce vertex] $X_i$ has one child bag $X_j$ and $X_i = X_j
    \cup \{v\}$
  \item[Introduce edge] $X_i$ has one child bag $X_j$ and $X_i = X_j$
    and $E(X_i) = E(X_j)\cup \{e(u,v)\}$
  \item[Forget] $X_i$ has one child bag $X_j$ and $X_j = X_i \cup
    \{v\}$
  \item[Join] $X_i$ has two children $X_j, X_j$ and $X_i = X_j = X_k$
  \end{description}
\end{quote}
We root this tree-decomposition at a leaf bag. Let $G_i = (V_i, E_i)$ be the subgraph formed by the edges and vertices of descendant bags of the bag $X_i$.

As with the dynamic program for the $r$-dominating set problem, we use a {\em distance labeling}, except we have {\em two types} of 0 labels:
\[c: X_i \rightarrow \{-r,\ldots,-1,0_1,0_2,1,\ldots,r\}\] A vertex $u$ is in a corresponding subsolution if $c(u) \in \{0_1,0_2\}$ and the subscript of 0 denotes the side of the consistent cut of the subsolution that $u$ is on.  {\em Throughout, we only allow the special root vertex to be labeled $0_1$.}  
We use the same notion of induces as for the non-connected version of the problem, with the additional requirement that we maintain bipartitions (cuts) of the solutions.  Specifically, a cut $(C_1,C_2)$ induces the labeling $c$ for a subset $X$ of vertices if $d(u,C_1\cup C_2) = c(u)$ if $c(u) > 0$, $u \in C_1$ if $c(u) = 0_1$ and $u \in C_2$ if $c(u) = 0_2$.  We limit ourselves to locally valid solutions as before.

A dynamic programming table $A_i$ for a bag $X_i$ of $\mathbb{T}$ is indexed by a distance labeling $c$ of $X_i$, and integers $t \in \{0,\ldots,n\}$ and $W \in \{0,1,\ldots,2n^2\}$.  $A_i[t,W,c]$ is the number of consistent subcuts $(C_1,C_2)$ of $G_i$ such that
\begin{itemize}
\item $|C_1 \cup C_2| = t$
\item $\omega(C_1 \cup C_2) = W$
\item $C_1 \cup C_2$ induces the labeling $c$ for $X_i$.
\end{itemize}

We show how to compute $A_i[t,W,c]$ of the bag $X_i$ given the tables of it children. 

\paragraph{Leaf} Let $X_i$ be a non-root leaf of $\mathbb{T}$:
	\[A_i[t,W,\emptyset] = 1\]
\paragraph{Introduce vertex} Let $X_i$ be an introduce vertex bag with its child $X_j$ and $X_i = X_j \cup \{u\}$.  Let $c \times d$ is a labeling of $X_i$ where $c$ is a labeling of $X_j$ and $d$ is the label of $u$. There are four cases for computing $A_i[t,W,c\times d]$ depending on the value of $d$.  Since $u$ is isolated in $G_i$, we need not worry about checking for local validity.
\[
A_i[t,W,c\times{d}] = 
\begin{cases}
  0 & \text{if $d > 0$ ($u$ cannot be $r$-dominated by a subset of $V_i$)}\\
  A_j[t-1,W-\omega(u),c] & \text{if $d = 0_1$ ($u$ is assigned to the first side of a consistent subcut)} \\
  A_j[t-1,W-\omega(u),c] & \text{if $d = 0_2$ ($u$ is assigned to the second side of a consistent subcut)} \\
  A_j[t,W,c] & \text{if $d < 0$}
\end{cases}
\]
	
\paragraph{Introduce edge} Let $X_i$ be an introduce edge bag with its child $X_j$ and $E_i = E_j\cup \{e(u,v)\}$.  We need only check for local validity, positive resolution and consistency of sub-cuts.
\begin{quote}
  \begin{description}
  \item[Local invalidity] For any labeling that is not locally valid
    upon the introduction of $uv$, that is, if $|c(u) - c(v)| > 1$, we
    set: $A_i(t,W,c) = 0$.
  \item[Inconsistent subcuts] If $c(u) = 0_1$ and $c(v) = 0_2$ (or
    vice versa), $c$ cannot correspond to a consistent subcut, so
    $A_i(t,W,c) = 0$.
  \item[Positive resolution] If $c(u) = c(v)- d_{G_i}(u,v) \ge 0$ then $u$
    positively resolves $v$. We say that a vertex $x \in G_i$ is \emph{uniquely resolved} by $u$ at distance $d$ if there is no vertex other than $u$ that positively resolves $v$ and $d_{G_i}(u,v) = d$. When the edge $e(u,v)$ is introduced to the bag $X_i$, some vertices will be positively resolved by $u$. The vertices $v \in X_i$ that are uniquely resolved by $u$ at distance $d_{G_i}(u,v)$ have negative labels in $X_j$.  We define a map $\phi$ applied to the
    label $c(x)$ of a vertex $x$:
    \[ \phi(c(x)) =
    \begin{cases}-c(x) & \mbox{if } x \mbox{ is uniquely resolved by } u  \mbox{ at distance $d_{G_i}(u,x)$} \\
      \hspace{.8em}c(x) & \mbox {otherwise}
    \end{cases}
    \]
    We use $\phi(c)$ to define the natural extension this map to a
    full labeling of $X_i$.   We get:
    \[A_i[t,W,c] = A_j[t,W,c] + A_j[t,W,\phi(c)]\]
  \end{description}
\end{quote}
In all other cases, the labeling is locally valid, the corresponding subcuts are valid and neither $u$ nor $v$ has been positively resolved, so $A_i[t,W,c] = A_j[t,W,c]$.
\paragraph{Forget}  Let $X_i$ be a forget bag with child $X_j$ such that $X_j = X_i \cup \{u\}$.  We compute $A_i[t,W,c]$ from $A_j[t,W,c\times d]$ where $c \times d$ is a labeling of $X_j$ where $u$ is labeled $d$.  We say that the labeling $c \times d$ is \emph{forgettable} if $d \geq 0$ or there is a vertex $v \in X_j$ such that $c(v) = d + d_{G_i}(u,v)$.  In the first case, $u$ has been dominated already; in the second case, the domination of $u$ must be handled through other vertices in $X_j$ in order for the labeling to be induced by a feasible solution.
\[
A_i[t,W,c] = \sum_{d \ : \ c\times d \text{ is forgettable}}A_j[t,W,c\times d]
\]

\paragraph{Join} Let $X_i$ be a join bag with children $X_j$ and $X_k$ and $X_i = X_j = X_k$.
We say the labeling $c_i$ (for $X_i)$ is \emph{consistent} with labelings $c_j$ and $c_k$ (for $X_j$ and $X_k$, respectively) if for every $u \in X_i$
	\begin{itemize} [noitemsep, nolistsep]
	\item If $c_i(u) = 0_j$ for $j \in \{1,2\}$, then $c_j(u) = 0_j$ and $c_k(u) = 0_j$.
	\item If $c_i(u) = t < 0$, then $c_j(u) = t$ and $c_k(u) = t$.
	\item If $c_i(u) = t > 0$, then one of the following must be true:
          \begin{itemize}[noitemsep, nolistsep]
          \item $c_j(u) = t$ and $c_k(u) = -t$
          \item $c_j(u) = -t$ and $c_k(u) = t$
          \item $c_j(u) = t$ and $c_k(u) = t$.
          \end{itemize}
	\end{itemize}
Given this, $A_i[t,W,c_i]$ is the product $A_j[t_1, W_1, c_j] \cdots A_k[t_2,W_2,c_k]$ summed over:
\begin{itemize}
\item all $t_1$ and $t_2$ such that $t_1+t_2-t$ is equal to the number of vertices that are labeled $0_1$ or $0_2$ by $c_i$
\item all $W_1$ and $W_2$ such that $W_1+W_2-W$ is equal to the weight of vertices that are labeled $0_1$ or $0_2$ by $c_i$
\item all $c_i$ and $c_j$ that are consistent with $c_j$ and $c_k$
\end{itemize}
By using the bar-coloring formulation, as for the disconnected case, we can avoid summing over all pairs of consistent distance labellings and instead compute $A_i[t,W,\bar c_i]$ as the product $A_j[t_1, W_1, \bar c_i] \cdots A_k[t_2,W_2,\bar c_i]$ summed over all $t_1$ and $t_2$ and all $W_1$ and $W_2$ as described above.
Using this latter formulation, we can compute $A_i$ in time $O(k^2n^3(2r+2)^{\tw})$. 
\paragraph{Running Time} The number of configurations for each  node of $\mathbb{T}$ is $O(kn^2(2r+2)^{\tw})$. The running time to update leaf, introduce edge, introduce vertex, and forget bags is $O(kn^2(2r+2)^{\tw})$. The running time to update join bags is $O(k^2n^3(2r+2)^{\tw})$. Therefore, the total running time of the counting algorithm is $O(nk^2n^3(2r+2)^{\tw}) = O(k^2 n^4(2r+2)^{\tw})$ after running the algorithm for all possible choices of root vertex $\rho$.

\section{Pathwidth of graphs in reductions}
To bound the pathwidth of our constructions, we use a \emph{mixed search game}~\cite{TUK95}. We view the graph $G$ as a system of tunnels. Initially, all edges are contaminated by a gas. An edge can be \emph{cleared} by placing two searchers at both ends of that edge simultaneously or by sliding a searcher along that edge. A cleared edge can be recontaminated if there is a path between this edge and a contaminated edge such that there is no searcher on this path. Set of rules for this game includes:	
	\begin{itemize} [noitemsep, nolistsep]
	\item Placing a searcher on a vertex
	\item Removing a searcher from a vertex
	\item Sliding a searcher on a vertex along an incident edge
	\end{itemize}
A search is a sequence of moves following these rules. A search strategy is \emph{winning} if all edges of $G$ are cleared after its termination. The minimum number of searchers required to win is the mixed search number of $G$, denoted by $\textbf{ms}(G)$. The following relation is established in~\cite{TUK95}:
	\begin{center}
	 $\pw (G) \leq \textbf{ms}(G) \leq \pw (G)  + 1 $ 
	\end{center}

%
%
%
\subsection{Proof of Lemma~\ref{lm:pw-bound}} \label{App:AppendixB}

We give a search strategy using at most $tp + O((2r+1)p)$ searchers. For a group gadget $\hat{B}$, we call the sets of vertices $\{P^1_i | 1 \leq i \leq p\}$ and $\{P^{2r+1}_i | 1 \leq i \leq p\}$ the sets of \emph{entry vertices} and \emph{exit vertices}, respectively. We search the graph $G$ in $m(2rpt+1)$ rounds. Initially, we place $tp$ searchers on the entry vertices of $t$ group gadgets $\hat{B}^1_i , 1 \leq i \leq t$. We use one more searcher to clear the path and the edges incident to $h_1$. In round $b, 1 \leq b \leq m(2prt+1)$ such that $b = m\ell+j, 0 \leq \ell \leq 2prt+1, 1 \leq j \leq m$, we keep $tp$ searchers on the entry vertices of all group gadgets $\hat{B}^{ml+j}_i , 1 \leq i \leq t$. We clear the group gadget $\hat{B}^{ml+j}_i$ by using at most $5(2r+1)p + 4$ searchers in which:
	\begin{itemize}[nolistsep, noitemsep]
	\item $(2r+1)p$ searchers are placed on the vertices of $p$ paths in ${\cal P}$.
	\item $2(2r+1)$ searchers to clear the guards and their $r-$frames.
	\item $3$ searchers are placed on $x$, $x'$ and $c_j^\ell$ and one more searcher to clear their attached paths. 
	\item $2(2r+1)$ to clear $x_S$ and their $r-$frames for all $S \in {\cal S}$.  
	\end{itemize}
After $\hat{B}^{ml+j}_i$ is cleared, we keep searchers on the exit vertices and $c_j^\ell$, remove other searchers and reuse them to clear $\hat{B}^{ml+j}_{i+1}$. After all the group gadgets in round $b$ are cleared, we slide searchers on the exit vertices of $\hat{B}^{b}_i$ to the entry vertices of $\hat{B}^{b+1}_i$ for all $1 \leq i \leq t$ and start a new round. When $b = m(2rpt+1)$, we need one more searcher to clear the path and the edges incident to $h_2$. In total, we use at most $tp + (5+p)(2r+1) + 4$ searchers which completes the proof of the lemma.     
%
%
%

\subsection{Proof of Lemma~\ref{lm:pw-cds}} \label{App:AppendixD}

Before we give a proof of Lemma~\ref{lm:pw-cds}, we have following observation:
\begin{observation} \label{obs:searcher-patterns}
Given a pattern $P_r(m)$, in which $m$ searchers are placed at $m$ leaves, at most $m+1$ more searchers are needed to clean the pattern $P_r(m)$. 
\end{observation}
Indeed, when $r$ is odd, we only need $2$ more searchers to clean the pattern $P_r(m)$ given $m$ searchers are placed at $m$ leaves.   

We give a mixed search strategy using at most $tp + O((2r+2)^{2p})$ searchers. For a group gadget $B$, we call the set of $p$ vertices $\{a_1\}_p$ of $p$ cores of $B$ the \emph{entry vertices}. We search the graph $G$ in $m((2r+1)tp + 1)$ rounds, each round we clean $t$ group gadgets. Initially, we place $tp$ searchers at entry vertices of $B^1_i$ for all $1\leq i \leq t$. We place $1$ searcher at the root $r_T$ to clean the edges between $r_T$ and entry vertices of $B^1_i$ for all $1\leq i \leq t$ and using $1$ more searchers to clean the path attached to $r_T$. In round $b, 1 \leq b \leq m((2r+1)tp+1)$ such that $b = m\ell+j, 0 \leq \ell < (2r+1)tp +1, 1 \leq j \leq m$, we keep $tp$ searchers on the entry vertices of all group gadgets $B^{ml+j}_i , 1 \leq i \leq t$. We clear the group gadget $B^{ml+j}_i$ by using at most $(2r+2)^{2p} + (2r+2)p + 1$ searchers in which:
	\begin{itemize}[nolistsep, noitemsep]
	\item $(2r+2)p$ searchers are placed and kept on the vertices $\{a_1,a_2,\ldots,a_{2r+2}\}$ of $p$ cores. We use $1$ more searcher to clean the edge between vertices in the set $\{a_1,a_2,\ldots,a_{2r+2}\}$ and $r+2$ more searchers to clean two patterns $P_r(r+1)$ attached to two sets $\{a_1,a_3,\ldots, a_{2r+1}\}$ and $\{a_2,a_4, \ldots, a_{2r+2}\}$. By Observation~\ref{obs:searcher-patterns}, we need $(r+1)(r+2)$ vertices to clean the $r$-paths connecting $r_T$ and $\{a_2,a_4, \ldots, a_{2r+2}\}$ and path-forcing patterns $P_r(r+1)$ attached along these paths. Execept $(2r+2)a$ searchers on the set of vertices $\{a_1,a_2,\ldots,a_{2r+2}\}$, we remove all other searchers after finishing this step.	
	\item $(2r+2)^p((2r+2)^p+1)$ searchers are needed to clean the paths connecting set vertices $\{\bar{x}_S| S \in {\cal S}\}$ and $r_T$ the path-forcing patterns $P_r((2r+2)^p)$ attached along these paths by Observation~\ref{obs:searcher-patterns}. Keep $(2r+2)^p$ searchers on $\{\bar{x}_S| S \in {\cal S}\}$ and use $1$ more searcher to clean $x$ and its attached path. Also by Observation~\ref{obs:searcher-patterns}, $2pr+1$ searcher is needed to clean the patterns $P_r(p2r)$ rooted at $x_S$ for all $S \in {\cal S}$.		
	\item $1$ searcher placed on $c_j^\ell$ and one more searcher to clear its attached paths. 
	\end{itemize}
Note that searchers can be reused after finishing each small step described above. After $B^{ml+j}_i$ is cleared, we keep $p$ searchers at $p$ vertices $\{a_{2r+3}\}_p$ of $p$ cores, which are entry vertices of the next cores and keep $1$ searcher on $c_j^\ell$. We remove other searchers and reuse them to clear $B^{ml+j}_{i+1}$. After all the group gadgets in round $b$ are cleared, searchers are moved to the entry vertices of $B^{b+1}_i$ for all $1 \leq i \leq t$ and start a new round. In total, we use at most $tp + O((2r+2)^{2p})$ searchers which completes the proof of the lemma.  